\documentclass[11pt]{article}

\usepackage[margin = 1in]{geometry}

\usepackage{enumitem}
\usepackage{amsmath,amssymb,amsthm,hyperref,color}
\usepackage{blkarray}
\usepackage{multirow}
\usepackage{url}
\hypersetup{colorlinks=true,citecolor=blue, linkcolor=blue, urlcolor=blue}

\newtheorem{theorem}{Theorem}
\newtheorem{lemma}[theorem]{Lemma}
\newtheorem{claim}[theorem]{Claim}

\newtheorem{observation}[theorem]{Observation}
\newtheorem{corollary}[theorem]{Corollary}
\newtheorem{remark}{Remark}

\newtheorem{definition}{Definition}

\newcommand{\eps}{\epsilon}
\newcommand{\ra}{\rightarrow}

\def\Bu{{\mathfrak{B}}_u}
\def\Bp{{\mathfrak{B}}_o}
\def\Bh{{\mathfrak{B}}_{rc}}

\def\S{\mathcal{S}}
\def\G{\mathcal{G}}

\def\E{\mathbb{E}}
\def\B{B}

\def\eps{\varepsilon}
\def\epsilon{\varepsilon}

\newcommand{\bR}{\mathbf{R}}
\def\m{\ensuremath{\mathbf{m}}}

\newcommand{\fptas}{\mathsf{FPTAS}}
\newcommand{\fpras}{\mathsf{FPRAS}}
\newcommand\Tree{\mathbb{T}}
\newcommand{\TreeD}{\mathbb{T}_{\Delta}}

\let\OldL\L
\def\a{\ensuremath{\mathbf{a}}}

\def\u{\ensuremath{\mathbf{u}}}

\def\x{\ensuremath{\mathbf{x}}}
\def\y{\ensuremath{\mathbf{y}}}
\def\z{\ensuremath{\mathbf{z}}}
\def\A{\ensuremath{\mathbf{A}}}
\def\B{\ensuremath{\mathbf{B}}}
\def\E{\ensuremath{\mathbf{E}}}
\def\H{\ensuremath{\mathbf{H}}}
\def\I{\ensuremath{\mathbf{I}}}
\def\L{\ensuremath{\mathbf{L}}}
\def\M{\ensuremath{\mathbf{M}}}

\def\S{\ensuremath{\mathbf{S}}}
\def\Sb{\ensuremath{\mathbf{S}}}
\def\Tb{\ensuremath{\mathbf{T}}}

\def\X{\ensuremath{\mathbf{x}}}
\def\Y{\ensuremath{\mathbf{y}}}

\newcommand{\integers}{\mathbb{Z}}
\def\T{\ensuremath{\intercal}}
\def\Sc{\ensuremath{\mathcal{S}}}

\def\Det{\ensuremath{\mathrm{Det}}}

\def\Cb{\ensuremath{\mathbf{C}}}
\def\Db{\ensuremath{\mathbf{D}}}

\def\Vb{\ensuremath{\mathbf{V}}}
\def\Wb{\ensuremath{\mathbf{W}}}
\def\Eb{\ensuremath{\mathbf{E}}}
\def\Rb{\ensuremath{\mathbf{R}}}
\def\Nb{\ensuremath{\mathbf{N}}}
\def\Zb{\ensuremath{\mathbf{Z}}}

\def\Yc{\ensuremath{\mathcal{Y}}}

\def\alphab{\ensuremath{\boldsymbol{\alpha}}}
\def\betab{\ensuremath{\boldsymbol{\beta}}}
\def\pib{\ensuremath{\boldsymbol{\pi}}}
\def\gammab{\ensuremath{\boldsymbol{\gamma}}}

\def\Gc{\ensuremath{\mathcal{G}}}
\def\Go{\ensuremath{\overline{G}}}

\newcommand{\norm}[1]{\left\|#1\right\|}

\def\Det{\ensuremath{\mathrm{Det}}}

\title{Ferromagnetic Potts Model: Refined \#BIS-hardness and Related Results\thanks{A preliminary version of this paper
appeared in 
{\em Approximation, Randomization, and Combinatorial Optimization. Algorithms and Techniques (APPROX/RANDOM 2014)}, p. 677--691, 2014.
}}

\author{Andreas Galanis\thanks{University of Oxford,
  Wolfson Building, Parks Road, Oxford, OX1~3QD, UK. \texttt{andreas.galanis@cs.ox.ac.uk}.
The research leading to these results has received funding from the European Research Council under
the European Union's Seventh Framework Programme (FP7/2007-2013) ERC grant agreement no. 334828. The paper
reflects only the authors' views and not the views of the ERC or the European Commission. The European Union is not liable for any use that may be made of the information contained therein.}
\and
 Daniel \v{S}tefankovi\v{c}\thanks{
Department of Computer Science, University of Rochester,
Rochester, NY 14627.  \texttt{stefanko@cs.rochester.edu}.
Research supported in part by NSF grant CCF-1318374.}
 \and Eric Vigoda\thanks{School of Computer Science, Georgia
Institute of Technology, Atlanta, GA 30332.
\texttt{vigoda@cc.gatech.edu}.
Research supported in part by NSF grant CCF-1217458.}
 \and
 Linji Yang\thanks{Facebook, Inc. \texttt{ljyang@gatech.edu}.}
 }

\begin{document}

\maketitle

\begin{abstract}
Recent results establish for the hard-core model (and more generally for 2-spin antiferromagnetic systems)
that the computational complexity of approximating the partition function on graphs of maximum degree $\Delta$
undergoes a phase transition that coincides with the uniqueness/non-uniqueness phase transition on
the infinite $\Delta$-regular tree.  For the ferromagnetic Potts model we investigate
whether analogous hardness results hold.
Goldberg and Jerrum showed that approximating the partition function of the ferromagnetic Potts model
is at least as hard as approximating the number of independent sets in bipartite graphs,
so-called \#BIS-hardness.  We improve this hardness result by establishing it
for bipartite graphs of maximum degree $\Delta$.  To this end, we first present a detailed
picture for the phase diagram for the infinite $\Delta$-regular tree, giving a refined
picture of its first-order phase transition and establishing the critical temperature
for the coexistence of the disordered and ordered phases.
We then prove for all temperatures below this critical temperature (corresponding
to the region where the ordered phase ``dominates'')
that it is \#BIS-hard to approximate the partition function on bipartite
graphs of maximum degree $\Delta$. As a simple corollary of this result, we obtain that it is \#BIS-hard to approximate the number of $k$-colorings on bipartite graphs of maximum degree $\Delta$ whenever $k\leq \Delta/(2\ln \Delta)$.

The \#BIS-hardness result for the ferromagnetic Potts model uses random bipartite regular graphs as a gadget in the reduction.
The analysis of these random graphs
relies on recent results establishing connections between the maxima of the expectation of their
partition function,  attractive fixpoints of the associated tree recursions, and
induced matrix norms.  In this paper we extend these connections
to random regular graphs for all ferromagnetic models.  Using these connections,
we establish the Bethe prediction for every ferromagnetic spin system
on random regular graphs, which says roughly that the expectation of the log of the partition
function $Z$ is the same as the log of the expectation of $Z$.
As a further consequence of our results, we prove for the ferromagnetic Potts model
that the Swendsen-Wang algorithm is torpidly mixing (i.e., exponentially
slow convergence to its stationary distribution) on random $\Delta$-regular graphs
at the critical temperature for sufficiently large~$q$.

\end{abstract}

\section{Background}

\subsection{Spin Systems}

We study the ferromagnetic Potts model and present tools which are useful for
any ferromagnetic spin system on random regular graphs.  Hence we begin with
a general definition of a spin system.

A spin system is defined, for an $n$-vertex graph $G=(V,E)$ and integer $q\geq 2$,
on the space $\Omega$ of configurations $\sigma$ which are assignments $\sigma:V\rightarrow [q]$.
The model is characterized by its energy or Hamiltonian $H(\sigma)$ which is a
function of the spin assignments to the vertices.
In the classical examples of the Ising ($q=2$) and Potts ($q\geq 3$) models
without external field, the Hamiltonian $H(\sigma)$ is the number of monochromatic edges in $\sigma$.
Each configuration has a weight $w(\sigma) = \exp(-\beta H(\sigma))$
for a parameter $\beta$ corresponding to the ``inverse temperature'' which controls the strength of edge interactions.

In our general setup, a \textit{specification} of a $q$-state spin model is defined by a symmetric $q\times q$ interaction matrix $\B=\{B_{ij}\}_{i,j\in[q]}$
with non-negative entries. For a graph $G=(V,E)$, the weight of a
configuration $\sigma:V\rightarrow [q]$ is given by:
$$
w_G(\sigma)=\prod_{\{u,v\}\in E} B_{\sigma(u),\sigma(v)}.
$$
We will occasionally drop the subscript $G$ when the graph under consideration is clear from context. The {\em Gibbs distribution} $\mu=\mu_G$ is defined as $\mu(\sigma)= w(\sigma)/Z$ where
$Z=Z_G(\B)=\sum_\sigma w(\sigma)$ is the {\em partition function}.  We remark here that many of our results also apply to models with
arbitrary external fields since we will work with $\Delta$-regular graphs
and in this case the external field can be incorporated into the interaction matrix.

The Ising ($q=2$) and Potts ($q>2$) models have interaction matrices with diagonal entries $B:=\exp(-\beta)$ and
off-diagonal entries $1$.  The models are called ferromagnetic if $B> 1$ since then neighboring spins prefer to align
and antiferromagnetic if $B< 1$.
The hard-core model  is an example
of a 2-spin antiferromagnetic system, its interaction matrix is defined so that
 $\Omega$ is the set of independent sets of $G$ and, for activity (external field) $\lambda>0$, a configuration $\sigma\in\Omega$ has
weight $w(\sigma) = \lambda^{|\sigma|}$ (with $|\sigma|$ denoting the cardinality of the independent set $\sigma$).

\subsection{Ferromagnetic Models}\label{sec:bits}
In this paper, we will focus on ferromagnetic models, and pay special attention to the ferromagnetic Potts model. We are not aware of a general definition of ferromagnetic and antiferromagnetic models.
We use the following notions which 
 generalize the analogous notions for 2-spin and for the Potts model.
 The ferromagnetic definition captures that neighboring spins prefer to align
 (see Observation \ref{obs:ferro} below).

To avoid degenerate cases\footnote{If $\B$ is reducible, by a suitable permutation of the labels of the spins, $\B$ can be put in a block diagonal form 
  where each of the blocks is either irreducible or zero. Such a model 
 can be studied  by considering the induced sub-models of each block corresponding to irreducible  symmetric matrices, since the partition function for the original model is simply the sum of the partition functions of each of these
 sub-models. If $\B$ is periodic, and since $\B$ is symmetric, its period must be two. Such a model is only interesting on bipartite graphs (otherwise the partition function is zero),  and the focus of our general results are for random $\Delta$-regular graphs which are non-bipartite with high probability.}, we assume throughout this paper that the interaction matrix $\B$ is ergodic, that is,
irreducible and aperiodic. 
Hence, by the Perron-Frobenius theorem
 (since $\B$ has non-negative entries) the eigenvalue of $\B$ with the largest magnitude is positive.

\begin{definition}
\label{def:ferro}
A model is called {\bf ferromagnetic}
if $\B$ is positive definite. Equivalently we have that all of its eigenvalues are positive
and also that
\begin{equation*}
\B = \hat{\B}^{\T} \hat{\B},
\end{equation*}
for some $q\times q$ matrix $\hat{\B}$.
\end{definition}

In contrast to the above notion of a ferromagnetic system, in \cite{GSV:colorings}
a model is called {\bf antiferromagnetic}
if all of the eigenvalues of $\B$ are negative except for the largest (which, as noted above, is positive). Note, when the number of spins is greater than 2, there are models which are neither ferromagnetic nor antiferromagnetic.

The most alluring aspect of this definition is that for ferromagnetic models,  neighboring vertices prefer to have the same spin. To see this, the following more general inequality is proved in  \cite{GSV:colorings}, which is a simple application of the Cauchy-Schwarz inequality.
\begin{observation}\label{obs:ferro}
Let $\z_1,\z_2\in \mathbb{R}^q_{\geq0}$ with $\norm{\z_1}_1=\norm{\z_2}_1=1$. For ferromagnetic $\B$, we have
\[(\z_1^{\T} \B \z_1)(\z_2^{\T} \B \z_2)\geq (\z_1^\T \B \z_2)^2.\]
Equality holds iff $\z_1=\z_2$. For antiferromagnetic $\B$, the inequality is reversed.
\end{observation}
Observe that if we plug in the above inequality the vectors with a single 1 in the positions $i$ and $j$ respectively, we obtain that any two spins $i,j$ induce a ferromagnetic two-spin system.

As observed in \cite{GSV:colorings}, an appealing aspect of defining  ferromagnetism in terms of the signature of the interaction matrix is that the definition remains invariant in the presence of external fields. More precisely, for $\Delta$-regular graphs, any external field can be incorporated into the interaction matrix by a congruence transformation of the matrix $\B$. The modified interaction matrix has the same number of positive, zero and negative eigenvalues as the original (this follows from the Sylvester's law of inertia), 
and hence it remains ferromagnetic/antiferromagnetic.

\subsection{Known Connections to Phase Transitions}
Exact computation of the partition function is \#P-complete, even for
very restricted classes of graphs \cite{Greenhill}.  
Hence we focus on whether there is a fully-polynomial
(randomized or deterministic) approximation scheme, a so-called $\fpras$ or
$\fptas$. 


One of our goals in this paper is to refine our understanding of connections between 
 approximating the partition function on graphs of
maximum degree $\Delta$
with phase transitions on the infinite $\Delta$-regular tree $\TreeD$.
A phase transition of particular interest in the infinite tree $\TreeD$ is the
uniqueness/non-uniqueness threshold.  Roughly speaking, in the
uniqueness phase, if one fixes a so-called ``boundary condition'' which is
a configuration $\sigma_\ell$ (for instance, an independent
set in the hard-core model) on the vertices distance $\ell$
from the root, then
in the Gibbs distribution conditioned on this configuration,
is the root ``unbiased''?  Specifically, for all sequences $(\sigma_\ell)$
of boundary conditions, in the limit $\ell\rightarrow\infty$, does the root have the
same marginal distribution?  If so, there is a unique Gibbs measure on the infinite tree
and hence we say the model is in the uniqueness region.  If there
are sequences of boundary conditions which influence the root in the limit then we
say the model is in the non-uniqueness region.

For 2-spin antiferromagnetic spin systems, it was shown that there is an $\fptas$
for estimating the partition function for graphs of maximum degree $\Delta$ when
the infinite tree $\TreeD$ is in the uniqueness region \cite{LLY}.  On the other side,
unless NP=RP, there is no $\fpras$ for the partition function for $\Delta$-regular
graphs when $\TreeD$ is in the non-uniqueness region \cite{SlySun} (see also \cite{GSV:2spin}).
Recently, an analogous NP-hardness result was shown for
approximating the number of $k$-colorings on triangle-free $\Delta$-regular graphs
for even $k$ when $k<\Delta$ \cite{GSV:colorings}.
In contrast to the above inapproximability results for antiferromagnetic systems,
for the ferromagnetic Ising model with or without external field \cite{JS:ising} and
for 2-spin ferromagnetic spin systems without external field \cite{GJP} there is an $\fpras$ for all graphs.
The situation for ferromagnetic multi-spin models, the ferromagnetic Potts
being the most prominent example, is more intricate.

\#BIS refers to the problem of computing the number of independent sets in bipartite graphs.
A series of results has presented evidence that there is unlikely to be a polynomial-time
algorithm for \#BIS, since a number of unsolved counting problems have been shown to be \#BIS-easy
(for example, see \cite{DGGJ,BDGJM,CDGJLMR}).  
The growing anecdotal evidence for \#BIS-hardness suggests that the problem is intractable,
though weaker than NP-hardness.
More recently, 
it was shown in \cite{2spinBIS} that for antiferromagnetic 2-spin models it is \#BIS-hard to approximate the partition function
on \emph{bipartite} graphs of maximum degree $\Delta$ when the parameters of the model lie in the non-uniqueness
region of the infinite $\Delta$-regular tree~$\TreeD$. Also, for ferromagnetic 2-spin models with external field, \cite{LLZ} shows \#BIS-hardness for some region of the parameter space (note, the known regions of the parameter space where an FPRAS exists, see \cite{GJP,LLZ,GL}, do not yet completely complement the \#BIS-hardness result). 

\subsection{Outline of Results}

Our focus in this paper is on understanding the behavior of ferromagnetic spin systems.
Our main tools are bipartite random regular graphs and random regular graphs. 
Whether we use bipartite or general graphs depends on the context, and we use
whichever yields the strongest results in that context.  For instance, we establish
\#BIS-hardness for the ferromagenetic Potts model in Section \ref{sec:results-Potts-BIS};
to obtain hardness results on the class of bipartite graphs we use bipartite
random regular graphs as the core of the gadget.    In Section \ref{sec:results-SW} we
establish results for the Swendsen-Wang algorithm; such results are more interesting
for general graphs and hence we prove this result for random regular graphs.

In \cite{GSV:colorings} we established
concentration of the partition function for general spin systems on bipartite random regular graphs.  At first glance
the picture for random regular graphs is more complicated than for their bipartite counterparts
since the connection to trees is less clear for general models, however for ferromagnetic models an analogous
connection holds as we will establish in Section \ref{sec:conn2tree}.  
For ferromagnetic systems, we establish concentration on random regular graphs
as detailed in Section \ref{sec:general-results}.  As a consequence we establish
the so-called Bethe prediction for random regular graphs as discussed in
Section \ref{sec:general-results}.

\section{Results for the Potts Model}

\subsection{\#BIS-hardness for the Potts model}\label{sec:results-Potts-BIS}

Goldberg and Jerrum \cite{GJ:potts} showed
that approximating the partition function of the ferromagnetic Potts model
is \#BIS-hard, hence it appears likely that the ferromagnetic Potts model is inapproximable
for general graphs.  We refine this \#BIS-hardness result for the ferromagnetic Potts model.
We prove that approximating the partition function for the
ferromagnetic Potts model on {\em bipartite } graphs of maximum degree $\Delta$
is \#BIS-hard for temperatures below the appropriate phase transition point in the infinite tree $\TreeD$.
The appropriate phase transition in the Potts model is not the uniqueness/non-uniqueness threshold,
but rather it is the ordered/disordered phase transition which occurs at $B=\Bp$ as explained in
the next section.

Formally, we study the following problem.

{\noindent \bf Name. } \textsc{\#BipFerroPotts}($q$,$B$,$\Delta$).

{\noindent \bf Instance. } A bipartite graph $G$ with maximum degree $\Delta$.

{\noindent \bf Output. } The partition function for the $q$-state Potts model on $G$.

We use the notion of approximation-preserving reductions, denoted as $\leq_{\mathrm{AP}}$, formally 
defined in \cite{DGGJ} (roughly, for counting problems $\#\Pi_1$ and $\#\Pi_2$, $\#\Pi_1\leq_{\mathrm{AP}}\#\Pi_2$ implies that the existence of an FPRAS for $\#\Pi_2$ implies the existence of an FPRAS for $\#\Pi_1$).  We can now formally state our main result.

\begin{theorem}
\label{thm:BIS-potts}
For all $q\geq3$, all $\Delta\geq 3$,
for the ferromagnetic $q$-state Potts model,
for any $B>\Bp$,
\[
 \textsc{\#BIS} \leq_{\mathrm{AP}} \textsc{\#BipFerroPotts}(q,B,\Delta),
\]
where $\Bp$ is given by~\eqref{eq:zmurko}.
\end{theorem}

Theorem~\ref{thm:BIS-potts} has a simple, yet interesting, consequence for the problem of approximately counting $k$-colorings on \emph{bipartite} graphs of maximum degree $\Delta$. Recall that for general graphs of maximum degree $\Delta$, approximately counting $k$-colorings is NP-hard whenever $k<\Delta$ (and $k$ is even) due to the result of \cite{GSV:colorings}. Theorem~\ref{thm:BIS-potts} yields \#BIS-hardness for bipartite $k$-colorings whenever $k\leq \Delta/(2\log \Delta)$. Formally, we are interested in the following problem.
 
{\noindent \bf Name. } \textsc{\#BipColorings}($k$,$\Delta$).

{\noindent \bf Instance. } A bipartite graph $G$ with maximum degree $\Delta$.

{\noindent \bf Output. } The number of proper $k$-colorings of $G$.

We use a relatively simple reduction from the ferromagnetic Potts model to bipartite colorings, first observed in \cite{DGGJ}, which works even for bounded-degree graphs. Theorem~\ref{thm:BIS-potts} then yields the following corollary (which is proved in Section~\ref{sec:biscolorings}).

\begin{corollary}
\label{cor:colorings}
For all $k,\Delta\geq3$, whenever $k\leq \Delta/(2\ln \Delta)$, it holds that
\[
 \textsc{\#BIS} \leq_{\mathrm{AP}} \textsc{\#BipColorings}(k,\Delta),
\]
\end{corollary}
It would be interesting to extend Corollary~\ref{cor:colorings} to all $k<\Delta$.

\subsection{Potts Model Phase Diagram}

To understand the critical point $\Bp$ we need to delve into the nature of the phase transition in the
ferromagnetic Potts model on the infinite $\Delta$-regular tree $\TreeD$.  We focus on how the phase
transition manifests on a random $\Delta$-regular graph.

For a configuration $\sigma\in\Omega$, denote the set of
vertices assigned spin $i$ by $\sigma^{-1}(i)$.
Let $\triangle_q$ denote the $(q-1)$-simplex, where recall that:
\[
\triangle_{t}=\{(x_1,x_2,\hdots,x_t)\in \mathbb{R}^t\,|\, \mbox{$\sum^t_{i=1}$}\,x_i=1\mbox{ and } x_i\geq 0\mbox{ for } i=1,\hdots,t\}.
\]
We refer to $\alphab\in\triangle_q$ as a {\em phase}.
For a phase $\alphab$, denote the set of configurations with frequencies of colors given by $\alphab$ as\footnote{Technically we need to define
$\Sigma^{\alphab} =
\left\{\sigma:V\rightarrow[q]\,\big|\,  |\sigma^{-1}(i)\cap
V| = \hat{\alpha}_i\right\}$,
where $\{\hat{\alpha}_i\}$  are $\{\alpha_i n\}$ rounded in a
canonical fashion so that their sum is preserved (for example using ``cascade rounding").}:
\[ \Sigma^{\alphab} = \left\{\sigma:V\rightarrow\{1,\hdots,q\}\,\big|\,  |\sigma^{-1}(i)| = \alpha_i n\mbox{ for } i=1,\hdots,q\right\},
\]
and denote the partition function restricted to these configurations by:
\begin{equation*}
Z^{\alphab}_G=\mbox{$\sum_{\sigma\in
\Sigma^{\alphab}}$}\, w_G(\sigma).
\end{equation*}
Let $\Gc$ denote the uniform distribution over $\Delta$-regular graphs with $n$ vertices (for $\Delta n$ even).
Denote the exponent of the first moment as:
\begin{equation}\label{eq:defofpsi1}
\Psi_1(\alphab)  := \Psi_1^{\B} (\alphab)  :=\lim_{n\rightarrow \infty}\frac{1}{n}\log\E_{\Gc}\big[Z^{\alphab}_G\big].
\end{equation}
We derive the expression for $\Psi_1$ in Section~\ref{sec:defer-prelim}. Those $\alphab$ which are global maxima of $\Psi_1$ we refer to as {\em dominant} phases.
We will see in Section \ref{sec:conn2tree} that, for all ferromagnetic models, roughly speaking, the candidates for dominant phases correspond to stable fixpoints of the so-called tree recursions.

For the ferromagnetic Potts model, there will be two types of phases with particular interest; we refer to these two types as the disordered phase and the ordered phases.
The disordered phase is the uniform vector $\alphab = (1/q,\dots,1/q)$.
The ordered phase
refers to a phase with one color dominating in the following sense:
 one coordinate is equal to $a>1/q$ and the other $q-1$ coordinates are equal to $(1-a)/(q-1)$.
 Due to the symmetry of the Potts model, when the ordered phase dominates, in fact, the $q$
 symmetric ordered phases dominate.
These ordered phases have a specific $a=a(q,B,\Delta)$ which corresponds to a fixpoint of the
tree recursions.  The exact definition of this marginal $a$ is not important at this stage, and hence we
defer its definition to a more detailed discussion which takes place in Section \ref{sec:Potts} (see equation \eqref{eq:marginala}).

One of the difficulties for the Potts model
is that the nature of the uniqueness/non-uniqueness phase
transition on $\TreeD$ is inherently different from that of the Ising model.
The ferromagnetic Ising model undergoes a second-order phase transition on $\TreeD$ which
manifests itself on random $\Delta$-regular graphs in the following manner.
In the uniqueness region the disordered phase dominates,
and in the non-uniqueness region the 2 ordered phases dominate.

In contrast, the ferromagnetic Potts model undergoes a first-order phase transition
at the critical activity $\Bu$.    For $B<\Bu$ there is a
unique Gibbs measure on $\TreeD$.  For $B\geq \Bu$ there are multiple
Gibbs measures on $\TreeD$, however there is a second critical activity $\Bp$
corresponding to the disordered/ordered phase transition:
for $B\leq \Bp$ the disordered phase dominates, and for $B\geq \Bp$ the
ordered phases dominate (and at the critical point $\Bp$ all of these $q+1$ phases dominate).

We present a detailed picture of the
phase diagram for the ferromagnetic Potts model.
Previously, H\"{a}ggstr\"{o}m \cite{Haggstrom} established the uniqueness threshold
$\Bu$ by studying percolation in the random cluster
representation.  In addition,
Dembo et al. \cite{DMSS,DMS} studied the ferromagnetic
Potts model (including the case with an external field) and proved that for $B>\Bu$,
either the disordered or the $q$ ordered phases are dominant, but they did not
establish the precise regions where each phase dominates.
For the simpler case of the complete graph (known as the Curie-Weiss model),
\cite{CET} detailed the phase diagram.

H\"{a}ggstr\"{o}m \cite{Haggstrom} established
that the uniqueness/non-uniqueness threshold for the infinite tree $\TreeD$ occurs at $\Bu$ which is
the unique value of $B$ for which the following polynomial has a
double root in $(0,1)$:
\begin{equation}\label{eq:uniqueness}
(q-1)x^{\Delta} + \left(2-B-q\right)x^{\Delta-1} + Bx - 1.
\end{equation}

 The disordered phase is dominant in the uniqueness region and continues to dominate until
 the following activity (which was considered by Peruggi et al. \cite{Peruggi}):
\begin{equation}\label{eq:zmurko}
\Bp := \frac{q-2}{ (q-1)^{(1-2/\Delta)} - 1}.
\end{equation}

Finally, H\"{a}ggstr\"{o}m \cite{Haggstrom} considers the following activity $\Bh$,
which he conjectures is a (second) threshold for uniqueness of the random-cluster
model, defined as:
\[
\Bh := 1 + \frac{q}{\Delta - 2}.
\]
Note, $\Bu < \Bp < \Bh$.

We prove the following picture for the phase diagram for the ferromagnetic Potts model in Section \ref{sec:Potts}. 
Note, to prove that a function has a local maximum at a point, a standard
approach is to show that its Hessian matrix is negative definite.  We often need this stronger
condition in our proofs, hence we use the following definition.
Those dominant phases $\alphab$ with negative definite Hessian are called {\em Hessian dominant} phases. Note that dominant phases always exist but a dominant phase can fail to be Hessian (when some eigenvalue of the underlying Hessian is equal to zero). In Section~\ref{sec:conn2tree}, we give an alternative formulation of the Hessian condition in terms of the local stability of fixpoints of the tree recursions.

\begin{theorem}
\label{thm:Potts-diagram-detailed}
For the ferromagnetic Potts model the following holds at activity $B$:
\begin{description}
\item{{\bf $B<\Bu$:}} There is a unique infinite-volume Gibbs measure on $\TreeD$.
The disordered phase
is Hessian dominant, and there are no other
    local maxima of $\Psi_1$.
\item{$\Bu < B < \Bh$:}  The local maxima of $\Psi_1$ are
the disordered phase $\u$ and the $q$ ordered phases
(the ordered phases are permutations of each other).
All of these $q+1$ phases are Hessian local maxima.
Moreover:
\begin{description}
\item{$\Bu < B < \Bp$:} The disordered phase is Hessian dominant.
\item{$B=\Bp$:} Both the disordered phase and the ordered phases are Hessian dominant.
\item{$\Bp< B < \Bh$:} The ordered phases are Hessian dominant.
\end{description}
\item{$B \geq \Bh$:}
The $q$ ordered phases (which are permutations of each other) are Hessian dominant.
For $B>\Bh$ there are no other local maxima of $\Psi_1$.
\end{description}
\end{theorem}

\subsection{Swendsen-Wang Algorithm}\label{sec:results-SW}

An algorithm of particular interest for the ferromagnetic Potts model is the Swendsen-Wang algorithm.
The Swendsen-Wang algorithm is an ergodic Markov chain whose stationarity distribution is the Gibbs distribution.
It utilizes the random-cluster representation to overcome potential ``bottlenecks'' for rapid mixing that are
expected to arise in the non-uniqueness region.
As a consequence of the above picture for the phase diagram on the infinite tree $\TreeD$ and
our tools for analyzing random regular graphs, we can prove torpid mixing of the Swendsen-Wang algorithm at the disordered/ordered phase transition point $\Bp$.
(Torpid mixing means that the mixing time is exponentially slow.)

The Swendsen-Wang algorithm utilizes the random cluster representation (see \cite{Grimmett}) of
the Potts model to potentially overcome bottlenecks that obstruct the simpler Glauber dynamics.
It is formally defined as follows.  From a configuration $X_t\in \Omega$:
\begin{enumerate}
\item Let $M$ be the set of monochromatic edges in $X_t$.
\item For each edge $e\in M$, delete it with
probability $1/B$.
Let $M'$ denote the set of monochromatic edges that were not deleted.
\item In the graph $(V,M')$, for each connected component, choose a color uniformly at random from $[q]$
and assign all vertices in that component the chosen color.
Let $X_{t+1}$ denote the resulting spin configuration.
\end{enumerate}

There are few results establishing rapid mixing of the Swendsen-Wang
algorithm beyond what is known for the Glauber dynamics, see \cite{Ullrich}
for recent progress showing rapid mixing on the
2-dimensional lattice.  However,
there are several results establishing torpid mixing of the Swendsen-Wang algorithm
at a critical value
for the
$q$-state ferromagnetic Potts model: on the complete graph  ($q\geq 3$) \cite{GoreJerrum},
on Erd\"{o}s-R\'{e}nyi random graphs ($q\geq 3$) \cite{CooperFrieze},
and on the $d$-dimensional integer lattice $\integers^d$ ($q$ sufficiently large) \cite{BCFKVV,BCT}.

Using our detailed picture of the phase diagram of the ferromagnetic Potts model
and our generic second moment analysis for ferromagnetic models on random regular graphs
which we explain in a moment,
we establish torpid mixing on random $\Delta$-regular graphs at
 the phase coexistence point~$\Bp$.


\begin{theorem}\label{thm:SW}
For all $\Delta\geq 3$ and $q\geq 2\Delta/\log\Delta$,  with probability $1-o(1)$ over the choice of a
random $\Delta$-regular graph, for the ferromagnetic Potts model with $B=\Bp$,
the Swendsen-Wang algorithm has mixing time $\exp(\Omega(n))$.
\end{theorem}

We believe that the lower bound on $q$ in Theorem~\ref{thm:SW} is an artifact of our proof, see Remark~\ref{rem:lowerboundq} in Section~\ref{sec:swendsen-wang} for details.


\section{Results for Ferromagnetic Models}

\subsection{Second Moment and Bethe prediction}\label{sec:general-results}

We analyze the Gibbs distribution on random $\Delta$-regular graphs using second
moment arguments.  The challenging aspect of the second moment is determining the
phase that dominates, as we will describe more precisely momentarily.
In a straightforward analysis of the second moment, this reduces to
an optimization problem over $q^4$ variables for a complicated expression.
Even for $q=2$ tackling this requires significant
effort (see, for example, \cite{MWW} for the hard-core model).

In a recent paper \cite{GSV:colorings}
we analyzed antiferromagnetic systems on
{\em bipartite} random $\Delta$-regular graphs, to use as gadgets for
inapproximability results.
In that work we presented a new approach for simplifying the
analysis of the second moment for antiferromagnetic models using the theory of matrix
norms.
In this paper we extend that approach using the theory of matrix norms to analyze the second moment
for random $\Delta$-regular graphs (non-bipartite) for ferromagnetic systems.
We obtain a short, elegant proof that  the exponential order of the second moment is twice the exponential order of  the first moment.

Denote the leading term of the second moment as
\begin{equation}\label{eq:defofpsi2}
\Psi_2(\alphab) := \Psi_2^{\B}(\alphab)  := \lim_{n\rightarrow \infty}\frac{1}{n}\log\E_{\G}\big[\left(Z^{\alphab}_G\right)^2\big].
\end{equation}
Our main technical result is the analysis of the second moment for ferromagnetic models.
We will relate the maximum of the second moment to the maximum of the first moment.
To analyze the second moment we need to determine the phase $\alphab$ that maximizes $\Psi_2$.
We will first show how to reexpress the critical points of $\Psi_1$ in a form that can be readily
expressed in terms of matrix norms (see Section \ref{sec:reexpress}).
Then, using the Cholesky decomposition of the interaction matrix $\B$ and properties of matrix norms
we will show that the second moment is maximized at a phase which is a tensor product of
the dominant phases of the first moment.   This results in the following theorem, which is proved
in Section \ref{sec:2nd-moment}.

\begin{theorem}\label{thm:copo}
For a ferromagnetic model with interaction matrix $\B$,
\[\max_{\alphab} \Psi_2(\alphab) = 2\max_{\alphab}\Psi_1(\alphab).\]
More specifically,
for dominant $\alphab$, $\Psi_2(\alphab)=2\Psi_1(\alphab)$.
\end{theorem}

Combining Theorem~\ref{thm:copo} with an elaborate variance analysis known as the small subgraph conditioning method allows us to obtain a lower bound on $Z_G^{\alphab}$ which matches its expectation up to a polynomial factor (see Lemma~\ref{lem:smallgraph}). In particular, we verify the so-called \textit{Bethe prediction} (see \cite{DMS,DMSS}) for general ferromagnetic models on random $\Delta$-regular graphs, which is captured in our setting by equation \eqref{eq:betheprediction} in the following theorem (the proof is in Section~\ref{sec:proofprediction}). Our interest in the quantity $\E_{\Gc}[\log Z_G]$ stems from the fact that it gives information about the typical configurations on a random regular graph and hence Theorem~\ref{thm:prediction} gives its value in terms of a much simpler quantity, $\log\E_{\Gc}[ Z_G]$, which can be calculated much more easily, see Section~\ref{sec:defer-prelim}. (The equality in \eqref{eq:betheprediction} is closely related to the cavity method, see \cite{Mezard}.)

\begin{theorem}\label{thm:prediction}
Let $\B$ specify a ferromagnetic model. Then, if there exists a Hessian dominant phase, it holds that
\begin{equation}\label{eq:betheprediction}
\lim_{n\rightarrow\infty}\frac{1}{n}\E_{\Gc}[\log Z_G]=\lim_{n\rightarrow\infty}\frac{1}{n}\log\E_{\Gc}[ Z_G].
\end{equation}
\end{theorem}

Note that for a ferromagnetic model the interaction matrix $\B$ is positive definite and hence the entries on the diagonal are all positive. Thus $Z_G$ is always positive for every graph $G$ (and hence $\log Z_G$ in \eqref{eq:betheprediction} is well-defined).

Theorem~\ref{thm:prediction} holds for all ferromagnetic models at any temperature.  Dembo et al. \cite{DMS} consider general factor models on graph sequences converging locally to trees and verify the Bethe prediction when the underlying tree is in the uniqueness regime. In \cite{DMSS}, the case of the ferromagnetic Potts model (with external field) is considered for graph sequences converging locally to trees and they obtain a general formula for the logarithm of the partition function.

Perhaps the most important conceptual content of Theorem~\ref{thm:prediction} is that it shows that \emph{all} ferromagnetic models, at any temperature, do not exhibit the complex behavior that other spin models, such as colorings or the antiferromagnetic Potts model, exhibit on random (regular) graphs. In particular, when the equality in \eqref{eq:betheprediction} fails, we have the so-called condensation regime, and in that case calculating $\frac{1}{n}\E_{\Gc}[\log Z_G]$ is a far more intricate task (see the recent works \cite{Bapst, SSZ}).

Theorem~\ref{thm:prediction} can be extended to general models (not necessarily ferromagnetic) on random $\Delta$-regular graphs under the stronger assumption that  there is a unique semi-translation invariant Gibbs measure on $\TreeD$. In this setting, one also obtains the analogue of Theorem \ref{thm:copo} and as a consequence concentration for $Z_G^{\alphab}$ for the (unique) dominant phase $\alphab$, which can be used to verify (in complete analogy) the  Bethe prediction, see Section~\ref{sec:bethe-uniqueness} for more details.

\subsection{Connection to Tree Recursions}\label{sec:conn2tree}

As a consequence of Theorem \ref{thm:copo}, to analyze ferromagnetic models on random
regular graphs, one only needs to analyze the first moment.  To simplify the analysis of the
first moment, we establish the following connection to the so-called tree recursions.
An analogous connection was established in \cite{GSV:colorings} for antiferromagnetic models
on random bipartite $\Delta$-regular graphs.

A key concept are the following recursions corresponding to the partition
function on trees, and hence we refer to them as the (depth one) {\em tree recursions}:
\begin{equation}\label{kkrtko}
\widehat{R}_i \propto \Big(\sum_{j=1}^q B_{ij}
R_j\Big)^{\Delta-1}
\end{equation}
The fixpoints of the tree recursions are those $\bR=(R_1,\dots,R_q)$ such that:
$
\widehat{R}_i \propto R_i
\mbox{ for all }i\in [q].  $
We refer to a fixpoint $\bR$ of the tree recursions as {\em Jacobian attractive} if  the Jacobian
at $\bR$ has spectral radius less than $1$.
We prove the following theorem detailing the connections between the tree recursions
and the critical points of the partition function for random regular graphs.

\begin{theorem}\label{thm:connection}
Assume that the model is ferromagnetic. Jacobian attractive fixpoints of the (depth one) tree recursions
are in one-to-one correspondence with the Hessian local maxima of $\Psi_1$.
\end{theorem}

The above connection fails for antiferromagnetic models, e.g., for the antiferromagnetic
Potts model the uniform distribution is a global maximum but it is not a stable fixpoint of the tree recursions for small enough temperature.
(In fact, for antiferromagnetic models every solution of the tree recursions is
a local maximum, see Remark~\ref{rema11}.)

Using the above connection we establish the detailed picture for the dominant phases
of the ferromagnetic Potts model as stated in Theorem \ref{thm:Potts-diagram-detailed}.

\section{Expressions for $\Psi_1$ and $\Psi_2$}
\label{sec:defer-prelim}

In this section, we derive expressions for the first and second moments of $Z_G^{\alphab}$, which will allow us to derive explicit expressions for the functions $\Psi_1(\alphab)$ and $\Psi_2(\alphab)$. Similar expressions have appeared in \cite[Section 2.1]{DMSS} in a slightly different form. Our exposition here is such that it provides a straightforward alignment with the analogous expressions in \cite{GSV:colorings}. The minor differences are due to the model of $\Delta$-regular  random graphs, which in this paper is the pairing model $\Gc(n,\Delta)$. We first specify the model of $\Delta$-regular  random graphs. 

The distribution $\Gc(n,\Delta)$ on $\Delta$-regular multigraphs is generated by the following random process. For $\Delta n$ even, consider the set $[\Delta n]$. Elements of $[\Delta n]$ will be called points. First, a random perfect matching of the $\Delta n$ points is sampled. Then, for $i=1,\hdots,n$, we identify the points $\Delta (i-1)+1,\hdots, \Delta i$ as a single vertex of a graph $G$. The edges of  $G$ are naturally induced by the edges of the random matching and hence every vertex has degree $\Delta$. Note that $G$ may contain parallel edges or self-loops. It is well known that any property which holds asymptotically almost surely  for the pairing model (i.e., with probability $1-o(1)$ as $n\rightarrow\infty$) holds asymptotically almost surely for the uniform distribution on $\Delta$-regular (simple) graphs as well, see for example \cite{JLR}. This is going to be the case for our results.

Recall, $\triangle_t$ denotes the simplex
\begin{equation}
\label{simplex}
\triangle_{t}=\{(x_1,x_2,\hdots,x_t)\in \mathbb{R}^t\,|\, \mbox{$\sum^t_{i=1}$}\,x_i=1\mbox{ and } x_i\geq 0\mbox{ for } i=1,\hdots,t\}.
\end{equation}
Let $G\sim\G(n,\Delta)$ and denote by $V$ the vertex set of $G$. For a configuration $\sigma:V\rightarrow\{1,\hdots,q\}$, we denote the set of
vertices assigned color $i$ by $\sigma^{-1}(i)$. For $\alphab\in \triangle_q$ and $n\alphab\in \mathbb{Z}^q$, let
\[ \Sigma^{\alphab} = \{\sigma:V\rightarrow\{1,\hdots,q\}\,\big|\,  |\sigma^{-1}(i)| = \alpha_i n,\, \mbox{ for } i=1,\hdots,q\},
\]
that is, $\Sigma^{\alphab}$ is the set of configurations $\sigma$ which assign $\alpha_i n$ vertices of $V$ the color $i$, for each $i\in [q]$. We are interested in the total weight
$Z^{\alphab}_G$ of configurations in
$\Sigma^{\alphab}$, namely
\begin{equation*}
Z^{\alphab}_G=\mbox{$\sum_{\sigma\in
\Sigma^{\alphab}}$}\, w_G(\sigma).
\end{equation*}
Note that $Z^{\alphab}_G$ is a r.v., and as indicated earlier, we will look at its moments $\E_\G[Z^{\alphab}_G]$ and
$\E_{\G}[(Z^{\alphab}_G)^2]$.

We begin with the first moment. For $\sigma\in\Sigma^{\alphab}$ and $i,j\in[q]$, 
let $e_{ij}n$ denote the number of edges between vertices
in $\sigma^{-1}(i)$ and $\sigma^{-1}(j)$. Clearly, $e_{ij}=e_{ji}$. It will be notationally convenient to reparameterize the variables $e_{ij}$ as follows:
for $i\neq j$ we set $e_{ij}=\Delta x_{ij}$ and for $i=j$ we
set $e_{ii}=\Delta x_{ii}/2$. For future use, when $G\sim \Gc(n,\Delta)$, we denote by $\x_G(\sigma)$ the random vector $(x_{11},\hdots,x_{qq})$.

The number of perfect matchings between $2n$ vertices is $(2n-1)!!=(2n)!/(n! 2^n)$.  Under
the convention that $0^0\equiv 1$, we then have
\begin{multline}\label{eq:firstmoment}
\E_{\G}[Z^{\alphab}_G]=\binom{n}{\alpha_{1}n,\hdots,\alpha_{q}n}\sum_{\x}\left\{\prod_{i}\binom{\Delta\alpha_{i}n}{\Delta x_{i1}n,\hdots,\Delta x_{iq}n}\right.\\ \left.\times\frac{\big[\prod_{i\neq j}(\Delta x_{ij}n)!\big]^{1/2}\prod_{i}(\Delta x_{ii}n-1)!!}{(\Delta n-1)!!}\prod_{i,j}B^{\Delta x_{ij}n/2}_{ij}\right\},
\end{multline}
where the sum ranges over all the possible values of the random vector $\x_G(\sigma)$. In particular,  $\x=(x_{11},\hdots,x_{qq})$
satisfying:
\begin{equation}\label{eq:constraintfirst}
\begin{gathered}
\begin{aligned}
\mbox{$\sum_{j}$}\,x_{ij}&=\alpha_i& & \big(\forall i\in[q]\big),\\
\end{aligned}\\
x_{ij}=x_{ji}\geq 0\ \ \big(\forall i,j\in[q]\big).
\end{gathered}
\end{equation}
The first line in \eqref{eq:firstmoment} accounts for the
cardinality of $\Sigma^{\alphab}$, while the second line is
$\E_\G[w_G(\sigma)]$ for a fixed
$\sigma\in\Sigma^{\alphab}$, since by symmetry we may focus on any fixed $\sigma$. The first product is the number of ways to choose a partition of the points which is consistent with the values prescribed by $\x$, the fraction is the probability that the random matching connects the points as prescribed, and the last product is the weight of the configuration $\sigma$ conditioned on $\x$.

We next consider the second moment of $Z^{\alphab}_G$. The desired expression may be derived  analogously to \eqref{eq:firstmoment}.
For $(\sigma_1,\sigma_2)\in\Sigma^{\alphab}\times
\Sigma^{\alphab}$, we need to compute the quantity
$\E_\G[w_G(\sigma_1)w_G(\sigma_2)]$. To do this, for $i,k\in[q]$, let $\gamma_{ik}n=
|\sigma_1^{-1}(i)\cap \sigma_2^{-1}(k)|$. The
vector $\gammab$ captures the overlap of the
configurations $\sigma_1$, $\sigma_2$. Denote by $e_{ikjl}n$ the number
of edges matching vertices in $\sigma_1^{-1}(i)\cap \sigma_2^{-1}(k)$ and
$\sigma_1^{-1}(j)\cap \sigma_2^{-1}(l)$. 
We reparameterize as follows: for $(i,k)\neq (j,l)$ we set $e_{ikjl}=\Delta y_{ikjl}$ and for $(i,k)= (j,l)$ we
set $e_{ikjl}=\Delta y_{ikjl}/2$.
\begin{multline}\label{eq:secondmoment}
\E_{\G}[(Z^{\alphab}_G)^2] = \sum_{\gammab}\binom{n}{\gamma_{11}n,\hdots,\gamma_{qq}n}\sum_{\y}\left\{\prod_{i,k}\binom{\Delta\gamma_{ik}n}{\Delta y_{ik11}n,\hdots,\Delta y_{ikqq}n}\right.\\
\left.\times\frac{\big[\prod_{(i, k)\neq (j,l)}(\Delta y_{ikjl}n)!\big]^{1/2}\prod_{i,k}(\Delta y_{ikik}n-1)!!}{(\Delta n-1)!!}\prod_{i,j,k,l}\big(B_{ij}B_{kl}\big)^{\Delta y_{ikjl}n/2}\right\},
\end{multline}
where the sums range over
$\gammab=(\gamma_{11},\hdots,\gamma_{qq})$,
$\y=(y_{1111},\hdots,y_{qqqq})$ satisfying
\begin{equation}\label{eq:constraintsecond}
\begin{gathered}
\begin{aligned}
\mbox{$\sum_{k}$}\,\gamma_{ik}&=\alpha_i   &  &\big(\forall i\in [q]\big), \\
\mbox{$\sum_{i}$}\,\gamma_{ik}&=\alpha_k   &  &\big(\forall k\in [q]\big),\\
\mbox{$\sum_{j,l}$}\,y_{ikjl}&=\gamma_{ik} & &\big(\forall(i,k)\in [q]^2\big)\\
\end{aligned}\\
\begin{aligned}
\gamma_{ik}&\geq 0 & &\big(\forall(i,k)\in[q]^2\big),& y_{ikjl} = y_{jlik}&\geq 0 &
&\big(\forall (i,k,j,l)\in[q]^4\big).
\end{aligned}
\end{gathered}
\end{equation}

The sums in \eqref{eq:firstmoment} and  \eqref{eq:secondmoment}
are typically exponential in $n$. The most critical component of
our arguments is to find the quantitative structure of
configurations which determine the exponential order of the
moments. Formally, we study the limits of
$\frac{1}{n}\log\E_{\G}\big[Z^{\alphab}_G\big]$ and
$\frac{1}{n}\log\E_{\G}\big[(Z^{\alphab}_G)^2\big]$ as
$n\rightarrow \infty$.  These limits can be derived from \eqref{eq:firstmoment} and \eqref{eq:secondmoment} using Stirling's
approximation formula. In particular, we shall use that for a constant $c>0$ with $cn$ even,
we have
\begin{equation}\label{eq:simpleasymp}
\frac{1}{n}\ln\big[(cn)!\big]\sim c\ln n+c\ln c-c\mbox{\ \  and\ \ } \frac{1}{n}\ln\big[(cn-1)!!\big]\sim \frac{c}{2}\ln n+\frac{c}{2}\ln c-\frac{c}{2}.
\end{equation}
Under the usual conventions that $\ln
0\equiv -\infty$ and $0\ln 0\equiv 0$, the above formulas are correct even in the degenerate case $c=0$.

We now derive asymptotics for the first moment $\E_{\G}\big[Z^{\alphab}_G\big]$ in order to obtain the function $\Psi_1(\alphab)$, see equation~\eqref{eq:defofpsi1}. Applying \eqref{eq:simpleasymp} yields:
\begin{align}
\Psi_1(\alphab):=\lim_{n\rightarrow\infty}\frac{1}{n}\log\E_{\G}\big[Z^{\alphab}_G\big]&=\max_{\x}\Upsilon_1(\alphab,\x),\label{eq:limitfirst}\\
\mbox{ where } \ \ \
\Upsilon_1(\alphab,\x)&:=(\Delta-1)f_1(\alphab)+ \Delta g_1(\x),\notag\\
f_1(\alphab)&:=\mbox{$\sum_i$}\, \alpha_i\ln\alpha_i,\notag\\
g_1(\x)&:=\mbox{$\frac{1}{2}\sum_{i,j}$}\,x_{ij}\ln B_{ij}-\mbox{$\frac{1}{2}\sum_{i,j}$}\, x_{ij}\ln x_{ij},\notag
\end{align}
defined on the region \eqref{eq:constraintfirst}.

Completely analogously, for the second moment we obtain:
\begin{align}
\Psi_2(\alphab):=\lim_{n\rightarrow \infty} \frac{1}{n}\log\E_{\G}\big[(Z^{\alphab}_G)^2\big]&=\max_{\gammab}\max_{\y}\Upsilon_{2}(\gammab,\y), \label{eq:limitsecond}\\
\mbox{ where } \ \ \
\Upsilon_{2}(\gammab,\y)&:=(\Delta-1)f_2(\gammab)+\Delta g_2(\y),\notag\\[0.15cm]
f_{2}(\gammab)&:=\mbox{$\sum_{i,k}$}\,\gamma_{ik}\ln\gamma_{ik},\notag\\[0.15cm]
g_2(\y)&:=\mbox{$\frac{1}{2}\sum_{i,k,j,l}$}\, y_{ikj\ell}\ln (B_{ij} B_{kl})-\mbox{$\frac{1}{2}\sum_{i,k,j,l}$}\, y_{ikjl}\ln y_{ikjl},\notag
\end{align}
defined on the region \eqref{eq:constraintsecond}.

\begin{remark}\label{rem:pair-spin}
It is useful to think of the second moment as the first moment of
a paired-spin model with interaction matrix $\B\otimes \B$. Indeed, from \eqref{eq:limitsecond},  we can interpret $B_{ij} B_{kl}$ as the activity between the paired spins $(i,k)$ and $(j,l)$, thus giving the desired  alignment.
\end{remark}

\section{Second Moment Analysis Using Induced Matrix Norms}
\label{sec:second-moment}


\subsection{Critical Points and Matrix Norms}
\label{sec:reexpress}

It will be useful to reformulate function $\Psi_1$ into the following version which will
preserve the critical points, and readily yield a formulation in terms of matrix norms. Let
\begin{equation}\label{mampo}
\Phi_1(\bR) = \frac{\Delta}{2}\ln\Big(\sum_{i=1}^q\sum_{j=1}^q B_{ij} R_i
R_j\Big) - (\Delta-1)\ln\Big(\sum_{i=1}^q
R_i^{\Delta/(\Delta-1)}\Big),
\end{equation}
where  $\bR = (R_1,\dots,R_q)^{\T}\geq 0$, i.e., $\bR$ has non-negative entries.  Let $p:=\Delta/(\Delta-1)$. Note that~\eqref{mampo}
has the following appealing form
\begin{equation}\label{zzzok}
\exp(2\Phi_1(\bR)/\Delta)=\frac{ \bR^{\T} \B \bR}{\|\bR\|_{p}^2},
\end{equation}
where $\|\bR\|_p = (\sum_{i=1}^n R_i^p)^{1/p}$. This will allow us to use
the techniques from the area of matrix norms in our arguments, more specifically, results
on induced matrix norms. The induced matrix norms will be denoted
$\|\cdot\|_{p\rightarrow q'}$:
\begin{equation}\label{hako111}
\| \B \|_{p\rightarrow q'} := \max_{\|\z\|_p = 1} \|\B \z\|_{q'}.
\end{equation}
Since we assume that $\B$ is ferromagnetic we have $\B = \hat{\B}^{\T} \hat{\B}$ and hence we can write
\begin{equation}\label{zzzok2}
\exp(\Phi_1(\bR)/\Delta)=\frac{ \|\hat{\B} \bR\|_2}{\|\bR\|_{p}}.
\end{equation}
The next lemma describes the connection between $\Phi_1$ and $\Psi_1$. We note that
 $\Phi_1$ is not a reparameterization of $\Psi_1$, however they do agree
at the critical points. This is sufficient for our purpose: to understand the maxima
of $\Psi_1$ it is enough to understand the maxima of $\Phi_1$. The maximization
\begin{equation}\label{zzzok3}
\max_{\bR\geq 0} \frac{ \|\hat{\B} \bR\|_2}{\|\bR\|_{p}} = \max_{\bR} \frac{ \|\hat{\B} \bR\|_2}{\|\bR\|_{p}} = \|\hat{\B}\|_{p\ra 2},
\end{equation}
is the induced $p\ra 2$ matrix norm of $\hat{\B}$. The first equality in~\eqref{zzzok3} follows
from the fact that the maximum on the right-hand-side of~\eqref{zzzok} is achieved for non-negative $\bR$ (this follows from the fact that $\B$
has non-negative entries).

\begin{lemma}\label{new:zako1}
There is a one-to-one correspondence between the fixpoints of the tree
recursions and the critical points of $\Phi_1$ (both considered for $R_i\geq 0$ in the projective space, that is, up to scaling
by a constant). The following transformation $\bR\mapsto\alphab$ given by:
\begin{equation}\label{jqwe}
\alpha_i =   R_i^{\Delta/(\Delta-1)} /    \mbox{$\sum_i$} R_i^{\Delta/(\Delta-1)}
\end{equation}
yields a one-to-one-to-one correspondence between the critical points of $\Phi_1$ and the critical points of $\Psi_1$ (in the region defined by
$\alpha_i\geq 0$ and $\sum_i \alpha_i = 1$).
Moreover, for the corresponding critical points $\bR$ and $\alphab$ one has
\begin{equation}\label{pwwww1}
\Phi_1(\bR) = \Psi_1(\alphab).
\end{equation}
Finally, the local maxima of $\Phi_1$ and $\Psi_1$ happen at the critical points (that is, there are no local
maxima on the boundary).
\end{lemma}

We omit the proof of Lemma \ref{new:zako1} since it follows the proof of Theorem 4.1 in \cite[Section 4]{GSV:colorings}. In that paper we consider random $\Delta$-regular bipartite graphs and in analogy to $\Psi_1(\alphab)$ and $\Phi_1(\Rb)$ we define $\Psi_1(\alphab,\betab)$ and $\Phi_1(\Rb,\Cb)$, respectively, where $\alphab,\Rb$ now correspond to the left-side of the bipartition and $\betab,\Cb$ to the right-side of the bipartition. The expressions in our setting (random $\Delta$-regular graphs) are identical to those in \cite{GSV:colorings} (random $\Delta$-regular bipartite graphs) after identifying $\alphab$ with $\betab$ and $\Rb$ with $\Cb$.  In fact, the proof of Theorem 4.1 in \cite{GSV:colorings} works almost verbatim in our case after this identification.

\subsection{Second Moment Analysis}
\label{sec:2nd-moment}

For ferromagnetic models, Lemma~\ref{new:zako1} allows us to reexpress the optimization problem associated with the first moment in terms of matrix norms.

\begin{lemma}\label{leko1}
Let $\B = \hat{\B}^{\T} \hat{\B}$ be the interaction matrix of a ferromagnetic spin system. We
have
\begin{equation*}
\max_{\alphab}\Psi_1(\alphab) = \Delta \ln \|\hat{\B}\|_{\frac{\Delta}{\Delta-1}\ra 2}.
\end{equation*}
\end{lemma}
\begin{proof}
Using Lemma~\ref{new:zako1} and equations \eqref{zzzok2} and \eqref{zzzok3}, we obtain
\begin{equation*}\label{hako222}
\max_{\alphab} \exp(\Psi_1(\alphab)/\Delta) = \max_{\bR} \exp( \Phi_1(\bR)/\Delta)
= \|\hat{\B}\|_{p\ra 2}.
\end{equation*}
\end{proof}
Recall, the definition of $\Psi_2$ (see \eqref{eq:defofpsi2}) corresponding to
the leading term of the second moment.
A key fact is that $\Psi_2$ is given by a constrained first moment calculation
on a ``paired-spin'' model where the interaction matrix in this model
is the tensor product of the original interaction matrix with itself (see Remark \ref{rem:pair-spin}
in Section~\ref{sec:defer-prelim}).
The second moment considers a pair of configurations, say $\sigma$ and $\sigma'$,
which are constrained to have a given phase $\alphab$.  We capture
this constraint using a vector
$\gammab$ corresponding to the overlap between $\sigma$ and $\sigma'$, in particular,
$\gamma_{ij}$ is the number of vertices with spin $i$ in $\sigma$ and
spin $j$ in $\sigma'$.

 Recall, $\Psi^{\B}_1$ indicates the dependence of the function $\Psi_1$ on the interaction matrix $\B$; to simplify the
notation we will drop the exponent if it is $\B$.
 More precisely,
\begin{equation}\label{dupp}
\Psi_2(\alphab) = \max_{\gammab} \Psi_1^{\B\otimes\B}(\gammab),
\end{equation}
where the optimization in~\eqref{dupp} is constrained to $\gammab$ such that
\begin{equation}\label{dupp2}
\mbox{$\sum_i$}\, \gamma_{ik} = \alpha_k \quad\mbox{and}\quad \mbox{$\sum_k$}\, \gamma_{ik} = \alpha_i.
\end{equation}
Ignoring the two constraints in~\eqref{dupp2} can only increase the value of~\eqref{dupp} and hence
\begin{equation}\label{huko3}
\max_{\alphab} \exp(\Psi_2(\alphab)/\Delta) \leq  \max_{\gammab} \exp(\Psi^{\B\otimes \B}_1(\gammab)/\Delta) = \|\hat{\B}\otimes
\hat{\B}\|_{\frac{\Delta}{\Delta-1}\ra 2}^2.
\end{equation}
For induced norms $\|\cdot\|_{p\rightarrow q'}$ with $p\leq q'$ it is known (Proposition 10.1 in \cite{MR0493490}) that
\begin{equation}\label{mm13}
\|\hat{\B}\otimes\hat{\B}\|_{p\rightarrow q'} = \|\hat{\B}\|_{p\rightarrow q'} \|\hat{\B}\|_{p\rightarrow q'}.
\end{equation}

Now we are ready to prove Theorem~\ref{thm:copo}.

\begin{proof}[Proof of Theorem~\ref{thm:copo}]
Combining Lemma~\ref{leko1} and equations~\eqref{huko3},\eqref{mm13} we obtain:
\[
\exp(\Psi_2(\alphab)/\Delta)=\max_{\gammab} \exp(\Psi_1^{\B\otimes \B}(\gammab)/\Delta) \leq
\|\hat{\B}\|_{\frac{\Delta}{\Delta-1}\ra 2}^2 =
\max_{\alphab} \exp(2\Psi_1(\alphab)/\Delta).
 \]
This proves that if $\alphab$ maximizes $\Psi_1$, we have $\Psi_2(\alphab)\leq 2\Psi_1(\alphab)$. The reverse inequality is trivial,  yielding Theorem \ref{thm:copo}.
\end{proof}


\begin{remark}\label{rem:anti-fails}
We will illustrate the necessity of the ferromagnetism assumption in Theorem~\ref{thm:copo} by
giving an example of an antiferromagnetic model for which the second moment method fails (i.e., the second moment is larger than the square of the first moment by an exponential factor and therefore does not provide any useful concentration). Consider
proper $3$-colorings of random $10$-regular graphs. As the size of the graph goes to infinity
the probability of it being $3$ colorable goes to zero. The intuitive effect of this is that
to achieve a large value in the ``paired-spin'' model it is better to correlate the coordinates
to agree. In terms of $\Psi_1$ and $\Psi_2$ we have that the maximum in the first moment is achieved for
$\alpha_1=\alpha_2=\alpha_3=1/3$ with $\Psi_1=5\ln 2 - 4\ln 3 < 0$.  To obtain a lower bound on the maximum in the second moment we take $\gamma_{11}=\gamma_{22}=\gamma_{33}=1/3$, which yields
$\Psi_2 = \Psi_1 > 2\Psi_1$. The argument actually applies whenever $\Psi_1<0$ (for models
whose interaction matrices have $0$'s and $1$'s). By continuity (taking small $B$ in the antiferromagnetic
Potts model) one can obtain an example of a model without hard constraints for which the second moment fails.
\end{remark}

\section{Connections}
\label{sec:connection}

In this section we prove Theorem \ref{thm:connection} which describes the connection between the stable fixpoints of the tree recursions and the local maxima of $\Psi_1$. Theorem \ref{thm:connection} will then be a key tool in our proof of Theorem~\ref{thm:Potts-diagram-detailed}. The technical core of the technique relies on the arguments in \cite{GSV:colorings}, where an analogous  connection has been established for random bipartite regular graphs. The arguments here are a minor modification of this approach, suitably modified to account for random regular graphs.

Our starting point is the one-to-one correspondence between fixpoints of the tree recursions and the critical points of $\Psi_1$
(see,~\cite{MWW}, and also \cite{GSV:colorings}). We show, roughly, that the stability of a fixpoint is equivalent to the local maximality of the
corresponding critical point. This will be done by relating the Jacobian of the tree recursions at a fixpoint with the Hessian of
$\Psi_1$ at the corresponding critical point. More precisely, we show that the Jacobian has spectral radius less than $1$ (a
sufficient condition for stability) if and only in the Hessian is negative definite (a sufficient condition for local
maximality). Both constraints on the matrices are independent of the choice of local coordinates (that is, they are invariant
under similarity transformations), however to make the connection between the Jacobian and the Hessian apparent we will have to
choose the local coordinates very carefully. A further technical complication is that the tree recursions are in the projective
space and that the optimization of $\Psi_1$ is constrained.

We give a high level overview of the Jacobian; the proofs for the $\Delta$-regular case follow the same reasoning as for the bipartite $\Delta$-regular case, see \cite[Section 4.2.2]{GSV:colorings}, after simply changing $C_j$'s to $R_j$'s and $\beta_j$'s to $\alpha_j$'s. Assume
that $R_1,\dots,R_q$ is a fixpoint of the tree recursions. Now we consider an infinitesimal perturbation of the fixpoint
$R_1+\eps R'_1,\dots,R_q+\eps R'_q$ and see how it is mapped by the tree recursions. Let $\alpha_i := \sum_{j} B_{ij} R_i R_j$.
The right parametrization (choice of local coordinates) is to take $R_i' = r_i R_i/\sqrt{\alpha_i}$, where $r_1,\dots,r_q$
determines the perturbation. Note that $R_i/\sqrt{\alpha_i}$ depends on the fixpoint.  The tree recursions map (in the projective
space) the perturbation as follows:
\begin{equation}\label{eq:perturb}
\left( R_1 + \eps r_1 \frac{R_1}{\sqrt{\alpha_1}},\dots, R_q + \eps r_q \frac{R_q}{\sqrt{\alpha_q}}\right)\mapsto \left( R_1 +
\eps \hat{r}_1 \frac{R_1}{\sqrt{\alpha_1}},\dots, R_q + \eps \hat{r}_1 \frac{R_q}{\sqrt{\alpha_q}}\right)+O(\eps^2),
\end{equation}
where $\hat{r}_i$'s are given by the following linear transformation
\begin{equation}\label{qa3}
\hat{r}_i = (\Delta-1)\sum_{j=1}^q \frac{B_{ij}R_iR_j}{\sqrt{\alpha_i\alpha_j}} r_j,
\end{equation}
and where the $r_i$'s are required to satisfy
\begin{equation}\label{pako2}
\sum_{i=1}^q\sqrt{\alpha_i} r_i = 0.
\end{equation}
The condition~\eqref{pako2} is invariant under the map~\eqref{qa3} and corresponds to choosing the representative of
$R_1,\dots,R_q$ with $\sum_i\sum_j B_{ij} R_i R_j=1$.

Next we give a high level description of the Hessian; again, this is almost identical to the one in \cite[Section 4.2.1]{GSV:colorings} after identifying $C_j$'s with $R_j$'s and $\beta_j$'s with $\alpha_j$'s. Recall that $\Psi_1$ is a function of $\alpha_1,\dots,\alpha_q$. There is an alternative parameterization of $\Psi_1$: instead of  $\alpha_1,\dots,\alpha_q$ (restricted to $\sum\alpha_i = 1$) we use $R_1,\dots,R_q$ (restricted to $\sum_i\sum_j B_{ij} R_i R_j =1$) and use the following
\begin{equation}\label{huako}
\alpha_i = \sum_j B_{ij} R_i R_j\quad\mbox{for all}\ i\in [q].
\end{equation}
Every $\alphab$ can be achieved using parameterization by ${\mathbf R}$.   Let
$\alpha_1,\dots,\alpha_q$ be a critical point of $\Psi_1$ and let $R_1,\dots,R_q$ satisfy~\eqref{huako}. We are going to evaluate
$\Psi_1$ in a small neighborhood around $\alpha_1,\dots,\alpha_q$. It is equivalent (and easier to understand) to perturb the
$R_1,\dots,R_q$ to $R_1+\eps R'_1,\dots,R_q+\eps R'_q$ and evaluate at the point given by~\eqref{huako}. Again, the correct
parameterization is to take $R_i' = r_i R_i/\sqrt{\alpha_i}$. This yields the following expression for the value of $\Psi_1$ at the perturbed point
\begin{multline}\label{eq:psi1hessian}
\Psi_1(\alpha_1,\dots,\alpha_q) \\+ \eps^2 \sum_{i=1}^q \left(r_i + \sum_{j=1}^q \frac{B_{ij}R_i R_j}{\sqrt{\alpha_i\alpha_j}}
r_j\right)\left( \sum_{j=1}^q (\Delta-1)\frac{B_{ij}R_i R_j}{\sqrt{\alpha_i\alpha_j}} r_j - r_i\right) + O(\eps^3).
\end{multline}
Note that there is no linear term, since we are at a critical point. Recall that the $\alpha_i$ have to satisfy $\sum_i \alpha_i = 1$ which corresponds to the restriction~\eqref{pako2}.

Now we are ready to prove Theorem~\ref{thm:connection}. Let $L$ be a linear map such that the Jacobian of the map $r\mapsto \hat{r}$ represented by~\eqref{qa3} is $(\Delta-1)L$. The Hessian of $\Psi_1$ is then $(I + L)( (\Delta-1) L - I)$.
Finally let $S$ be the linear subspace defined by~\eqref{pako2}.

\begin{proof}[Proof of Theorem \ref{thm:connection}]
We will use the correspondence between fixpoints of the tree recursions and critical points of $\Psi_1$ given by Lemma~\ref{new:zako1}. The constraint for the fixpoint to be Jacobian attractive is that $(\Delta-1)L$ on $S$ has spectral radius less than $1$, see equation \eqref{eq:perturb}. The constraint for the critical point to be Hessian maximum is that the eigenvalues of $(I + L)( (\Delta-1) L - I)$ on $S$ are
negative, see  equation \eqref{eq:psi1hessian}.

Note that $L$ is symmetric and if $B$ is positive semidefinite then
$L$ is positive semidefinite (since $L$ is congruent to $B$; $L$ is obtained by multiplying $B$
by a diagonal matrix on the left and on the right). Hence $L$ has non-negative real spectrum.
Note that $S$ is invariant under $L$ and hence the spectrum of $L$ on $S$ is a subset of the
spectrum of $L$ (it is still non-negative real; the restriction wiped
out the eigenvalue $1$).

The constraint for the fixpoint to be Jacobian attractive, in
terms of eigenvalues, is: for each eigenvalue $x$ of $L$ on $S$
\begin{equation}\label{ooo1}
(\Delta-1) |x| < 1.
\end{equation}
The constraint for the critical point to be Hessian maximum, in
terms of eigenvalues, is: for each eigenvalue $x$ of $L$ on $S$
\begin{equation}\label{ooo2}
(1+x)((\Delta-1)x-1) < 0.
\end{equation}
Note that conditions~\eqref{ooo1} and~\eqref{ooo2} are equivalent (since
$x\geq 0$).
\end{proof}

\begin{remark}\label{rema11}
For antiferromagnetic models every critical point of $\Psi_1$ is a local maximum. Indeed, we need only to prove that equation~\eqref{ooo2} is satisfied for every critical point. The matrix $L$ has non-negative entries hence
$1$ is the largest eigenvalue and all the other eigenvalues have magnitude less than $1$ (since $\B$
is ergodic). Moreover the matrix $L$ has the same signature as $\B$ (since they are congruent) and hence
the eigenvalues other than $1$ are negative. These $2$ facts imply~\eqref{ooo2}.
\end{remark}

\begin{remark}\label{rema33}
Note that one direction of the implication in Theorem \ref{thm:connection}, namely,
that a Jacobian attractive fixpoint is Hessian local maximum, holds for every model
(without the ferromagnetism assumption), since \eqref{ooo1} always implies \eqref{ooo2}.
However, for the reverse implication, the ferromagnetic assumption is essential.
For example, in an antiferromagnetic model, by Remark~\ref{rema11}, every critical point is a local maximum. For the antiferromagnetic Potts model, the only critical point is the uniform vector and hence it is always a local maximum for every value of $B$. On the other hand, it is straightforward to check that for the antiferromagnetic Potts model the uniform fixpoint is Jacobian unstable when $B<\frac{\Delta-q}{\Delta}$.
\end{remark}

\section{Bethe Prediction for Ferromagnetic Models}
\label{sec:Bethe-prediction}

\subsection{Small Subgraph Conditioning Method}
\label{sec:small-graph-overview}


By Theorem~\ref{thm:copo}, we have that for the random variable $Z_G^{\alphab}$, when $\alphab$ is a global maximizer of $\Psi_1$, the exponential order of its second moment is twice the exponential order of  its first moment. This is not sufficient however to obtain high probability results, since it turns out that, in the limit $n\rightarrow \infty$, the ratio of the second moment to the square of the first moment  converges to a constant greater than 1. Hence, the second moment method fails to give statements that hold with high probability over a uniform random $\Delta$-regular graph. More specifically, to obtain our results we need sharp lower bounds on the partition function which hold for almost all $\Delta$-regular graphs. In the setting we described, the second moment method only implies the existence of a graph which satisfies the desired bounds and even there in a not sufficiently strong form.

For  random $\Delta$-regular graph ensembles, the standard way to circumvent this failure is to use the small subgraph conditioning method of Robinson and Wormald \cite{RW}. While the method is quite technical, its application is relatively streamlined when employed in the right  framework. The method was first used for the analysis of spin systems in the work of \cite{MWW} for the hard-core model and subsequently in \cite{Sly10}, \cite{GSV:2spin}. In \cite{GSV:colorings}, we extended the approach to $q$-spin models for all $q\geq 2$, where the major technical obstacle was the computation of certain determinants which arise in the computation of the moments' asymptotics. While the arguments there  are for random \textit{bipartite} $\Delta$-regular graphs, the approach extends in a straightforward manner
to random $\Delta$-regular graphs.

We defer the details of the application of the method in the present setting to Section \ref{sec:small-graph}.
We state here the following lemma which is the final outcome of the method.

\begin{lemma}\label{lem:smallgraph}
For every ferromagnetic model $\B$, if $\alphab$ is a Hessian dominant phase (cf. Section~\ref{sec:conn2tree})
with probability $1-o(1)$ over the choice of the graph $G\sim\Gc(n,\Delta)$, it holds that $Z_G^{\alphab}\geq\frac{1}{n}\E_{\Gc}\big[Z_G^{\alphab}\big]$.
\end{lemma}

\subsection{Proof of Theorem \ref{thm:prediction}}\label{sec:proofprediction}

Using Lemma~\ref{lem:smallgraph}, the proof of Theorem~\ref{thm:prediction} is straightforward.

\begin{proof}[Proof of Theorem~\ref{thm:prediction}]
Let $\alphab$ be a Hessian dominant phase, whose existence is guaranteed by the assumptions. By Lemma~\ref{lem:smallgraph}, with probability $1-o(1)$ over the choice of the graph, we have $Z^{\alphab}_G\geq \frac{1}{n}\E_{\Gc}\big[Z_G^{\alphab}\big]$, which implies $\frac{1}{n}\log Z_G\geq \Psi_1(\alphab)+o(1)$.

Moreover, since the model is ferromagnetic,  for $\Delta$-regular graphs $G$ with $n$ vertices, $\frac{1}{n}\log Z_G\geq C$ for some constant $C>-\infty$ (explicitly, one can take $C:=\frac{\Delta}{2}\log \max_{i\in[q]} {B_{ii}}$, see the remarks after Theorem~\ref{thm:prediction}). We thus obtain
\[\liminf_{n\rightarrow\infty}\frac{1}{n}\E_\Gc[\log Z_G]\geq  \liminf_{n\rightarrow\infty}\big[(1-o(1))\Psi_1(\alphab)+o(1) C\big]=\Psi_1(\alphab).\]
By Jensen's inequality, we also have
\[\limsup_{n\rightarrow\infty}\frac{1}{n}\E_\Gc[\log Z_G]\leq \lim_{n\rightarrow\infty}\frac{1}{n}\log\E_{\Gc}[ Z_G].\]
All that remains to show is that $\frac{1}{n}\log\E_\Gc[ Z_G]=\Psi_1(\alphab)+o(1)$. This is straightforward; if we decompose $Z_G$ as $Z_G=\sum_{\alphab'}Z_G^{\alphab'}$, we obtain $\exp(o(n))   \E_\Gc[ Z_G^{\alphab}]\geq \E_\Gc[ Z_G]\geq  \E_\Gc[ Z_G^{\alphab}]$. Note the factor $\exp(o(n))$, which is there to account for dominant  phases which are not Hessian.

This concludes the proof.
\end{proof}

\section{Phase Diagram for the Ferromagnetic Potts model}
\label{sec:Potts}

In this section we prove Theorem \ref{thm:Potts-diagram-detailed} detailing the phase diagram
for the Potts model.  

To prove Theorem \ref{thm:Potts-diagram-detailed},
we will use Theorem~\ref{thm:connection} and the results of Section~\ref{sec:connection}. We briefly overview the approach. In order to determine the local maxima for the Potts model (or, more generally, for any ferromagnetic model),
we need to compute the spectral radius of the map $L:(r_1,\hdots,r_q)\mapsto(\widehat{r}_1,\hdots,\widehat{r}_q)$,
given by\footnote{Note that, relative to \eqref{qa3}, there is a factor of $(\Delta-1)$ ``missing" in the rhs of \eqref{lmp}. This factor will be accounted shortly later by demanding that the eigenvalues of the map $L$ in \eqref{lmp} are less (in absolute value) than $1/(\Delta-1)$ (instead of 1).}
\begin{equation}
\widehat{r}_i=\sum_{j=1}^{q}\frac{B_{ij}R_iR_j}{\sqrt{\alpha_i\alpha_j}}r_j\label{lmp}
\end{equation}
in the subspace
\begin{equation}
\sum^{q}_{i=1}\sqrt{\alpha_i} r_i=0,\label{eq:sub}
\end{equation}
where the $R_i$'s specify a fixpoint of the tree recursions \eqref{kkrtko} and the $\alpha_i$'s are given by
\[\alpha_i=R_i\sum_{j=1}^{q}B_{ij}R_j\mbox{ for } i=1,\hdots, q.\]
Our goal is to determine the local maxima by verifying when the spectral radius of this map (in the
subspace \eqref{eq:sub}) is less than $1/(\Delta-1)$. Let
$$\M=\left\{\frac{B_{ij}R_iR_j}{\sqrt{\alpha_i\alpha_j}}\right\}_{i,j=1}^q$$
be the matrix of the linear map $L$. Note that $\M$ is symmetric and has an eigenvalue equal to $1$ with
eigenvector $e=\left[\sqrt{\alpha_1},\hdots,\sqrt{\alpha_q}\right]^{\T}$. It follows that the eigenvalues of $L$ in the subspace \eqref{eq:sub} are precisely the eigenvalues different from 1 of the matrix $\M$.

To proceed, we need to restrict our attention to the ferromagnetic Potts model. First we argue that the fixpoints of the tree recursions \eqref{kkrtko} in the case of the ferromagnetic Potts model are simple---they are supported on only two values.
\begin{lemma}\label{hufff1}
Let $(R_1,\dots,R_q)$ be a fixpoint of the tree recursions \eqref{kkrtko} of the ferromagnetic Potts model. Then the $R_i$'s have at most two
distinct values.
\end{lemma}

\begin{proof} W.l.o.g. we may assume that the implicit constant in \eqref{kkrtko} is 1. Let $r_i = R_i^{1/d}$ and $r=\sum_{i=1}^q r_i^d$, where $d:=\Delta-1$. We have
$$
r_i = r + (B-1) r_i^d.
$$
The polynomial $f(x) = (B-1)x^d - x + r$ has at most $2$ positive roots (counted with their multiplicities; by the Descartes'
rule of signs) and hence there are at most $2$ different values of the $r_i$'s. \end{proof}

\begin{lemma}\label{lem:fixpoints}
The fixpoints of the tree recursions, assuming $R_1\geq R_2\geq\hdots R_q$, satisfy $R_1=R_2=\hdots=R_t$ and $R_{t+1}=\hdots R_q$ for some $1\leq t\leq q$. It follows that $\alpha_1=\alpha_2=\hdots=\alpha_t$ and $\alpha_{t+1}=\hdots=\alpha_q$.
\end{lemma}
\begin{proof}
This follows from Lemma~\ref{hufff1}.
\end{proof}
\begin{remark}\label{rem:majorityordered}
Two settings for $t$ in the setting of Lemma~\ref{lem:fixpoints} will be of particular interest, namely $t=1$ and $t=q$. We shall refer to the latter as the uniform fixpoint, and this corresponds to the disordered phase. We shall refer to fixpoints with $t=1$ as the  ``majority" fixpoints. This class includes either one or two (depending on the value of $B$, c.f. Lemma~\ref{lem:numbermajority}) fixpoints where color 1 dominates and the remaining appear with equal probability. The ordered phases correspond to the majority fixpoint for which the ratio $R_1/R_q$ is maximum (cf. the upcoming Lemma~\ref{lem:exactone}).
\end{remark} 

\begin{remark}
We note here that fixpoints with $t\neq q$ exist iff $B\geq \Bu$. In Lemma~\ref{lem:numbermajority}, we only show one side of this equivalence. In particular, we show that $B\geq \Bu$ implies the existence of the majority fixpoints, and this turns out to be the only ``existential" fact that is needed for the proof of Theorem~\ref{thm:Potts-diagram-detailed}. More precisely, in Lemmas~\ref{lem:eigenvaluesMM} and~\ref{lem:fixstablemaj}, we show that fixpoints with $t\neq 1,q$ are not attractive which in turn implies that the corresponding phases are not local maxima of $\Psi_1$ and, thus, are not dominant as well. 
\end{remark}

Lemma~\ref{lem:fixpoints} implies that in the case of the ferromagnetic Potts model, $\M$ has a very simple structure. The following simple lemma describes the eigenvalues of $\M$.
\begin{lemma}\label{lem:eigenvaluesMM}
In the setting of Lemma~\ref{lem:fixpoints}, $\M$ has the following eigenvalues for $1\leq t<q$:
\begin{itemize}
\item $1$ with multiplicity $1$,
\item  $(B-1)R^2_1/\alpha_1$ with multiplicity $t-1$ (assuming $t>1$),
\item $(B-1)R^2_q/\alpha_q$ with multiplicity $q-t-1$ (assuming $t<q-1$), and
\item  $(B+t-1)R^2_1/\alpha_1+(B+q-t-1) R^2_q/\alpha_q-1$ with multiplicity $1$.
\end{itemize}
For $t=q$ the eigenvalues of $\M$ are
\begin{itemize}
\item $1$ with multiplicity $1$,
\item  $(B-1)R^2_1/\alpha_1$ with multiplicity $q-1$.
\end{itemize}
\end{lemma}

\begin{proof}
We already described the eigenvector for eigenvalue $1$. For every $i$ such that $2\leq i\leq t$, a vector with $1$ at position $1$ and $-1$ at
a position $i$ (and zeros elsewhere) yields eigenvalue $(B-1)R^2_1/\alpha_1$.
Similarly, for every $i$ such that $t+1\leq i<q$, a vector with $1$ at position $q$ and $-1$ at  position $i$ (and zeros elsewhere) yields eigenvalue $(B-1)R^2_q/\alpha_q$. Note that in the case $t=q$ this
accounts for all the eigenvalues. In the case $t<q$ we deduce the remaining eigenvalue
by considering the trace of $\M$:
\[t \frac{ B R_{1}^2}{\alpha_{1}} + (q-t)  \frac{B R_{q}^2}{\alpha_{q}} - (t-1)\frac{(B-1)R_1^2}{\alpha_1} - (q-t-1)\frac{(B-1)R_q^2}{\alpha_q} - 1.\]
\end{proof}

\begin{lemma}\label{lem:uniform}
The uniform fixpoint is Jacobian attractive if $(\Delta-2)(B-1)<q$. The uniform fixpoint is not attractive
if $(\Delta-2)(B-1)>q$.
\end{lemma}
\begin{proof}
The uniform fixpoint of the tree recursions corresponds to $R_1=\dots=R_q$ and hence $\alpha_1=\dots=\alpha_q=(B+q-1) R_1^2$. By Lemma~\ref{lem:eigenvaluesMM}, the only relevant eigenvalue  is $(B-1)/(B+q-1)$
(with multiplicity $q-1$), which we compare with $1/(\Delta-1)$ to obtain the lemma.
\end{proof}

Lemma \ref{lem:uniform} allows us to restrict our focus on $q-1\geq t\geq 1$. In this setting, the tree recursions, with $x:=y^d:=\frac{R_1}{R_q}$ and $d:=\Delta-1$, yield:
\begin{equation}
x=\left(\frac{(B+t-1)x+(q-t)}{tx+(B+q-t-1)}\right)^d \mbox{ or } B-1=\frac{(y-1) \left(t y^d+q-t\right)}{y^d-y}.\label{eq:bprecursion}
\end{equation}
The following lemma implies that all fixpoints with $q-1\geq t\geq 2$ are unstable in the whole non-uniqueness regime, since the respective matrices $\M$ have an eigenvalue greater than $1/(\Delta-1)$.
\begin{lemma}\label{lem:fixstablemaj}
When $\frac{R_1}{R_q}>1$, it holds that $(B-1)\frac{R^2_1}{\alpha_1}> \frac{1}{\Delta-1}$.
\end{lemma}
\begin{proof}
The desired inequality is equivalent to
\[(\Delta-1)(B-1)R_1> (B+t-1)R_1+(q-t)R_q,\]
which after simple manipulations reduces into
\[\big((\Delta-2)(B-1)-t\big)\frac{R_1}{R_q}> q-t.\]
Substituting $\frac{R_1}{R_q}=y^d$ and $B-1$ from equation \eqref{eq:bprecursion}, the inequality becomes
\[\left(\frac{(d-1)(y-1) \left(t y^d+q-t\right)}{y^d-y}-t\right)y^d-(q-t)>0.\]
Doing the necessary simplifications, we obtain the following equivalent inequality
\[\frac{\left((d-1) y^{1+d}-d y^d+y\right) \left(q+t \left(y^d-1\right)\right)}{y^d-y} > 0.\]
Since $y>1$, the only non-trivial factor to prove positivity is $p(y):=(d-1) y^{d+1}-d y^d+y$. By the  Descartes' rule of signs, $p(y)$ can have at most two positive roots. It holds that $p(1)=p'(1)=0$, so that $p(y)$ is always positive for $y>1$.
\end{proof}

In light of Lemmas~\ref{lem:eigenvaluesMM} and~\ref{lem:fixstablemaj}, it remains to classify fixpoints with $t=1$, i.e., the majority fixpoints. The following lemma gives the number of the  majority fixpoints in the regimes of interest.
\begin{lemma}\label{lem:numbermajority}
When $\Bu< B< \Bh$, there are exactly two distinct majority fixpoints. When $B\geq \Bh$, there is exactly one majority fixpoint.
\end{lemma}
\begin{remark}
At $B=\Bu$, it follows from the proof of Lemma~\ref{lem:numbermajority} that there is a  unique majority fixpoint which ``bifurcates" in the regime $\Bu< B< \Bh$. 
\end{remark}
\begin{proof}[Proof of Lemma~\ref{lem:numbermajority}]
We need to look at \eqref{eq:bprecursion} for $t=1$ and check how many values of $y>1$ satisfy the equation in the two regimes $\Bu< B< \Bh$ and $B\geq \Bh$. For $t=1$, equation \eqref{eq:bprecursion} reads as
\begin{equation}\label{eq:t1reads}
B-1=f(y):=\frac{(y-1)(y^d+q-1)}{y^d-y},\mbox{ so that } f'(y)=\frac{p(y)}{(y^d-y)^2},
\end{equation}
where $p(y)$ is the polynomial
\begin{equation}\label{eq:py}
p(y):=y^{2 d}-d y^{d+1}-(d-1) (q-2) y^{d}+d (q-1) y^{d-1}-(q-1).
\end{equation}
Employing the Descartes' rule of signs we see that $p(y)$ has one or three positive roots counted by multiplicities. It is easy to check that $p(1)=p'(1)=0$, so that $p$ has in fact 3 positive roots counted by multiplicities (since $1$ is a double root), let $\rho$ denote the other positive root. We next prove that $\rho>1$ so that $p(y)\leq 0$ if $1\leq y\leq \rho$ and $p(y)\geq 0$ if $y\geq \rho$.  It follows that for positive $y$ we have $p(y)>0$ iff $y>\rho$.

To prove that $\rho>1$, for the sake of contradiction assume that $0< \rho\leq 1$.  If $\rho=1$, then $1$ is a root with multiplicity 3 of the polynomial $p(y)$ and hence $p''(1)=0$. By straightforward calculations we see that $p''(1)=(q-2)(d-d^2)$ which is clearly non-zero for $q\geq 3$ and $d\geq 2$. Thus, we may assume that $0<\rho<1$. Since $p(1)=p(\rho)=0$, by Rolle's theorem there is a root $\rho'\in(\rho,1)$ of the polynomial $p'(y)=dy^{d-2}g(y)$ where
\[g(y):= 2y^{d+1}-(d+1)y^2-(d-1)(q-2)y+(d-1)(q-1).\] Since $g(1)=g(\rho')=0$, by the same token there is a root $\rho''\in (\rho',1)$ of \[g'(y)=2(d+1)y^d-2(d+1)y-(d-1)(q-2).\]
We thus obtain the desired contradiction since, for $q\geq 3$ and $d\geq 2$,  $g'(y)< 0$ for all $y\in [0,1]$.

From the above, it follows that $f(\rho)=\min_{y\geq 1}\{f(y)\}$. Observe also that $f(y)\rightarrow \infty$ as $y\rightarrow\infty$, while $f(y)\rightarrow \frac{q}{d-1}$ as $y\downarrow 1$. Thus, when $y\downarrow 1$, we have that $B\uparrow \Bh$ (viewing $B$ as a function of $y$, see \eqref{eq:t1reads}). 

To obtain the lemma, it thus suffices  to show that $\Bu=f(\rho)+1$. Recall, that $\Bu$ is the unique value of $B$ for which the polynomial $(q-1)z^{d+1} + (2-B-q)z^{d} + Bz - 1$ has a double root in $(0,1)$. We reparameterize $z\rightarrow 1/z$, so that $\Bu$ is the unique value of $B$ for which the following polynomial has a
double root in $(1,\infty)$:
\begin{equation*}
r(z)=z^{d+1}- Bz^d- (2-B-q)z-(q-1).
\end{equation*}
Let $z_c$ be the double root of this polynomial when $B=\Bu$. Solving each of $r(z_c)=0$ and $r'(z_c)=0$ with respect to $B$ and equating the expressions, we obtain that $p(z_c)=0$. It follows that for $B=\Bu$ the double root of the polynomial $r(z)$ is equal to $\rho$. Now, solving $r(\rho)=0$ with respect to $B$ gives us that $\Bu=f(\rho)+1$, as wanted.
\end{proof}

We can now classify the stability of fixpoints with $t=1$.
\begin{lemma}\label{lem:exactone}
For $B>\Bu$, exactly one majority fixpoint is Jacobian attractive. More precisely, the only Jacobian attractive fixpoint with $t=1$ is the one maximizing $x=R_1/R_q$ (among the solutions of \eqref{eq:bprecursion} for $t=1$).
\end{lemma}
We can now explicitly specify the marginal $a$ in the definition of the ordered phase stated in the Introduction. With $x$ as in Lemma~\ref{lem:exactone}, we apply \eqref{jqwe}, which yields:
\begin{equation}\label{eq:marginala}
a=\frac{x^{\Delta/(\Delta-1)}}{x^{\Delta/(\Delta-1)}+q-1}.
\end{equation}
Next, we give the proof of Lemma~\ref{lem:exactone}.

\begin{proof}[Proof of Lemma~\ref{lem:exactone}]
In the setting of Lemma~\ref{lem:fixpoints} we have $t=1$ and thus the interesting eigenvalues of $\M$ are  $\lambda_1:=(B-1)R^2_q/\alpha_q$ and $\lambda_2:=BR^2_1/\alpha_1+(q-2+B)R^2_q/\alpha_q-1$. 
Using that $\alpha_1=R_1(BR_1+(q-1)R_q)$ and $\alpha_q=R_q(R_1+(q-2+B)R_q)$, we obtain the following equivalent expressions for $\lambda_1,\lambda_2$:
\begin{equation}\label{eq:lambda1lambda2}
\lambda_1=\frac{(B-1)R_q}{R_1+(q-2+B)R_q},\quad \lambda_2=\frac{ BR_1}{BR_1+(q-1)R_q}-\frac{R_1}{R_1+(q-2+B)R_q}.
\end{equation}
We will show that
\begin{enumerate}
\item \label{it:edcedc} $\lambda_1<\lambda_2$,
\item \label{it:edcedc2} If $R_1,R_q$ correspond to the fixpoint which maximizes $x$, it holds that $\lambda_2<1/(\Delta-1)$. Otherwise (i.e., if $R_1,R_q$ correspond to the other fixpoint with $t=1$), it holds that  $\lambda_2>1/(\Delta-1)$.
\end{enumerate}
From these two Items, the lemma follows.

For Item~\ref{it:edcedc}, let $W:=\lambda_2-\lambda_1$. Then, expanding everything out, we obtain that 
\begin{align*}
W&=\frac{(B-1)(q-1)R_q( R_1- R_q)}{\big(BR_1+(q-1)R_q\big)\big(R_1+(q-2+B)R_q)}
\end{align*}
and thus $W>0$ (since $B>1$ and $R_1>R_q$).

For Item~\ref{it:edcedc2}, let $Q:=\lambda_2-\frac{1}{\Delta-1}$. Expanding everything out, we have
\begin{align*}
Q&=\frac{R_1R_q\big((\Delta-2)B(q-2+B)-\Delta(q-1)\big)-BR^2_1-(q-1)(q-2+B)R^2_q}{(\Delta-1)(BR_1+(q-1)R_q)(R_1+(q-2+B)R_q)}.
\end{align*}
Thus, to check whether $Q<0$, it is equivalent to check, with $x=\frac{R_1}{R_q}$, whether
\begin{equation}
\big((\Delta-2)B(q-2+B)-\Delta(q-1)\big)x< Bx^2+(q-1)(q-2+B).\label{eq:mainobst}
\end{equation}
Substituting $x=y^d$ and $B-1$ from \eqref{eq:bprecursion}, we obtain the equivalent inequality
\[0<\frac{y\left(y^d-1\right) \left(y^d+q-1\right)p(y)}{\left(y^d-y\right)^2}.\]
where $p(y)$ is the polynomial defined in \eqref{eq:py}. By the proof of Lemma~\ref{lem:numbermajority}, $p(y)>0$ iff $y>\rho$. The proof of Lemma~\ref{lem:numbermajority} further yields that the latter  inequality, throughout the regime $B> \Bu$, is only satisfied by the majority fixpoint with $x$ maximum, thus yielding Item~\ref{it:edcedc2}.

This concludes the proof.
\end{proof}

Having classified the fixpoints which are Jacobian attractive, we now need to see when these are dominant. This entails comparing the values of $\Psi_1$ for the respective phases. Rather than doing this directly, we use Lemma~\ref{new:zako1}. In particular, it is equivalent to compare the values of $\Phi_1$ at the fixpoints. Moreover, note that the expression \eqref{mampo} is invariant upon scaling $R_i$'s by the same factor and hence we only need to compare $\Phi_1(x,1,\hdots,1)$ and $\Phi_1(1,\hdots,1)$, where $x$ is a solution of \eqref{eq:bprecursion} for $t=1$.
\begin{lemma}\label{pppwq}
Let $t=1$ and  $x$ be the solution of \eqref{eq:bprecursion}  with $x$ maximum. Then $\Phi_1(x,1,\hdots,1)\geq \Phi_1(1,1,\hdots,1)$ iff $B\geq \Bp$. Equality holds iff $B=\Bp$.
\end{lemma}
\begin{proof}
By a direct calculation
\begin{align*}
\Phi_1(x,1,\hdots,1)&=\frac{\Delta}{2}\ln\big((x+q-1)^2+(B-1)(x^2+q-1)\big)\\
&-(\Delta-1)\ln\big(x^{\Delta/(\Delta-1)}+q-1\big),\\
\Phi_1(1,1,\hdots,1)&=\frac{\Delta}{2}\ln\big(q^2+(B-1)q\big)-(\Delta-1)\ln(q).
\end{align*}
Using the substitutions $d=\Delta-1$, $x=y^d$ and the second equation in \eqref{eq:bprecursion}, after careful manipulations we obtain
\begin{equation*}
DIF:=\Phi_1(x,1,\hdots,1)-\Phi_1(1,1,\hdots,1)=\frac{1}{2}\ln\bigg(\frac{q^{d-1}\left(y^d+q-1\right)^{d+1}}{(q+y-1)^{d+1}\left(y^{d+1}+q-1\right)^{d-1}}\bigg).
\end{equation*}
It is straightforward to check that for $y=(q-1)^{2/(d+1)}$, $DIF=0$. To find the respective value of $B$ for this value of $y$, we just need to plug the value  $y=(q-1)^{2/(d+1)}$ in  the second equation in \eqref{eq:bprecursion} for $t=1$. In particular, from \eqref{eq:bprecursion} we have that
\[B=1+\frac{(y - 1) (y^d + q - 1)}{(y^d - y)}=\frac{y^{d+1}+(q-2)y-(q-1)}{y^d-y}.\]
For $y=(q-1)^{2/(d+1)}$, we have 
\begin{equation*}
y^{d+1}+(q-2)y-(q-1)=(q-2)\big(q-1+(q-1)^{2/(d+1)}\big),
\end{equation*}
and
\begin{align*}
y^d-y&=(q-1)^{2d/(d+1)}-(q-1)^{2/(d+1)}\\
&=\big((q-1)^{(d-1)/(d+1)} - 1\big)\big(q-1+(q-1)^{2/(d+1)}\big).
\end{align*}
It follows that
\[B=\frac{q-2}{ (q-1)^{(d-1)/(d+1)} - 1}=\frac{q-2}{ (q-1)^{1-2/\Delta} - 1}=\Bp.\]
To prove the lemma, it thus remains to show that $y$ is an increasing function of  $B$ and $DIF$ increases as $y$ increases. This is indeed true. Using \eqref{eq:bprecursion}, one calculates (see the relevant \eqref{eq:t1reads})
\begin{equation*}
\frac{\partial y}{\partial B}\cdot \frac{p(y)}{(y^d-y)^2}=1,
\end{equation*}
and
\begin{equation*}
\begin{aligned}
\frac{\partial DIF}{\partial y}&=\frac{1}{2}\Big(\frac{d(d+1)y^{d-1}}{y^d+q-1}-\frac{d+1}{y+q-1}-\frac{(d-1)(d+1)y^d}{y^{d+1}+q-1}\Big)\\
&=\frac{(d+1)(q-1)p(y)}{2(y+q-1)(y^d+q-1)(y^{d+1}+q-1)},
\end{aligned}
\end{equation*}
where $p(y)$ is the polynomial defined in \eqref{eq:py}, whose positivity has already been established for all $y>\rho$, see the proof of Lemma~\ref{lem:numbermajority}. The claim follows.
\end{proof}

\begin{proof}[Proof of Theorem~\ref{thm:Potts-diagram-detailed}]
We first argue about the local maxima. By Theorem~\ref{thm:connection}, we just need to check the stability of the corresponding fixpoints. By Lemmas~\ref{lem:fixstablemaj} and~\ref{lem:exactone}, only the uniform and the $q$ majority fixpoints can be Jacobian attractive. The uniform fixpoint, by Lemma~\ref{lem:uniform}, is Jacobian attractive when $1<B< \Bh$. The $q$ majority fixpoints, by Lemma~\ref{lem:exactone}, are Jacobian attractive when $B> \Bu$. This proves the assertions in Theorem~\ref{thm:Potts-diagram-detailed} about the local maxima of the function $\Psi_1$.

To argue about the Hessian dominant phases,  it only remains to find the regimes where the disordered/ordered phases are dominant. From the first part of Lemma~\ref{pppwq}, the disordered phase is dominant iff $B\leq\Bp$, whereas the $q$ ordered phases are dominant iff $B\geq \Bp$. From the second part of Lemma~\ref{pppwq}, all of the $q+1$ phases are dominant at $B=\Bp$. This completes the proof of  Theorem~\ref{thm:Potts-diagram-detailed}. 
\end{proof}

\section{\#BIS-hardness for Potts}\label{sec:BISPotts}

We first give a rough description of our reduction.  We will construct a gadget $G$ which is
a balanced, bipartite graph on $(2+o(1))n$ vertices.  There will be $n'=O(n^{1/8})$ vertices on each side
of $G$ which will have degree $\Delta-1$, the remainder have degree $\Delta$.
The key is that $G$ behaves similarly to a random bipartite $\Delta$-regular graph.  Hence, for the ferromagnetic Potts model, when $B>\Bp$, the $q$ ordered phases will dominate (ferromagnetic models on random bipartite $\Delta$-regular graphs have the same dominant phases as on random $\Delta$-regular graphs, see footnote~\ref{foot:bipartite}).
We will take an instance $H$ for  \textsc{\#FerroPotts}($q,B,\Delta$) where $H$ has $n'$ vertices.
We then replace each vertex in $H$ by a gadget $G$.  Then we will use the degree $\Delta-1$ vertices
in these gadgets to encode the edges of $H$, while preserving bipartiteness. The resulting 
graph $H^G$ will have bounded degree $\Delta$ and the Potts model on $H^G$ will ``simulate'' the 
Potts model on $H$.

The gadget $G$ is identical to the one used by Sly~\cite{Sly10}. Before giving the detailed description of the gadget, it would be instructive to explain at an intuitive level the basic construction of the gadget and the properties we are trying to ensure. Roughly, the gadget is a random bipartite graph $\overline{G}$ with $(\Delta-1)$-ary trees of depth $\Theta(\log n)$ attached on a small set of vertices on each side of $\overline{G}$. We use the random bipartite  graph to ensure that the final gadget has $q$ ordered phases when $B>\Bp$, which will be used to encode the $q$ spins of the Potts model in the graph $H$. The trees will ensure, roughly, that the relative error we introduce by approximating the partition function in $H$ with the partition function in $H^{G}$ is polynomially small. To explain this further, the reader should keep in mind that the analysis of the Potts model on the random graph $\overline{G}$ will only give estimates of its partition function within a multiplicative factor $(1\pm \epsilon)$ for some small constant $\epsilon>0$. In turn, using  this bound to analyze the Potts model on $H^{G}$ would result in estimating the partition function of $H$ within a multiplicative factor $(1\pm \epsilon)^{|H|}$, which is way bigger than the polynomial accuracy we seek in approximation-preserving reductions. This obstacle is precisely the reason why the trees are attached to $\overline{G}$: the trees will boost the constant $\epsilon$ to a much smaller quantity of order $n^{-O(1)}$; and then, the final approximation can be made polynomially small as desired.\footnote{Let us remark that sometimes even the $(1\pm \epsilon)$ factor estimates are sufficient to get strong inapproximability results and the attachment of the trees is not needed, see for example the reductions in \cite{SlySun,GSV:colorings}. The difference in those settings is that the corresponding counting problems, i.e., counting independent sets or counting colorings, can be connected to NP-hard problems (such as \textsc{Max-Cut}), yielding that it is NP-hard to approximate the corresponding partition function even within an exponential factor. In contrast, for problems like \#BIS or $\#\textsc{FerroPotts}$ (which correspond to easy decision problems, e.g., finding the max independent set on a bipartite graph or the maximum weight configuration in the ferromagnetic Potts model), such strong inapproximability results are not known; in fact, we only have evidence that an FPRAS is unlikely to exist. This necessitates the study of approximation-preserving reductions in our setting (as in, e.g., \cite{DGGJ,2spinBIS,LLZ}) and thus our quest for the polynomial precision.} As in \cite{Sly10,GSV:2spin,2spinBIS}, this boosting is possible  due to the fact that the ordered phases correspond to non-reconstructible Gibbs measures on the infinite $\Delta$-regular tree (we will expand later on this in the proof of Lemma~\ref{lem:gadgetferropotts} Item~\ref{it:approxind}).

Next, we give the description of the gadget $G$. The gadget $G$ is defined by two parameters $\theta,\psi$ where $0<\theta,\psi<1/8$. The construction of the gadget $G$ has two parts.
First construct the following bipartite graph $\Go$ with vertex set $V^+\cup V^-$.
For $s\in\{+,-\}$,  $|V^s| = n+m'$ where  $m'$
will be defined precisely later.
Take $\Delta$ random perfect matchings between $V^+$ and $V^-$.  Then remove a matching of size $m'$ from one of the $\Delta$ matchings.  Call this graph $\Go$. For later use, let $U:=U^+\cup U^-$ denote the vertices of degree $\Delta$ in $\Go$ and $W:=W^+\cup W^-$ denote the vertices of degree $\Delta-1$ in $\Go$.

In the second stage, for each side of $\Go$, partition the degree $\Delta-1$ vertices into
$n^\theta$ equal sized sets and attach to each set a $(\Delta-1)$-ary tree of depth $\ell$
where $\ell =  \lfloor \psi \log_{\Delta-1}n\rfloor$.  (Use the vertices of $\Go$ as the leaves of these trees.)
Hence each side contains $n^\theta$ trees of size $O(n^\psi)$.
(More precisely, $(\Delta-1)^{\lfloor \theta \log_{\Delta-1}n\rfloor}$ trees, each having  $(\Delta-1)^{\lfloor \psi \log_{\Delta-1}n\rfloor}$ leaves.)
This defines the gadget $G$.
For $s\in\{+,-\}$, let $R^s$ denote the roots of the trees on side $s$ and $R:=R^+\cup R^-$.
Note that each vertex in $R$ has degree $\Delta-1$ and these will be used to encode the
edges of $H$ (we give the reduction explicitly just after stating the relevant gadget lemma, Lemma~\ref{lem:gadgetferropotts}).
Note that $m'=(\Delta-1)^{\lfloor \theta \log_{\Delta-1}n\rfloor+\lfloor \psi \log_{\Delta-1}n\rfloor}$ and $m'=o(n^{1/4})$.

Denote by $G=(V,E)$ the final graph.
Recall, for a configuration $\sigma\in\Omega$, the set of
vertices assigned spin $i$ is denoted by $\sigma^{-1}(i)$.
The phase of a configuration $\sigma: V\rightarrow[q]$ is defined as
the dominant spin among vertices in $U=U^+\cup U^-$ (the vertices of degree $\Delta$ in $\Go$):
\[Y(\sigma):=\arg \max_{i\in [q]}|\sigma^{-1}(i)\cap U|,\]
where ties are broken with an arbitrary deterministic criterion (e.g., the lowest index).

The gadget $G$ behaves like a random bipartite $\Delta$-regular graph because $m'\ll n$,
as we will detail in the upcoming Lemma \ref{lem:gadgetferropotts}.  Hence, since $B>\Bp$,
Theorem~\ref{thm:Potts-diagram-detailed} implies that the $q$ ordered phases are dominant.  Therefore, we will get that
for a sample $\sigma$ from the Gibbs distribution, the phase of $\sigma$ will be (close to) uniformly distributed
over these $q$ ordered phases.   Let phase $i$ refer to
the ordered phase where spin $i$ is the majority.  Once we condition
on the phase for the vertices in $U$, say it is phase $i$,
then each of the roots of the trees appended to $G$, roughly independently, will have spin $i$ with probability $\approx p$ and spin $j\neq i$
with probability $\approx (1-p)/(q-1)$ where
 $p$ is the probability that the root of the infinite $(\Delta-1)$-ary tree has spin $i$ in the Gibbs measure corresponding to the ordered phase $i$.\footnote{The ordered phase $\alphab=(a,(1-a)/(q-1),\dots,(1-a)/(q-1))$ specifies the marginal probabilities for
the root of the infinite $\Delta$-regular tree.  To account for the root having degree $\Delta-1$ one obtains that:
\[ p = \frac{a^{(\Delta-1)/\Delta}}{
(a/(1-a))^{(\Delta-1)/\Delta} + (q-1)^{1/\Delta}}.
\]
Alternatively, $p=x/(x+q-1)$, for the same $x$ as in \eqref{eq:marginala}.
}
Hence, for each of the $q$ possible phases,
we define the following product distribution on the configurations $\sigma_R: R\rightarrow [q]$. For $i\in[q]$, let
\begin{equation}\label{def:proddistribution}
Q^{i}_R(\sigma_R)=p^{|\sigma^{-1}_R(i)|} \Big(\frac{1-p}{q-1}\Big)^{|R\backslash \sigma^{-1}_R(i)|}.
\end{equation}
For future use, one can define completely analogously the product measure $Q^{i}_W(\cdot)$ on configurations  $\sigma_W: W\rightarrow[q]$ (recall that $W$ is the set of vertices with degree $\Delta-1$ in $\Go$).

The following lemma is proved using methods in \cite{Sly10} and its proof is given in Section~\ref{sec:gadgets}. Roughly, the first item in the lemma follows from the symmetries of the Potts model. For the second item, the rough idea is that when the phase  is $i$, the marginal spin distribution of vertices in $W$ in the graph $\overline{G}$ is close to $Q^{i}_W$. The purpose of the trees is to boost this effect; more precisely, make the distance between the marginal spin distribution of vertices in $R$ and $Q^i_R$ an inverse polynomial factor (see Item~\ref{it:approxind} in Lemma~\ref{lem:gadgetferropotts}). In turn, the reason that the trees can accomodate the ``boosting" is that the marginal distribution on $W$ corresponds to an extremal Gibbs measure on the tree, which results in the spins of the roots of the trees being strongly concentrated. 
\begin{lemma}\label{lem:gadgetferropotts}
For every $q,\Delta\geq 3$ and $B>\Bp$, there exist constants $\theta, \psi>0$ such that the graph $G$ satisfies the following with probability $1-o(1)$ over the choice of the graph:
\begin{enumerate}
\item The phases occur with roughly equal probability, so that for every phase $i\in [q]$, we have \label{it:phaseprob}
\[\Big|\mu_G\big(Y(\sigma)=i\big)-\frac{1}{q}\Big|\leq n^{-2\theta}.\]
\item Conditioned on the phase $i$, the spins of vertices in $R$ are approximately independent, that is,   \label{it:approxind}
\[\max_{\sigma_R}\Big|\frac{\mu_G\big(\sigma_R\, |\, Y=i\big)}{Q_R^{i}(\sigma_R)}-1\Big|\leq n^{-2\theta}.\]
\end{enumerate}
\end{lemma}

With Lemma~\ref{lem:gadgetferropotts} at hand, we can now formally state the reduction that we sketched earlier. Let $B>\Bp$. Let $H$ be a graph on $n'$ vertices, where $n'\leq n^{\theta/4}$ and $\theta$ is as in Lemma~\ref{lem:gadgetferropotts}. Assuming an FPRAS for the ferromagnetic Potts model on max-degree $\Delta$ graphs and parameter $B$, we will show that we can approximate $Z_H(B^{*})$, the partition  function of  $H$ in the ferromagnetic Potts model with parameter $B^{*}$, where $B^{*}$ will be determined shortly.

To do this, we first construct a graph $H^G$. First, take $|H|$ disconnected copies of the gadget $G$ in Lemma~\ref{lem:gadgetferropotts} and identify each copy with a vertex $v\in H$. Denote by $\hat{H}^G$ the resulting graph, $G_v$ the copy of the gadget associated to the vertex $v$ in $H$ and by $R^{+}_v,R^{-}_v, R_v$ the images of $R^+,R^-,R$ in the gadget $G_v$, respectively. We next add the edges of $H$ in $\hat{H}^G$. To do this, fix an arbitrary orientation of the edges of $H$.  For each oriented edge $(u,v)$ of $H$, we  add an edge  between one vertex in $R^{+}_u$ and one vertex in $R^{-}_v$, using mutually distinct vertices for distinct edges of $H$.   The resulting graph will be denoted by  $H^G$.  Note that $H^G$ is bipartite and has maximum degree $\Delta$.

For a graph $H$ and activity $B\geq 1$, recall that $Z_H(B)$ is the partition function for the ferromagnetic Potts model at activity $B$ on the graph $H$.
We have the following connection:
\begin{lemma}\label{lem:functionconnectpotts}
Let $\Delta,q\geq 3$ and $B>\Bp$. There exists $B^*>1$ (depending only on $q,\Delta,B$) such that the following holds for every graph $H$ with $n'$ vertices:
\[\big(1-O(n^{-\theta})\big)\frac{q^{n'}Z_{H^G}(B)}{C_H\big(Z_G(B)\big)^{n'}}\leq Z_{H}(B^*)\leq \big(1+O(n^{-\theta})\big)\frac{q^{n'}Z_{H^G}(B)}{C_H\big(Z_G(B)\big )^{n'}},\]
where $C_H=D^{|E(H)|}$ and $D=1+(B-1)\Big(\frac{2p(1-p)}{(q-1)^2}+(q-2)\frac{(1-p)^2}{(q-1)^2}\Big)$.
\end{lemma}

Using Lemma \ref{lem:functionconnectpotts} we can now prove that for all $\Delta\geq 3$, all $B>\Bp$, it is
\#BIS-hard to approximate the partition function of the ferromagnetic Potts model on bipartite graphs of maximum degree~$\Delta$.
\begin{proof}[Proof of Theorem~\ref{thm:BIS-potts}]
Goldberg and Jerrum \cite{GJ:potts} showed that for every $B>1$ it is \#BIS-hard to approximate the partition function
of the ferromagnetic Potts model on all graphs.
Fix $\Delta,q\geq 3$ and $B>\Bp$ for which we intend to prove Theorem \ref{thm:BIS-potts}, and let $B^*=B^*(q,\Delta,B)>1$ be specified as in Lemma \ref{lem:functionconnectpotts}.
We first show that an FPRAS for  approximating the partition function with activity $B$ on graphs with maximum degree $\Delta$
implies an FPRAS for approximating the partition function with activity $B^*$ on all graphs.  It will then be clear that our reduction
is in fact approximation-preserving and hence the theorem will be proven.

Suppose that there exists an FPRAS for approximating the partition function with activity $B$ on bipartite graphs with maximum degree $\Delta$.
 Take an input instance $H$ for which we would like to estimate the partition function of the Potts model at activity $B^*$.
 First generate a random gadget $G$ using the construction defined earlier.  This graph $G$ satisfies the properties in Lemma~\ref{lem:gadgetferropotts} with probability $1-o(1)$.   Approximate the partition function of $G$ at activity $B$ within a multiplicative factor $1\pm\epsilon/10n'$
 using our presumed FPRAS (where, recall, $n'$ is the number of vertices in $H$).
Also, using the presumed FPRAS approximate the partition function of $H^G$ at activity $B$ within a multiplicative factor $1\pm\epsilon/5$. The bounds for $Z_{H}(B^*)$ in Lemma~\ref{lem:functionconnectpotts} are then within a factor $1\pm \epsilon$ for sufficiently large $n$, implying an FPRAS for approximating the partition function at activity $B^*$. This, together with the result of \cite{GJ:potts}, implies an FPRAS for counting independent sets in bipartite graphs.
\end{proof}

\begin{proof}[Proof of Lemma~\ref{lem:functionconnectpotts}]
Recall that $\hat{H}^G$ are the disconnected copies of the gadgets, as defined in the construction of $H^G$.
Note, $Z_{\hat{H}^{G}}(B)=\big(Z_G(B)\big)^{n'}$.
Hence to prove the lemma it suffices to analyze $\frac{Z_{{H}^{G}}(B)}{Z_{\hat{H}^{G}}(B)}$.

For a configuration $\sigma$ on $H^G$, for each $v\in H$, let $Y_v(\sigma)$ denote the phase of $\sigma$ on $G_v$.
Denote the vector of these phases by $\Yc(\sigma)=(Y_v(\sigma))_{v\in H}\in [q]^{H}$, we refer to $\Yc(\sigma)$ as the phase vector for~$\sigma$.

\newcommand\U{\mathcal{U}}

 For $\U\in [q]^{H}$, let $\Omega_\U$ denote the set of configurations $\sigma$ on $H^G$ where
 $\Yc(\sigma) = \U$.  Let $Z_{H^G}(\U)$ be the partition function of $H^G$ restricted to configurations $\sigma\in\Omega_\U$, that is,
 \[Z_{H^{G}}(\U)=\sum_{\sigma\in\Omega_{\U}} B^{m(\sigma)},  \]
where for a configuration $\sigma$, $m(\sigma)$ is the number of monochromatic edges under $\sigma$. We may view $\U$
as an assignment $V(H)\rightarrow [q]$ where $V(H)$ are the vertices in the graph $H$.
Hence, we can consider the number of monochromatic edges in the graph $H$ under the assignment $\U$, which
we denote by $m(\U)$.
Recall the goal is to analyze $\frac{Z_{{H}^{G}}(B)}{Z_{\hat{H}^{G}}(B)}$.  To this end we will analyze
$\frac{Z_{H^{G}}(\U)}{Z_{\hat{H}^{G}}(\U)}$ for every $\U$ and then we will use that every $\U$ is (close to) equally likely in
$\hat{H}^G$ which will follow from Property \ref{it:phaseprob} in Lemma \ref{lem:gadgetferropotts}.

Denote by $R_H$ the set of vertices $\cup_v R_v$, i.e., the union of all the vertices of degree $\Delta-1$ in  $\hat{H}^{G}$. Notice that once we fix an assignment to all of the vertices in $R_H$, by the Markov property of the model, we have that
\begin{align*}
\frac{Z_{H^{G}}(\U)}{Z_{\hat{H}^{G}}(\U)}
&=\sum_{\sigma_{R_{H}}}\mu_{\hat{H}^{G}}(\sigma_{R_{H}} \, |\, \Yc(\sigma)=\U)\prod_{(u,v)\in E(H^G)\backslash E(\hat{H}^G)}B^{\mathbf{1}\{\sigma_{R_H}(u)=\sigma_{R_H}(v)\}}.
\end{align*}
Note that $\mu_{\hat{H}^{G}}(\sigma_{R_{H}} \, |\, \Yc(\sigma)=\U)=\big(1+O(n^{-\theta})\big)\prod_{v\in V(H)}Q^{\U_v}_{R_v}(\sigma_{R_v})$ since $\hat{H}^{G}$ is a union of disconnected copies of $G$ and in each copy of $G$ we have Property~\ref{it:approxind}
of Lemma \ref{lem:gadgetferropotts}. It follows that
\begin{align*}
\frac{Z_{H^{G}}(\U)}{Z_{\hat{H}^{G}}(\U)}&=\big(1+O(n^{-\theta})\big)\sum_{\sigma_{R_H}}\prod_{v\in V(H)}Q^{\U_v}_{R_v}(\sigma_{R_v})\prod_{(u,v)\in E(H^G)\backslash E(\hat{H}^G)}B^{\mathbf{1}\{\sigma_{R_H}(u)=\sigma_{R_H}(v)\}}\\
&=\big(1+O(n^{-\theta})\big)A^{m(\U)}D^{|E(H)|-m(\U)},
\end{align*}
where $A$ (resp. $D$)  is the expected weight of an edge connecting two gadgets which have the same (resp. different) phases. Simple calculations show that
\[A=1+(B-1)\Big(p^2+\frac{(1-p)^2}{q-1}\Big),\quad D=1+(B-1)\Big(\frac{2p(1-p)}{q-1}+(q-2)\frac{(1-p)^2}{(q-1)^2}\Big).\]
Letting $B^{*}=A/D$ and $C_H=D^{|E(H)|}$, we obtain
\begin{equation}\label{eq:phasecontribution}
\frac{Z_{H^{G}}(\U)}{Z_{\hat{H}^{G}}(\U)}=\big(1+O(n^{-\theta})\big)(B^{*})^{m(\U)}C_H.
\end{equation}

Property~\ref{it:phaseprob} in Lemma~\ref{lem:gadgetferropotts} gives that for every $\U$ it holds that
\begin{equation}\label{eq:phasecontribution2}
\big(1-O(n^{-\theta})\big)q^{-n'}\leq \Big(\frac{1}{q}-n^{-2\theta}\Big)^{n'}\leq \frac{Z_{\hat{H}^{G}}(\U)}{Z_{\hat{H}^{G}}}\leq \Big(\frac{1}{q}+n^{-2\theta}\Big)^{n'}\leq \big(1+O(n^{-\theta})\big)q^{-n'}.
\end{equation}
We also have
\begin{equation}\label{eq:totalpart}
Z_{H^G}(B)=\sum_{\U}Z_{H^{G}}(\U)=\sum_{\U}\frac{Z_{H^G}(\U)}{Z_{\hat{H}^{G}}(\U)}Z_{\hat{H}^{G}}(\U)=Z_{\hat{H}^{G}}\sum_{\U}\frac{Z_{H^{G}}(\U)}{Z_{\hat{H}^{G}}(\U)}\frac{Z_{\hat{H}^{G}}(\U)}{Z_{\hat{H}^{G}}}.
\end{equation}

Using the estimates \eqref{eq:phasecontribution}, \eqref{eq:phasecontribution2} in \eqref{eq:totalpart}, we obtain
\[\big(1-O(n^{-\theta})\big)q^{-n'}C_H Z_H(B^{*})\leq \frac{Z_{H^{G}}(B)}{Z_{\hat{H}^{G}}(B)}\leq \big(1+O(n^{-\theta})\big)q^{-n'}C_H Z_H(B^{*}).\]
The result follows after observing that $Z_{\hat{H}^{G}}(B)=\big(Z_G(B)\big)^{n'}$ and rearranging the inequality.
\end{proof}

\subsection{Proving the properties of the gadget}\label{sec:gadgets}
In this section, we prove the properties of the gadget we use, as stated in Lemma~\ref{lem:gadgetferropotts}. We outline the proof and introduce the relevant notation. The proof follows the same approach as in \cite[Theorem 2.1]{Sly10} and uses non-reconstruction results in \cite{MST}. We argue however more thoroughly for Item~\ref{it:phaseprob} in Lemma~\ref{lem:gadgetferropotts}, since in \cite{Sly10} a cruder bound for the probability that a phase appears was sufficient. In our case, the more delicate bound will follow from the symmetries of the Potts model. We first illustrate how symmetry comes into play.

Let $\Sigma^i_G$ be the set of configurations on $G$ which have phase $i$, i.e., $\Sigma^i_G:=\{\sigma:\, V\rightarrow [q]\, | \, Y(\sigma)=i\}$. Moreover, let $\Sigma^{o}_G$ be the set of configurations $\sigma$ which satisfy $\big|\arg \max_{i\in [q]}|\sigma^{-1}(i)\cap U|\big|\geq 2$, that is, $\Sigma^{o}_G$ consists of these configurations whose phase was determined by breaking a tie. We first show that Item~\ref{it:phaseprob} in Lemma~\ref{lem:gadgetferropotts} will follow from showing that  $\Sigma^{o}_G$ has exponentially smaller contribution to the partition function of $G$ than $\Sigma^i_G$ for every $i\in [q]$.

To capture this, for a subset $\Sigma\subseteq \Omega_G$ of the configuration space, denote by $Z_G(\Sigma)$  the partition function restricted to configurations in $\Sigma$, that is,
\begin{equation*}
Z_G(\Sigma)=\sum_{\sigma\in\Sigma} w_G(\sigma).
\end{equation*}

Let $\pi$ be a permutation of the colors $[q]$ which maps color $i$ to color $j$. For a configuration $\sigma$, we denote by $\pi(\sigma)$ the configuration $\pi \circ \sigma$. Clearly, for every configuration $\sigma \in \Sigma^{i}_G\backslash \Sigma^{o}_G$ we have $\pi(\sigma)\in \Sigma^{j}_G\backslash \Sigma^{o}_G$. It follows that for every two colors $i,j$ we have $Z_G(\Sigma^{i}_G\backslash  \Sigma^{o}_G)=Z_G(\Sigma^{j}_G\backslash \Sigma^{o}_G)$. Since
\[Z_G=Z_G(\Sigma^{o}_G)+\sum_{i} Z_G(\Sigma^{i}_G\backslash \Sigma^{o}_G),\]
to get the inequality in Item~\ref{it:phaseprob} of Lemma~\ref{lem:gadgetferropotts} it suffices to show that $Z_G(\Sigma^{o}_G)$ is smaller than $Z_G$ by a sufficiently large polynomial factor. 

We briefly outline the argument for proving that $Z_G(\Sigma^{o}_G)$ is smaller than $Z_G$ by a sufficiently large polynomial factor, introducing at the same time some relevant notation. First, note that the definition of the phase of a configuration makes sense for configurations on $\Go$ as well. For convenience, we will henceforth  use  $Z_G^o, Z_G^i$ as shorthands for $Z_G(\Sigma^{o}_G), Z_G(\Sigma^{i}_G)$ and $Z_{\Go}^o, Z_{\Go}^i$ for their analogues in $\Go$. Roughly, we will first show that $Z^o_{\overline{G}}$ is exponentially smaller than $Z^i_{\overline{G}}$ with probability $1-o(1)$ (over the choice of $\Go$). This part follows from the fact that the $q$ ordered phases are dominant and the fact that $Z^i_{\overline{G}}$ matches its expectation (up to a polynomial factor) with probability $1-o(1)$ (this is the analogue of Lemma~\ref{lem:smallgraph} for the graph distribution induced by $\Go$).   We will then show that $Z^o_{G}$ is smaller  than $Z^i_{G}$ by a factor of  $\exp(n^{1/4})$ by crudely accounting for the contribution of the trees attached in the second step of the construction of $G$. Summing over $i\in[q]$ then yields the desired bound and thus completes the ``symmetry" argument.

To formalize the outline in the previous paragraph, we will have to capture how the partition functions $Z_{\Go}$ and $Z_G$ interplay. Due to the Markov property, this happens only through vertices in $W$ (recall, this is the set of vertices of degree $\Delta-1$ in the graph $\Go$ on which the trees are attached). Thus, we will partition the sets $\Sigma^o_{\Go},\Sigma^i_{\Go}$ according to the configuration $\eta$ on $W$. In particular, $\Sigma^o_{\Go}(\eta)$ will be those configurations $\sigma$ in $\Sigma^o_{\Go}$ such that $\sigma_W=\eta$ and $Z^o_{\Go}(\eta)$ will be the contribution to the partition function of $\Go$ from configurations in $\Sigma^o_{\Go}(\eta)$. Define similarly $\Sigma^i_{\Go}(\eta)$ and $Z_{\Go}(\eta)$. 

We need a final piece of notation. Let $J$ be the union of the trees appended in the second step of the construction of the gadget $G$. Note that the only vertices of $\Go$ included in $J$ are vertices in $W$. Let $Z_J(\eta)$ be the contribution to the partition function of $J$ from configurations $\sigma$ (on $J$) such that $\sigma_W=\eta$. We are now able to put these definitions into work. In particular, we have that
\[Z^i_G=\sum_{\eta:W\rightarrow [q]} Z^i_{\Go}(\eta)Z_J(\eta)\mbox{ and } Z^o_G=\sum_{\eta:W\rightarrow [q]} Z^o_{\Go}(\eta)Z_J(\eta).\]
We will need the following lemma, which is proved combining techniques from \cite{Sly10}, \cite[Appendices A \& B]{GSV:colorings} and the phase diagram for the Potts model (note that for ferromagnetic models, the dominant phases have the same quantitative structure in bipartite graphs\footnote{\label{foot:bipartite}Dominant phases on random bipartite $\Delta$-regular graphs correspond to the global maximizers of $\max_{\bR,\Cb}\frac{ \bR^{\T} \B \Cb}{\|\bR\|_{p}\|\Cb\|_{p}}$ where $p=\Delta/(\Delta-1)$, see \cite[Theorem 4.1]{GSV:colorings}. As a consequence of Observation~\ref{obs:ferro}, for a ferromagnetic model, any such maximum must satisfy $\bR=\Cb$ (up to a scaling factor). This yields that the maximizers are in one-to-one correspondence with the maximizers of the r.h.s of \eqref{zzzok}.}).

\begin{lemma}\label{lem:gadgetproof}
Let $\Gc:=\Gc_n$ denote the distribution of the random bipartite graph $\Go$. For $B>\Bp$, it holds that 
\begin{enumerate}[label=\emph{(\roman*)}, ref=(\roman*)]
\item  \label{it:moments} There exist constants $C_1,C_2,\Psi$ depending only on $q,B,\Delta$, such that for every $i\in[q]$ and  $\eta:W\rightarrow[q]$,
\begin{equation}\label{eq:asympmoments}
\begin{aligned}
\E_{\Gc}\left[Z^i_{\Go}\right]&=\big(1+O(n^{-1/2})\big)C_1C^{2m'}_2\exp(n \Psi),\\ \E_{\Gc}[Z^i_{\Go}(\eta)]&=\big(1+O(n^{-1/2})\big)Q^i_{W}(\eta)\E_{\Gc}\left[Z^i_{\Go}\right].
\end{aligned}
\end{equation}
\item \label{it:null} For all sufficiently small $\epsilon>0$ and sufficiently large $n$, for $i\in [q]$, 
\begin{equation}\label{eq:zo}
\E_{\Gc}\left[Z^{o}_{\Go}\right]\leq \exp(-\epsilon n)\E_{\Gc}\left[Z^i_{\Go}\right].
\end{equation}
\item \label{it:smallgraph}$\displaystyle\max_{i\in[q],\, \eta:\, W\rightarrow [q]}\mathrm{Pr}_{\Gc}\left(Z^i_{\Go}(\eta)<\frac{1}{n}\E_{\Gc}[Z^i_{\Go}(\eta)]\right)\rightarrow 0\mbox{ as } n\rightarrow\infty$.
\end{enumerate}
\end{lemma}
\begin{proof}[Proof of Lemma~\ref{lem:gadgetproof}]
The first equality in \eqref{eq:asympmoments}  is proved in \cite[Lemma B.3]{GSV:colorings}. The second equality in \eqref{eq:asympmoments} is proved in \cite[Lemma 6.11]{GSV:colorings}: the lemma is stated in the case where $m'$ is a constant, but the proof also holds when $m'=o(n^{1/2})$ whenever the dominant phases are Hessian as was first illustrated in \cite{Sly10}. The explicit error factors $O(n^{-1/2})$ in \eqref{eq:asympmoments} are a consequence of Stirling's approximation (see \cite[Lemma 3.1]{Sly10} for the explicit derivation in the hard-core model which is straightforward to adapt to the present setting as well). Combining the above yields Item~\ref{it:moments}.

Item~\ref{it:null} is a consequence of the fact that for some small $\epsilon'>0$, for each ordered phase $i$ there exists an $\epsilon'$-ball around it consisting solely of configurations which are contained in $\Sigma^i_{\Go}\backslash \Sigma^o_{\Go}$. Since the ordered phase $i$ is dominant, a standard compactness argument (see for example the upcoming proof of Lemma~\ref{lem:pain})  yields that $\frac{1}{n}\log\E_{\Gc}\big[Z^{o}_{\Go}\big]$ is strictly less than $\frac{1}{n}\log\E_{\Gc}\big[Z^i_{\Go}\big]$.

Finally, Item~\ref{it:smallgraph} follows from the small subgraph conditioning method in \cite[Appendix A]{GSV:colorings} (see also Theorem~\ref{thm:smallgraphmethod}).
\end{proof}

We conclude this section by giving the proof of Lemma~\ref{lem:gadgetferropotts}. 
\begin{proof}[Proof of Lemma~\ref{lem:gadgetferropotts}]
To get Item~\ref{it:phaseprob}, by the symmetry argument described in the beginning of the section, it suffices to show that for every $i\in[q]$ it holds that $Z^o_G\leq \exp(-n^{1/4})Z^i_G$ with probability $1-o(1)$ over the choice of the graph $G$. We use Lemma~\ref{lem:gadgetproof}. In particular, Markov's inequality yields
\begin{equation}\label{eq:nullsmall}
\mathrm{\Pr}_{\Gc}\left(\sum_{\eta:W\rightarrow [q]} Z^o_{\Go}(\eta) Z_{J}(\eta)>n\sum_{\eta:W\rightarrow [q]} Z_{J}(\eta)\E_{\Gc}[Z^o_{\Go}(\eta)]\right)\rightarrow 0\mbox{ as } n\rightarrow \infty.
\end{equation}
Item~\ref{it:smallgraph} of Lemma~\ref{lem:gadgetproof} yields for every $i\in[q]$
\begin{equation}\label{eq:phaseslarge}
\mathrm{\Pr}_{\Gc}\left(\sum_{\eta:W\rightarrow [q]} Z^i_{\Go}(\eta) Z_{J}(\eta)<\frac{1}{2n}\sum_{\eta:W\rightarrow [q]} Z_{J}(\eta)\E_{\Gc}[Z^i_{\Go}(\eta)]\right)\rightarrow 0\mbox{ as } n\rightarrow \infty.
\end{equation}
 From \eqref{eq:asympmoments} and $\sum_{\eta}Q^i_W(\eta)=1$, it follows that  
\begin{equation}\label{eq:oneoneone}
\begin{aligned}
\sum_{\eta:W\rightarrow [q]} Z_{J}(\eta)\E_{\Gc}[Z^i_{\Go}(\eta)]&= (1+o(1))\E_{\Gc}[Z^i_{\Go}]\sum_{\eta:W\rightarrow [q]} Z_{J}(\eta)Q^i_W(\eta)\\
&\geq (1+o(1))\E_{\Gc}[Z^i_{\Go}]\min_{\eta:W\rightarrow [q]} Z_{J}(\eta).
\end{aligned}
\end{equation}
From the crude bound $\max_{\eta}Z_{J}(\eta)\leq \exp(o(n^{1/4}))\min_{\eta}Z_{J}(\eta)$ and \eqref{eq:zo}, it follows that  
\begin{equation}\label{eq:twotwotwo}
\begin{aligned}
\E_{\Gc}[Z^i_{\Go}]\min_{\eta:W\rightarrow [q]} Z_{J}(\eta)&\geq \exp(n^{1/2})\E_{\Gc}[Z^o_{\Go}]\max_{\eta:W\rightarrow [q]} Z_{J}(\eta)\\
&\geq \exp(n^{1/2})\sum_{\eta:W\rightarrow [q]} Z_{J}(\eta)\E_{\Gc}[Z^o_{\Go}(\eta)].
\end{aligned}
\end{equation}
Combining \eqref{eq:oneoneone} and \eqref{eq:twotwotwo}, yields 
\begin{equation}\label{eq:aaabbb}
\sum_{\eta:W\rightarrow [q]} Z_{J}(\eta)\E_{\Gc}[Z^i_{\Go}(\eta)]\geq \exp(n^{1/2})\sum_{\eta:W\rightarrow [q]} Z_{J}(\eta)\E_{\Gc}[Z^o_{\Go}(\eta)],
\end{equation}
Combining  \eqref{eq:nullsmall}, \eqref{eq:phaseslarge} and \eqref{eq:aaabbb} yields that $Z^o_G\leq n\exp(-n^{1/2})Z^i_G\leq \exp(-n^{1/4})Z^i_G$ with probability $1-o(1)$ over the choice of the graph $G$, as wanted.  This proves the first item of the lemma.

Item~\ref{it:approxind} of the lemma follows exactly the approach in \cite{Sly10}. The required non-reconstruction results to push the approach in \cite{Sly10} are given in \cite[Proof of Theorem 1.4]{MST} (ferromagnetic Potts model on the tree with constant boundary condition). Together with Lemma~\ref{lem:gadgetproof}, the proof of \cite[Theorem 2.1]{Sly10} extends almost verbatim to our case as well. We briefly outline the main ideas of the proof as carried out in \cite{Sly10}.

Recall that the goal is to show that, conditioned on the phase $i$, the distribution of the spins in vertices in $R$ is close to $Q^i_R(\cdot)$. For $i\in [q]$, let (analogously to \cite{Sly10})
\[\mathcal{B}_i:=\Big\{\eta:W\rightarrow[q]\mid \max_{\tau:R\rightarrow[q]}\Big|\mu_G\big(\sigma_R=\tau \mid \sigma_W=\eta\big)-Q_R^{i}(\sigma_R=\tau)\Big|> n^{-3\theta}\Big\},\]
i.e., $\mathcal{B}_i$ is the set of ``bad" configurations on $W$ which exert large influence on vertices in $R$. Note that, while we defined $\mathcal{B}_i$ using  the Gibbs distribution of the graph $G$, we could have used instead the Gibbs distribution of $J$, since, by the Markov property, conditioned on the spins of vertices in $W$, the spins of the vertices in $J$ are conditionally independent from the rest of the vertices in the graph $G$. It follows that
\begin{equation}\label{eq:realbi}
\mathcal{B}_i=\Big\{\eta:W\rightarrow[q]\mid \max_{\tau:R\rightarrow[q]}\big|\mu_J\big(\sigma_R=\tau \mid \sigma_W=\eta\big)-Q_R^{i}(\sigma_R=\tau)\big|> n^{-3\theta}\Big\}.
\end{equation}
Back to the proof, the result will follow from $\mu_G(\sigma_W\in \mathcal{B}_i\mid Y(\sigma)=i)\leq \exp(-n^{2\theta})$, for the technical details see \cite[Proof of Theorem 2.1]{Sly10}.

Note that 
\begin{equation}\label{eq:Hproductmeasure}
\mu_G\big(\sigma_W=\eta\, |\, Y(\sigma)=i\big)=\frac{Z^{i}_{\Go}(\eta)Z_J(\eta)}{Z_G^{i}}=\frac{Z^{i}_{\Go}(\eta)Z_J(\eta)}{\sum_{\eta':W\rightarrow [q]} Z^i_{\Go}(\eta') Z_{J}(\eta')}.
\end{equation}
Using analogous inequalities to those we used to prove Item~\ref{it:phaseprob}, it can be proved that
\[\mu_G\big(\sigma_W\in \mathcal{B}_i\, |\, Y(\sigma)=i\big)\leq poly(n) \nu^i(\sigma_W\in \mathcal{B}_i),\]
where the measure $\nu^i$ is defined on the space of all configurations $\eta:W\rightarrow [q]$ given by
\begin{equation}\label{eq:nuidef}
\nu^{i}(\eta):=\frac{Z_{J}(\eta)Q^{i}_W(\eta)}{\sum_{\eta':W\rightarrow[q]}Z_J(\eta')Q^{i}_W(\eta')}\propto \mu_{J}(\sigma_W=\eta)Q^{i}_W(\eta).
\end{equation}
Our goal is thus to show that 
\begin{equation}\label{eq:MSTtr}
\nu^i(\sigma_W\in\mathcal{B}_i)\leq \exp(-n^{2\theta}),
\end{equation}
It is useful to note at this point that the bound in \eqref{eq:MSTtr} is a property of the trees and, in particular, does not depend on the Gibbs distribution of the (random) graph $\overline{G}$. (Indeed, $\nu^i$ is specified by the Gibbs distribution on the graph $J$, which is a disjoint union of $(\Delta-1)$-ary trees, and the product measure $Q^{i}(\eta)$. Also, $\mathcal{B}_i$ is specified by the Gibbs distribution of the graph $J$, see \eqref{eq:realbi}.) Also, since $J$ is the disjoint union of a polynomial number of identical trees, by a union bound, it suffices to show \eqref{eq:MSTtr} when $J$ consists of a single $(\Delta-1)$-ary tree with height $\ell=\Theta(\log n)$ and $W$ denoting the leaves of the tree. In turn, this will follow from the following doubly exponential upper bound 
\begin{equation}\label{eq:MST}
\nu^i(\sigma_W\in\mathcal{B}_i)\leq \exp(-\exp(C\ell)), 
\end{equation} 
where $C>0$ is a constant. (To recover \eqref{eq:MSTtr} from \eqref{eq:MST}, we just need to tune the parameter $\theta$ of the gadget to ensure that the trees have sufficiently large height $\ell$ relative to $\theta$; recall that $\ell =  \lfloor \psi \log_{\Delta-1}n\rfloor$, so we can choose any constant $\theta$ so that $0<\theta<\frac{C\psi}{2\ln (\Delta-1)}$.)

We will conclude this proof by sketching the main idea behind the strong bound in \eqref{eq:MST}, a detailed proof of the bound with all the relevant connections can be found in Appendix~\ref{app:nonreconstruction}. The bound in \eqref{eq:MST} goes back to the works of Martinelli, Sinclair and Weitz \cite{MSTb,MST} who studied the mixing time of Glauber dynamics on trees with boundary conditions and was first used for the construction of gadgets by Sly \cite{Sly10}. A key idea, captured in \cite[Proof of Lemma 4.3]{Sly10}, is that the measures $\nu^{i}$ on configurations $\eta:W\rightarrow[q]$ can be viewed as projections of Gibbs measures corresponding to the ordered phases on the infinite $(\Delta-1)$-ary tree. The Gibbs measure corresponding to the ordered phase $i$ can be obtained by taking the weak limit of the Potts distribution of a finite tree with depth $\ell$ whose leaves are conditioned to have spin $i$ as $\ell\rightarrow\infty$. In Appendix~\ref{app:fixpoints}, we give an alternative Markov chain construction of these Gibss measures using a broadcasting process which is more convenient to work with. These Gibbs measures are well-known to be extremal; or, equivalently, that the broadcasting process has the non-reconstruction property, which roughly says that the spin of the root can not be reconstructed from a typical configuration on the leaves of the tree (asymptotically in $\ell$), see Appendix~\ref{app:broadcasting} for details that are relevant in our setting and the survey \cite{Mossel} for more details on broadcasting processes on trees. Then, the techniques of \cite{MSTb, MST} further show that a certain eigenvalue condition of the relevant broadcasting matrix allows to quantify the dependence on $\ell$ and thus obtain the bound in \eqref{eq:MST}. In Appendix~\ref{app:nonreconstruction}, we use a similar-flavored result from Sly and Zhang \cite{SlyZhang} which we can apply more directly in our setting.

This concludes the proof of Lemma~\ref{lem:gadgetferropotts}.
\end{proof}

\subsection{\#BIS-Hardness for Bipartite Colorings}\label{sec:biscolorings}
Using our \#BIS-hardness result for the ferromagnetic Potts model on bounded-degree graphs, we now prove our \#BIS-hardness result for colorings on bounded-degree bipartite graphs (Corollary~\ref{cor:colorings}). The reduction between these two problems was first observed in \cite{DGGJ}; here, we just have to work out the bound on $k$ that the application of Theorem~\ref{thm:BIS-potts} yields.
\begin{proof}[Proof of Corollary~\ref{cor:colorings}]
We will show that for all integer $k,\Delta\geq 3$, it holds that
\begin{equation}\label{eq:simpleredred}
\textsc{\#BipFerroPotts}\big(q=k,B=\frac{k-1}{k-2},\Delta\big)\leq_{\mathrm{AP}} \textsc{\#BipColorings}(k,\Delta),
\end{equation}
and that, 
\begin{equation}\label{eq:simpleredredred}
\mbox{whenever $k\leq \Delta/(2\ln \Delta)$, it holds that $B=\frac{k-1}{k-2}>\frac{k-2}{(k-1)^{1-2/\Delta}-1}=\Bp$}. 
\end{equation}
The corollary will then  follow from Theorem~\ref{thm:BIS-potts}. 

To prove \eqref{eq:simpleredred}, let $G=(V,E)$ be an input graph to the problem \textsc{\#BipFerroPotts}$(k,B,\Delta)$ with $B=(k-1)/(k-2)>1$. Construct an instance $G'$ of \textsc{\#BipColorings}$(k,\Delta)$ by subdividing each edge of $G$, i.e., $G'$ is a graph with vertex set $V'=V\cup E$ and edge set $E'=\bigcup_{e=(u,v)\in E}\{(u,e),(e,v)\}$. It is clear that $G'$ is bipartite and every vertex has degree at most $\Delta$.

We claim that the partition function for the $k$-state ferromagnetic Potts model on $G$ with $B=(k-1)/(k-2)$ is equal to the number of proper $k$-colorings on $G'$ (times an easily computable factor equal to $(k-2)^{|E|}$). 

To see this, for a $k$-coloring $\sigma'$ of $G'$, map $\sigma'$ to  a configuration $\sigma$ on $G$ given by the restriction of $\sigma'$ to vertices in $V$. Let $\sigma:V\rightarrow [k]$ be any configuration of the Potts model on $G$. The claim will follow by showing that the number of colorings of $G'$ which map to $\sigma$ is given by $(k-1)^{m(\sigma)}(k-2)^{|E|-m(\sigma)}$ where $m(\sigma)$ denotes the number of monochromatic edges in $G$ under $\sigma$. Indeed, for a monochromatic edge $e=(u,v)\in E$ under $\sigma$, there are $k-1$ ways to choose the color of the vertex $e\in V'$ in the graph $G'$. In contrast, if $e=(u,v)\in E$ is not monochromatic under $\sigma$, there are $k-2$ ways to choose the color of the vertex $e\in V'$ in the graph $G'$. This completes the proof of \eqref{eq:simpleredred}.  

We next show \eqref{eq:simpleredredred}. For $\Delta\leq 16$, we have $k\leq\frac{\Delta}{2\ln\Delta}<3$, so we may assume that $\Delta\geq 17$ (otherwise, there is nothing to prove).  We first reduce \eqref{eq:simpleredredred} to the case $k=\frac{\Delta}{2\ln\Delta}$. Let $d:=\Delta-1$, $z_o:=(k-1)^{1/(d+1)}$. The inequality $\frac{k-1}{k-2}>\frac{k-2}{(k-1)^{(d-1)/(d+1)}-1}$ is equivalent to
\[h_d(z_o)>0 \mbox{ where } h_d(z)=-z^{2(d+1)}_o+z^{2d}_o+z_o^{d+1}-1.\]
Fix $d\geq 16$. The polynomial $h_d(z)$ has two change of signs, so by the Descarte's rule of signs, it has at most two positive roots. Clearly, $z=1$ is a root of $h_d(z)$ and since $h'_d(1)>0$, there is one more root $z_d>1$. Thus, $h_d(z_o)>0$ iff $z_o<z_d$ (note that $z_o>1$). Thus, to show \eqref{eq:simpleredredred}, it suffices to consider the case where $k=\Delta/(2\ln \Delta)$.

Now, we prove the desired inequality for $k=\Delta/(2\ln \Delta)$. The inequality is equivalent to 
\[(k-1)^{-2/\Delta}>1-\frac{k-2}{(k-1)^2}\]
Now using the bound $x^{-2/\Delta}=\exp\big(-\frac{2\ln x}{\Delta}\big)\geq 1-\frac{2\ln x}{\Delta}$ for $x=k-1$, we only need to prove that 
\[\frac{2\ln(k-1)}{\Delta}<\frac{k-2}{(k-1)^2}. \] 
Now $2(k-1)\ln(k-1)\leq 2k\ln k=\Delta(1-\frac{\ln(2 \ln\Delta)}{\ln \Delta})$, so the inequality will follow from
\[k-1>\frac{\ln \Delta}{\ln(2 \ln\Delta)},\]
which holds for $k=\Delta/(2\ln \Delta)$ and all $\Delta\geq 17$, as wanted. This completes the proof of \eqref{eq:simpleredredred}.

This concludes the proof of Corollary~\ref{cor:colorings}.
\end{proof}

\section{Torpid mixing of Swendsen-Wang}\label{sec:swendsen-wang}


In this section, we prove Theorem \ref{thm:SW} about
torpid mixing of the Swendsen-Wang algorithm at the critical activity $B=\Bp$. More precisely, we will show that, with probability $1-o(1)$ over the choice of a random $\Delta$-regular graph with $n$ vertices, the  mixing time of the Swendsen-Wang algorithm is exponential in $n$. We will exploit Theorem \ref{thm:Potts-diagram-detailed} for $B=\Bp$, which in combination with Lemma \ref{lem:smallgraph}, essentially implies that for this value of $B$, we have \textit{coexistence} of the ordered and disordered phases in a random $\Delta$-regular graph (with probability $1-o(1)$).

Denote by $\u$ the disordered phase and by $\m_1,\hdots,\m_q$ the $q$ ordered phases of Theorem \ref{thm:Potts-diagram-detailed}. Note that the $q$ ordered phases are identical up to a permutations of the colors and for the purposes of this section we can treat them in a uniform manner. Thus, denote $\m = \{\m_1,\hdots,\m_q\}$. We will say that a configuration $\sigma$ is close to $\u$ (resp. $\m$) if the color frequencies in $\sigma$ are close to those prescribed by $\u$ (resp. one of $\m_1,\hdots,\m_q$).

Lemma \ref{lem:smallgraph} implies that, with probability $1-o(1)$ over the choice of the graph, the set of configurations near $\u$ and $\m$ dominate the Gibbs distribution, in the sense that these sets each have measure $\geq 1/poly(n)$ and the rest of the configurations have exponentially smaller mass.  To analyze the Swendsen-Wang algorithm we need a more refined picture which includes
the number of monochromatic edges in such configurations.  To this end, we define the following
quantities which roughly correspond to the expected number of monochromatic
edges for configurations in these two sets scaled by a factor of $n$ (see the upcoming equation \eqref{eq:optimalxs} and the remarks thereafter for the derivations).  Hence, let
\begin{equation}\label{eq:monochr}
E_{\m}:= \frac{\Delta}{2} \frac{B (x^2 + q - 1)}{(x + q - 1)^2 + (B - 1) (x^2 + q - 1)}, \mbox { and } E_{\u} := \frac{\Delta}{2} \frac{B}{q+B-1},
\end{equation}
where $x$, defined in Lemma~\ref{pppwq}, is a solution
of the normalized tree recursions.

Now we can define the set of configurations with vertex marginals close to $\u$ and $\m$ and
edge marginals close to $nE_{\u}$ and $nE_{\m}$, respectively.
For a configuration $\sigma\in \Omega$,
let $e_G(\sigma)$ denote the number of \textit{monochromatic} edges in $G$ under the spin configuration $\sigma$.
Recall $\sigma^{-1}(i)$ is the set of vertices with spin $i$ in $\sigma$. Let $\mathbf{c}(\sigma)$ denote the vector 
$(|\sigma^{-1}(1)|/n,\hdots,|\sigma^{-1}(q)|/n)$. For $\eps>0$, let
\begin{align}
 U=U(\eps) &:= \left\{ \sigma\in\Omega\, \big| \, \norm{\mathbf{c}(\sigma) - \u}_{\infty} \leq \eps \mbox{ and }
 |e_G(\sigma) - E_{\u} n| < \eps n \right\}.\label{eq:uconfigurations}\\
M=M(\eps) &:= \left\{ \sigma\in\Omega \,  \big|\ \mbox{ for some }\m_j,\, \norm{\mathbf{c}(\sigma) - \m_j}_{\infty} \leq \eps \mbox{ and }
 |e_G(\sigma) - E_{\m} n| < \eps n \right\},\label{eq:mconfigurations}\\
 T=T(\eps)&:= \Omega\setminus (U(\eps)\cup M(\eps)).
\end{align}

The following lemma is proved in Section~\ref{sec:phasecoexistence} and is the  main tool to obtain our torpid mixing results.
 \begin{lemma}
 \label{lem:main-SW}
Let $B=\Bp$. For all sufficiently small $\epsilon>0$,  there exists $C>0$, such that with probability $1-o(1)$ over the choice of the graph $G\sim \Gc(n,\Delta)$, it holds that
\begin{equation}
\label{eq:U-M-T}
\mu_G(U)  \geq  1/poly(n), \qquad
\mu_G(M)  \geq  1/poly(n), \qquad
\mu_G(T)  \leq  \exp(-Cn).
\end{equation}
\end{lemma}

For the rest of this section, we fix a graph $G$  whose Gibbs distribution satisfies \eqref{eq:U-M-T}. By Lemma~\ref{lem:main-SW}, this holds for asymptotically almost all $\Delta$-regular graphs.

Now to prove that the chain is torpidly mixing we will bound its conductance defined as
$\Phi_{SW} = \min_{S;\emptyset \subset S\subset\Omega} \Phi_{SW}(S)$ where
\[    \Phi_{SW}(S) = \frac{\sum_{\sigma \in S}\mu(\sigma)P(\sigma, \overline{S})}{\mu(S)\mu(\overline{S})},
\]
where $P$ denotes the transition matrix  for Swendsen-Wang.  To bound the conductance of the set $M$, we prove that
a configuration in $M$ is unlikely to transition to $U$ in one step.

\begin{lemma} \label{claim:M}
For $\sigma\in M$, \ \   $P\big(\sigma, U\big) < \exp(-c n)$ for some positive constant $c>0$.
\end{lemma}

\begin{proof}
We are going to argue that, with probability $1-\exp(-\Omega(n))$, the number of mononochromatic edges after one transition of Swendsen-Wang is too large to be in the set $U$. Note that the third step of Swendsen-Wang cannot decrease the number of monochromatic edges, so it suffices to analyze the first two steps.

Since $\sigma\in M$, by definition, the number of monochromatic edges under $\sigma$ is $(E_\m\pm \epsilon)n$. The expected number of edges left after the second step of Swendsen-Wang is thus $(1-1/B)(E_\m\pm \epsilon)n$. The following claim implies that for sufficiently small $\epsilon$, this is greater than $(E_\u\pm \epsilon)n$, the number of monochromatic edges in a configuration from $U$. The proof is given in Section~\ref{sec:phasecoexistence}.
\begin{claim}\label{cl:technicalineq}
Let $\Delta\geq 3$ and $q\geq 2\Delta/\log\Delta$. For $B=\Bp$, it holds that $E_{\m}/E_{\u}>1/(1-1/B)$.
\end{claim}
By standard Chernoff bounds, we can thus conclude that for sufficiently large $n$ the transition from $\sigma$ to $U$ happens with exponentially small probability.
\end{proof}
\begin{remark}\label{rem:lowerboundq}
We believe that the lower bound on $q$ in Theorem~\ref{thm:SW} arises from ignoring the effect of the third step of Swendsen-Wang.  
\end{remark}

We can now bound the conductance $\Phi_{SW}$ of Swendsen-Wang.
\begin{align*}
\Phi_{SW} & \leq
\Phi_{SW}(M)
 =  \frac{\sum_{\sigma \in M}\mu(\sigma)P(\sigma, \overline{M})}{\mu(M)\mu(\overline{M})}
\\
& \leq
poly(n)\Big(\sum_{\sigma \in M}\mu(\sigma)P(\sigma, U) + \sum_{\sigma \in M}\mu(\sigma)P(\sigma, T)   \Big)
\qquad \mbox{by \eqref{eq:U-M-T} } \\
& \leq
poly(n)\Big( \exp(-Cn)\mu(M) + \sum_{\tau\in T} \mu(\tau)P(\tau,M)  \Big)
\ \    \mbox{by Lemma \ref{claim:M} \& reversibility}
\\
& \leq
poly(n)\big( \exp(-Cn) + \mu(T) \big)
\\
& \leq
\exp(-C'n) \qquad \mbox{by \eqref{eq:U-M-T}}
\end{align*}
\begin{proof}[Proof of Theorem \ref{thm:SW}]
Standard conductance results imply that the mixing time is $\Omega(1/\Phi_{SW})$, which proves the theorem
based on the above bounds.
(See \cite[p. 255]{MT} for such a statement using the form of (normalized) conductance as used here.)
\end{proof}

\subsection{Phase Coexistence for Random $\Delta$-regular Graphs}\label{sec:phasecoexistence}
In this section, we prove Lemma~\ref{lem:main-SW}. The lemma will mostly follow from Lemma~\ref{lem:smallgraph}. We will need though a more refined analysis of the partition function conditioned on  configurations close to $\u$ and $\m$.

In analogy to \eqref{eq:uconfigurations}, \eqref{eq:mconfigurations}, define
\begin{align*}
\widehat{U}=\widehat{U}(\epsilon)&:=\Big\{\sigma\in \Omega\, \Big|\, \norm{\mathbf{c}(\sigma)-\u}_{\infty}\leq \epsilon\Big\},\\
\widehat{M}_j=\widehat{M}_j(\epsilon)&:=\Big\{\sigma\in \Omega\, \Big|\, \norm{\mathbf{c}(\sigma)-\m_j}_{\infty}\leq \epsilon\Big\}.\\
\widehat{T}=\widehat{T}(\epsilon)&:=\Omega\backslash \big(\widehat{U}\cup \widehat{M}_1\cup \hdots\cup \widehat{M}_q\big),
\end{align*}
and for a subset $\Sigma\subseteq \Omega$ of the configuration space, denote by $Z_G(\Sigma)$  the partition function restricted to configurations in $\Sigma$, that is,
\begin{equation*}
Z_G(\Sigma)=\sum_{\sigma\in\Sigma} w_G(\sigma), \mbox{ so that } \mu_G(\Sigma)=\frac{Z_G(\Sigma)}{Z_G}.
\end{equation*}
We first prove the following weaker version of Lemma~\ref{lem:main-SW}.
\begin{lemma}\label{lem:weak}
Let $B=\Bp$. Let $\Psi=\max_{\alphab\in\triangle_q}\Psi_1(\alphab)$, where $\Psi_1(\alphab)$ is given by \eqref{eq:limitfirst}. For all sufficiently small $\epsilon>0$,  there exists $C>0$, such that with probability $1-o(1)$ over the choice of the graph $G\sim \Gc(n,\Delta)$, it holds that
$Z_G\big(\widehat{T}\big)\leq \exp(-Cn)\exp(\Psi n)$ and
\begin{align*}
\frac{\exp(n\Psi)}{poly(n)} \leq \min \big\{&Z_G\big(\widehat{U}\big),Z_G\big(\widehat{M}_1\big),\hdots,Z_G\big(\widehat{M}_q\big)\big\}\\
&\leq\max\big\{Z_G\big(\widehat{U}\big),Z_G\big(\widehat{M}_1\big),\hdots,Z_G\big(\widehat{M}_q\big)\big\}\leq poly(n)\exp(n\Psi).
\end{align*}
\end{lemma}
\begin{proof}[Proof of Lemma~\ref{lem:weak}]
Define the region $\mathcal{T}$ by
\[\mathcal{T}=\{\alphab\in \triangle_q\, |\, \norm{\alphab-\u}_{\infty}\geq \epsilon, \norm{\alphab-\m_1}_{\infty}\geq \epsilon,\hdots, \norm{\alphab-\m_q}_{\infty}\geq \epsilon\}.\]
Recall from Theorem~\ref{thm:Potts-diagram-detailed} that, for $B=\Bp$, the global maximum of $\Psi_1(\alphab)$ occurs exactly when $\alphab$ is equal to one of $\u,\m_1,\hdots,\m_q$. It follows that \[\max_{\alphab\in\mathcal{T}}\Psi_1(\alphab)<\Psi-C',\]
for some constant $C'=C'(\epsilon)>0$.  Note that for fixed $n$ the possible values of $\alphab\in\mathcal{T}$ are polynomially many. For all sufficiently large $n$, we thus have
\[\E_\Gc[Z_G\big(\widehat{T}\big)]\leq poly(n)\max_{\alphab\in \mathcal{T}}\E_\Gc[Z_G^{\alphab}]\leq poly(n)\exp\big((\Psi-C')n\big).\]
Let $C$ be such that $C'>C>0$.  By Markov's inequality, we obtain $Z_G\big(\widehat{T}\big)\leq \exp(-Cn)\exp(\Psi n)$ with probability $1-o(1)$. This establishes the first part. For the second part, the upper bounds follow from Markov's inequality by the same token. The lower bounds follow from Lemma~\ref{lem:smallgraph} and the observation that at the disordered/ordered critical activity $\Bp$ the expectations of $Z_G^{\u},Z_G^{\m_1},\hdots,Z_G^{\m_q}$ are within a polynomial factor of $\exp(n\Psi)$. By a simple union bound we can thus ensure all the properties stated with probability $1-o(1)$, as desired.
\end{proof}

Having established Lemma~\ref{lem:weak}, we are ready to start arguing about the empirical distribution of edges, more precisely, the fraction of edges in $G$ whose endpoints are assigned colors $i,j$ for configurations $\sigma\in \widehat{U},\widehat{M}$. The rough idea is as follows. Fix an arbitrary $\alphab\in \triangle_q$. In Section~\ref{sec:defer-prelim}, we established that $\Psi_1(\alphab)=\max_{\x}\Upsilon_1(\alphab,\x)=(\Delta-1)f_1(\alphab)+\Delta\max_{\x}g_1(\x)$, where the latter maximization is over $\x$ which satisfy the constraints \eqref{eq:constraintfirst}. Since $g_1(\x)$ is strictly concave in the convex region it is defined, this maximum is attained for a unique vector $\x$, which from here on we shall denote by $\x_{\alphab}$. Essentially by the same line of arguments as in the proof of  Lemma~\ref{lem:weak}, all the contribution to the first moment $\E_\Gc[Z^{\alphab}_G]$ comes from those $\x$ which are close to $\x_{\alphab}$. Thus, by Markov's inequality and the lower bounds of Lemma~\ref{lem:weak}, if $\x$ is sufficiently far away from $\x_{\alphab}$, with probability $1-o(1)$ over the choice of the graph $G$, the empirical edge distribution of a configuration $\sigma\sim \mu_G$ will equal $\x$ with exponentially small probability. We are thus left to argue that for those $\x$ close to $\x_{\alphab}$, the actual contribution to $Z^{\alphab}_G$ from configurations with edge empirical distribution $\x$ is close to its expectation. But for a graph satisfying Lemma~\ref{lem:weak} this is immediately guaranteed, since we know a lower bound on $Z^{\alphab}_G$ which matches its expectation up to a polynomial factor.

It is useful at this point to give the expressions for the optimal vector $\x_{\alphab}$, when $\alphab$ is a dominant phase. Note that a dominant phase $\alphab$ corresponds (via \eqref{jqwe}) to a fixpoint $(R_1,\hdots,R_q)$ of the tree recursions \eqref{kkrtko}. The entries of the optimal vector $\x^*=\x_{\alphab}$ are given by
\begin{equation}\label{eq:optimalxs}
x_{ij}^{*}=\frac{B_{ij}R_iR_j}{\sum_{i,j}B_{ij}R_iR_j}\mbox{ for } i,j\in[q].
\end{equation}
Indeed, using that $(R_1,\hdots,R_q)$ specify a fixpoint of the tree recursions \eqref{kkrtko}, it can  be checked that $\sum_{j}x_{ij}^*=\alpha_i$ for $i\in[q]$, and also that $\x^*$ is a critical point of $g(\x)$. By the strict concavity of the function $g_1$, it follows that $\x^{*}$ is the unique optimal vector, as claimed. Note that the expressions for $E_\m$ and $E_\u$ in \eqref{eq:monochr} can easily be derived from \eqref{eq:optimalxs} when adapted to the ferromagnetic Potts model. 

We next introduce some relevant notation.  We first want to capture the contribution of configurations with specific edge empirical distribution. To do this, we need the notation $\x_G(\sigma)$ of Section~\ref{sec:defer-prelim} introduced just before the expression for the first moment \eqref{eq:firstmoment}. For $\x_0\in \triangle_{q^2}$, $\Sigma\subseteq \Omega$,  define $Z_G(\Sigma,\x_0)$ as
\[Z_G(\Sigma,\x_0)=\sum_{\sigma\in \Sigma}w_G(\sigma)\mathbf{1}\{\x_G(\sigma)=\x_0\}, \mbox{ so that } Z_G(\Sigma)=\sum_{\x\in \triangle_{q^2}}Z_G(\Sigma,\x).\]
Thus, the definition of $Z_G(\Sigma,\x_0)$ restricts the partition function not only to configurations belonging to $\Sigma$, but also to those having edge empirical distribution equal to $\x_0$. Now, we further extend this definition to capture that the edge empirical distribution is far from a prescribed vector $\x_0$.  For $\epsilon>0$, define $Z_G(\Sigma,\x_0,\epsilon)$ as
\[Z_G(\Sigma,\x_0,\epsilon)=\sum_{\sigma\in \Sigma}\,\sum_{\x; \norm{\x-\x_0}_{\infty}\geq \epsilon}w_G(\sigma)\mathbf{1}\{\x_G(\sigma)=\x\}.\]
We are ready to prove the following.
\begin{lemma}\label{lem:pain}
Let $\Psi=\max_{\alphab\in\triangle_q}\Psi_1(\alphab)$, where $\Psi_1(\alphab)$ is given by \eqref{eq:limitfirst}. For all sufficiently small $\epsilon>0$, there exists $C>0$ such that with probability $1-o(1)$ over the choice of the graph $G\sim \Gc(n,\Delta)$, it holds that
\[\max\{Z_G(\widehat{U},\x_\u,\epsilon),Z_G(\widehat{M}_1,\x_{\m_1},\epsilon),\hdots, Z_G(\widehat{M}_q,\x_{\m_q},\epsilon)\}\leq \exp(-C n)\exp(\Psi n).\]
\end{lemma}
\begin{proof}
We prove that the upper bound $Z_G(\widehat{U},\x_\u,\epsilon)\leq \exp(-C n)\exp(\Psi n)$ holds with probability $1-o(1)$ over the choice of the graph. The remaining random variables may be treated similarly and thus the claim follows by a union bound.

Observe that
\[Z_G\big(\widehat{U},\x_\u,\epsilon\big)=\sum_{\alphab; \norm{\alphab-\u}_{\infty}\leq \epsilon}Z_G(\Sigma^{\alphab},\x_\u,\epsilon).\]
It follows that 
\begin{equation}\label{eq:applymaxmax}
\E_\Gc\big[Z_G\big(\widehat{U},\x_\u,\epsilon\big)\big]=\sum_{\alphab; \norm{\alphab-\u}_{\infty}\leq \epsilon}\E_\Gc[Z_G(\Sigma^{\alphab},\x_\u,\epsilon)].
\end{equation}
Note that the sum in \eqref{eq:applymaxmax} is over polynomially many vectors $\alphab$ satisfying $\norm{\alphab-\u}_{\infty}\leq \epsilon$. Further, for fixed $\alphab$, $\E_\Gc[Z_G(\Sigma^{\alphab},\x_\u,\epsilon)]$ is also a sum over polynomially many $\x$ satisfying $\norm{\x-\x_\u}_\infty\geq \epsilon$. For a fixed $\x$, the exponential order of the term in the latter sum corresponding to $\x$ is given by the function $\Upsilon_1(\alphab,\x)$. By approximating the sums with their maximum terms (and using the continuity of the function $\Upsilon_1$), it is standard to conclude from here that
\begin{equation}\label{eq:applymaxmax2}
\E_\Gc\big[Z_G\big(\widehat{U},\x_\u,\epsilon\big)\big]=\exp(o(n))\exp\big(n\max_{\norm{\alphab-\u}_{\infty}\leq \epsilon}\max_{\norm{\x-\x^{*}_\u}_{\infty}\geq \epsilon}\Upsilon_1(\alphab,\x)\big).
\end{equation}
Note that the maximum in \eqref{eq:applymaxmax2} is justified by standard compactness arguments (which we give below since we will need it in the proof).

Consider the region \[\mathcal{T}(\epsilon):=\{(\alphab,\x)\mid \alphab\in \triangle_q, \x\in\triangle_{q^2},\norm{\alphab-\u}_\infty\leq \epsilon, \norm{\x-\x_{\u}}_\infty\geq \epsilon\},\] i.e., the region $T(\epsilon)$ consists of those pairs $(\alphab,\x)$ such that $\alphab$ is $\epsilon$-close to $\u$, but $\x$ is $\epsilon$-far from the optimal vector $\x_\u$. Let $\Psi'$ be the maximum of $\Upsilon_1(\alphab,\x)$ over the region $\mathcal{T}(\epsilon)$ (this exists since $\mathcal{T}(\epsilon)$ is compact and $\Upsilon_1(\alphab,\x)$ is continuous). Recall now that for fixed $\alphab$ the function $\Upsilon_1(\alphab,\x)$ is strictly concave in $\x$. Since the maximizers of $\Psi_1(\alphab)=\max_{\x\in \triangle_{q^2}}\Upsilon_1(\alphab,\x)$ are the vectors $\u,\m_1,\hdots,\m_q$, it follows that the maximizers of $\Upsilon_1(\alphab,\x)$ are  $(\u,\x_{\u}),(\m_1,\x_{\m_1}),\hdots,(\m_q,\x_{\m_q})$. Since for all sufficiently small $\epsilon>0$ none of these maximizers lies in the region $\mathcal{T}(\epsilon)$, we have that $\Psi'<\Psi$ and hence there exists $C'(\epsilon)>0$ such that
\[\max_{\norm{\alphab-\u}_{\infty}\leq \epsilon}\max_{\norm{\x-\x_\u}_{\infty}\geq \epsilon}\Upsilon_1(\alphab,\x)<\Upsilon_1(\u,\x_\u)-C'.\]
By choosing $C$ so that $0<C<C'$, the desired bound now follows from \eqref{eq:applymaxmax} by an application of  Markov's inequality.
\end{proof}

We are now ready to give the proof of Lemma~\ref{lem:main-SW}.
\begin{proof}[Proof of Lemma~\ref{lem:main-SW}]
By a union bound, a graph $G\sim\Gc(n,\Delta)$ satisfies with probability $1-o(1)$ both Lemmas~\ref{lem:weak} and~\ref{lem:pain}. We have $Z_G=Z_G\big(\widehat{U}\big)+Z_G\big(\widehat{M}_1\big)+\hdots+Z_G\big(\widehat{M}_q\big)+Z_G\big(\widehat{T}\big)$. By Lemma~\ref{lem:weak}, we thus have
\[\frac{\exp(n\Psi)}{poly(n)} \leq Z_G\leq poly(n)\exp(n\Psi).\]
Now observe that the sets $U$, $M$ defined in \eqref{eq:uconfigurations} and \eqref{eq:mconfigurations} satisfy
\[Z_G(U)\geq Z_G\big(\widehat{U}\big)-Z_G\big(\widehat{U},\x^*_{\u},\epsilon\big)\geq \frac{\exp(n\Psi)}{poly(n)},\]
\[Z_G(M)\geq Z_G\big(\widehat{M}\big)-Z_G\big(\widehat{M},\x^*_{\m},\epsilon\big)\geq \frac{\exp(n\Psi)}{poly(n)},\]
\[Z_G(T)\leq Z_G\big(\widehat{T}\big)\leq \exp(-Cn)\exp(n\Psi).\]
The conclusion follows.
\end{proof}

To complete the proofs for Section~\ref{sec:swendsen-wang}, we now give the proof of Claim~\ref{cl:technicalineq}.
\begin{proof}[Proof of Claim~\ref{cl:technicalineq}]
Using \eqref{eq:monochr}, we have that 
\[\frac{E_{\m}}{E_{\u}}=\frac{(q+B-1)(x^2 + q - 1)}{(x + q - 1)^2 + (B - 1) (x^2 + q - 1)}=1+\frac{(q-1)(x-1)^2}{{(x + q - 1)^2 + (B - 1) (x^2 + q - 1)}}.\]
It follows that 
\[\frac{E_{\m}}{E_{\u}}-\frac{1}{1-1/B}=\frac{(q-1)(x-1)^2}{{(x + q - 1)^2 + (B - 1) (x^2 + q - 1)}}-\frac{1}{B-1}.\]
Using the substitutions $d=\Delta-1$, $x=y^d$ and the second equation in \eqref{eq:bprecursion} (for $t=1$), the r.h.s. can be rewritten as
\begin{equation*}
\frac{E_{\m}}{E_{\u}}-\frac{1}{1-1/B}=\frac{y\left(y^d-y\right) \big((q-2) y^{d}-(q-1) y^{d-1}-(q-1)\big)}{(y-1) \big(y^d+q-1\big) \big(y^{d+1}+q-1\big)}.
\end{equation*}
Recall, from the proof of Lemma~\ref{pppwq}, that for $B=\Bp$, it holds that $y=y_o$, where $y_o=(q-1)^{2/(d+1)}$. Since $y_{o}>1$, to prove the claim we only need to show that $p(y_o)>0$, where $p(y):=(q-2) y^{d}-(q-1) y^{d-1}-(q-1)$. Massaging, we obtain the equivalent inequality
\begin{equation}\label{eq:qineq}
h_d((q-1)^{1/(d+1)})>0,\mbox{ where } h_d(z):=z^{2d}-z^{2(d-1)}-z^{d-1}-1.
\end{equation}
Fix $d\geq 2$. If $z_d$ is such that $h_d(z_d)>0$, then $h_d(z)>0$ for all $z>z_d$. This is again a consequence of the Descartes' rule of signs: the polynomial $h_d(z)$ has exactly one positive root, say $\rho_d$, and hence $h_d(z)>0$ for $z$ positive is equivalent to $z>\rho_d$. It follows that to prove \eqref{eq:qineq} for $q\geq 2(d+1)/\ln(d+1)$, we only need to argue for its validity when  $q=q_o:=2(d+1)/\ln(d+1)$. In other words, we need to show that $h_d(z_o)>0$ where $z_o:= (q_o-1)^{1/(d+1)}$. By direct calculations, it can be checked that the inequality is true for $d=2,\hdots,9$. We therefore assume that $d\geq 10$ in what follows. 

For all $w>1$ it holds that $(w+3/4)^2>w^2+w+1$, which for $w=z^{d-1}_o$ gives $(z^{d-1}_o+3/4)^2>z^{2(d-1)}_o+z^{d-1}_o+1$. Thus, we only need to show that $z^d_o>z^{d-1}_o+3/4$, or 
\begin{equation}\label{eq:zoineq}
z_o^{d-1}(z_o-1)>3/4.
\end{equation}  
To handle the two factors in \eqref{eq:zoineq}, we will use the following bounds on $z_o$.
\begin{equation}\label{eq:zobounds}
z_o\geq \left(C_1\frac{d+1}{\ln(d+1)}\right)^{1/(d+1)},\quad z_o\geq 1+C_2\frac{\ln(d+1)}{d+1},\mbox{ where } C_1:=\frac{3}{2}, \ C_2:=\frac{3}{4}.
\end{equation}
The first bound follows  from the inequality $q_o\geq 1+\frac{3}{2}\frac{d+1}{\ln(d+1)}$. The second bound follows from the first bound and $x^{1/(d+1)}=\exp\big(\frac{\ln x}{d+1}\big)\geq 1+\frac{\ln x}{d+1}$ together with $\ln(C_1\frac{d+1}{\ln(d+1)})\geq C_2\ln(d+1)$.

Plugging the bounds \eqref{eq:zobounds} in \eqref{eq:zoineq}, it suffices to check whether 
\begin{equation*}
C_2\frac{\ln(d+1)}{d+1}\left(C_1\frac{d+1}{\ln(d+1)}\right)^{(d-1)/(d+1)}\geq \frac{3}{4}, \mbox{ or } \left(\frac{4C_1C_2}{3}\right)^{d+1}\geq \left(C_1\frac{d+1}{\ln(d+1)}\right)^{2}.
\end{equation*}
This holds for all $d\geq 10$, completing the proof (for all $d\geq 2$).
\end{proof}

\section{Remaining Proofs}

\subsection{Small Subgraph Conditioning Method}\label{sec:small-graph}

In this section, we give the outline for the proof of Lemma~\ref{lem:smallgraph}. The proof is a minor modification of the arguments in \cite[Appendices A \& B]{GSV:colorings} which were carried out for random $\Delta$-regular bipartite graphs. Here, we just need to account for the non-bipartite case which turns out to be completely analogous. For completeness, we give the adaptation of the calculations therein to account for the slightly different setting.

 The main tool we are going to use is the following Theorem, which is due to \cite{RW}.  The notation $[X]_{m}$ refers to the $m$-th order falling factorial of the variable $X$.

\begin{theorem}\label{thm:smallgraphmethod}
For  $i=1,2,\hdots$, let $\lambda_i>0$ and $\delta_i>-1$ be constants  and assume that for each $n$ there are random variables $X_{in}$,  $i=1,2,\hdots,$ and $Y_n$, all defined on the same probability space $\G=\G_n$ such that $X_{in}$ is non-negative integer valued, $Y_n\geq 0$ and $\E\big[Y_n\big]>0$ (for $n$ sufficiently large). Furthermore,   the following hold:
\begin{enumerate}
\item [\emph{(A1)}]$X_{in}\stackrel{d}{\longrightarrow}Z_i$ as $n\rightarrow\infty$, jointly for all $i$, where $Z_i\sim\mathrm{Po}(\lambda_i)$ are independent Poisson random variables;\label{it:poissonconv}
\item [\emph{(A2)}] for every finite sequence $j_1,\hdots,j_m$ of non-negative integers,
\begin{equation}
\frac{\E_{\G}\big[Y_n[X_{1n}]_{j_1}\cdots [X_{mn}]_{j_m}\big]}{\E_{\G}\big[Y_n\big]}\rightarrow \prod^m_{i=1}\big(\lambda_i(1+\delta_i)\big)^{j_i} \quad \text{ as } n\rightarrow\infty;\label{eq:prodsmall}
\end{equation}
\item [\emph{(A3)}] $\sum_i \lambda_i\delta_i^2<\infty$;
\item [\emph{(A4)}] $\E_{\G}\big[Y_n^2\big]/\big(\E_{\G}[Y_n]\big)^2\leq\exp\big(\sum_i \lambda_i \delta_i^2\big)+o(1)$ as $n\rightarrow \infty$;
\end{enumerate}
Let $r(n)$ be a function such that $r(n)\rightarrow 0$ as $n\rightarrow \infty$. It holds that $Y_n>r(n)\E_{\G}\big[Y_n\big]$ asymptotically almost surely.
\end{theorem}

To obtain Lemma~\ref{lem:smallgraph}, we verify the assumptions of Theorem~\ref{thm:smallgraphmethod} for the random variables $Z_G^{\alphab}$. Recall, we restrict our attention to $\alphab$ which are Hessian dominant. For $G\sim \Gc(n,\Delta)$, let $X_i=X_{in}$ be the number of cycles of length $i$ in $G$, $i=1,2,\hdots$.

The most technical part of this verification is assumption (A4) which requires computing the precise asymptotics of the moments. This in turn reduces to certain determinants which are not completely trivial. Nevertheless, the arguments have been carried out in full generality in \cite{GSV:colorings}. The only minor modification required in the present case is to account for random $\Delta$-regular graphs instead of the bipartite random $\Delta$-regular graphs studied in \cite{GSV:colorings}.

We obtain the following lemmas.

\begin{lemma}\label{lem:a1}
Assumption (A1) holds with $\lambda_i=\frac{(\Delta-1)^{i}}{2i}$.
\end{lemma}
Lemma~\ref{lem:a1} is well-known, see for example \cite{Janson}.

\begin{lemma}\label{lem:a2}
Assumption (A2) holds with $\delta_i=\sum^{q-1}_{j=1}\mu_j^{i}$, where $\mu_1,\mu_2,\hdots,\mu_{q-1}$ are the eigenvalues different than 1 of the matrix $\M$ defined in Section~\ref{sec:Potts}. For Hessian dominant $\alphab$, it holds that the $\mu_i$ are positive and strictly smaller than $1/(\Delta-1)$.
\end{lemma}
\begin{proof}[Proof of Lemma~\ref{lem:a2}]
The proof is close to \cite[Proof of Lemma A.6]{GSV:colorings}, which is in turn close to \cite[Proof of Lemma 7.4]{MWW}. We just modify the approach to account for the distribution induced by the pairing model. We make the minor notation change from $X_i$ to $X_\ell$, i.e., for $\ell\geq 1$, $X_{\ell}$ denotes the number of cycles of length $\ell$ in $G$. We show that Assumption (A2) in Theorem~\ref{thm:smallgraphmethod} holds when $m=1$ and $j_1=1$, the extension to $m>1$ and arbitrary indices $j_1,\hdots,j_m$ follows by standard arguments, see for example \cite[Section 2]{Kemkes} for an exposition of  the argument in a very similar setting.

Let $\Sc=\{S_1,\hdots,S_q\}$ be a partition of $V$ such that
$|S_i|=\alpha_i n$ for all $i\in [q]$. Note that $\Sc$ induces a configuration $\sigma(\Sc)$ by setting, for every vertex $v\in V$, $\sigma(v)=i$
iff $v\in S_i$. Denote by $Y_{\Sc}$ the weight of the configuration $\sigma(\Sc)$. 

Fix a specific partition $\Sc$. By symmetry,
\begin{equation}
\frac{\E[Z^{\alphab}_GX_\ell]}{\E[Z^{\alphab}_G]}=\frac{\E[Y_{\Sc}
X_\ell]}{\E[Y_{\Sc}]}. \label{eq:smallratio}
\end{equation}
We decompose $X_\ell$ as follows:
\begin{itemize}
\item $\xi$ will denote a rooted and
oriented $\ell$-cycle, whose vertices are colored with $\{1,\hdots,q\}$ (note that the coloring is \emph{not} assumed to be proper). A vertex colored $i$ in $\xi$ signifies that it corresponds to a vertex in $S_i$.


\item Once we have specified $\xi$, we use $\zeta$ to denote the $\ell$ \emph{points} that the cycle  traverses in order, such that the prescription of the vertex colors of $\xi$ is satisfied. (Recall from Section~\ref{sec:defer-prelim} that points are elements of $[\Delta n]$; $\zeta$ specifies the pre-image of the cycle $\xi$ in the pairing model, i.e., the matched points that respect the colors $\xi$ of the cycle with respect to the partition $\Sc$).

\item $\mathbf{1}_{\xi,\zeta}$ is the indicator function whether the cycle
specified by $\xi,\zeta$ is present in the graph $G$ generated by the pairing model.
\end{itemize}
Once the vertex colors of a cycle have been specified, note that each possible cycle corresponds to exactly $2\ell$ different configurations $\xi$ (the number of ways to root and
orient the cycle). For each of those $\xi$, the respective sets of
configurations $\zeta$ are the same. Hence, we may write
\[X_\ell=\frac{1}{2\ell}\sum_\xi\sum_\zeta \mathbf{1}_{\xi,\zeta}.\]
Let $p_1:= \Pr[\mathbf{1}_{\xi,\zeta}=1]$. It follows that
\begin{align*}
\E[Y_{\Sc}X_\ell]&=\frac{1}{2\ell}\sum_\xi\sum_\zeta
p_1\cdot\E[Y_{\Sc}|\mathbf{1}_{\xi,\zeta}=1].
\end{align*}
In light of \eqref{eq:smallratio}, we need to study the ratio
$\E[Y_{\Sc}\mid\mathbf{1}_{\xi,\zeta}=1]/\E[Y_{\Sc}]$. At
this point, to simplify notation, we may assume that $\xi,\zeta$
are fixed.

We have shown in Section~\ref{sec:defer-prelim} that
\begin{equation}\label{eq:ysimple}
\begin{aligned}
&\E[Y_{\Sc}]=\\
&\sum_{\x}\prod_{i}\binom{\Delta\alpha_{i}n}{\Delta x_{i1}n,\hdots,\Delta x_{iq}n}\frac{\big[\prod_{i\neq j}(\Delta x_{ij}n)!\big]^{1/2}\prod_{i}(\Delta x_{ii}n-1)!!}{(\Delta n-1)!!}\prod_{i,j}B^{\Delta x_{ij}n/2}_{ij},
\end{aligned}
\end{equation}
where the variables $\x=(x_{11},\hdots,x_{qq})$ capture the number
of edges between the different color classes in $\Sc$. In particular, for $i\neq j$, 
$\Delta x_{ij}n$ is the number of edges between the sets $S_i$ and $S_j$, whereas $\Delta x_{ii}n/2$ is the number of edges within the set $S_i$ (cf. Section~\ref{sec:defer-prelim} for more details).

To calculate $\E[Y_{\Sc}\mid\mathbf{1}_{\xi,\zeta}=1]$, we need some
notation. For colors $i,j\in\{1,\hdots,q\}$, let $a'_{ij}$ be the number of edges in $\xi$ whose one endpoint has color $i$ and the other $j$. It will be convenient to denote $a_{ii}:=2a'_{ii}$ and $a_{ij}:=a'_{ij}$ whenever $i\neq j$. Finally, let $c_i$ denote the number of vertices in $\xi$ colored with $i$. The following equalities are immediate:
\begin{equation}\label{eq:degreeconsideration}
\mbox{$\sum_{j}$}\,a_{ij}=2c_i, \
\mbox{$\sum_{i,j}$}\, a_{ij}=2\ell.
\end{equation}
We are almost set to compute
$\E[Y_{\Sc}\mid\mathbf{1}_{\xi,\zeta}=1]$. We denote by $\x$ the
same set of variables as in \eqref{eq:ysimple}. This
number includes the $a_{ij}$ edges prescribed by $\xi,\zeta$.
To make the following formulas easier to digest let
$n\Delta x'_{ij}=n\Delta x_{ij}-a_{ij}$.
We have
\begin{equation*}
\begin{aligned}
&\E[Y_{\Sc}\mid \mathbf{1}_{\xi,\zeta}=1]=\\
&\sum_{\x}\prod_{i}\binom{\Delta\alpha_{i}n-2c_i}{\Delta x_{i1}'n,\hdots,\Delta x_{iq}'n}\frac{\big[\prod_{i\neq j}(\Delta x_{ij}'n)!\big]^{1/2}\prod_{i}(\Delta x_{ii}'n-1)!!}{(\Delta n-2\ell-1)!!}\prod_{i,j}B^{\Delta x_{ij}n/2}_{ij}.
\end{aligned}
\end{equation*}

Using that for constants $c_1,c_2>0$, it holds that $(c_1n-c_2)!/(c_1n)!=(1+o(1))(c_1n)^{-c_2}$, we obtain
\begin{equation*}
\frac{\binom{\Delta\alpha_{i}n-2c_i}{\Delta x_{i1}'n,\hdots,\Delta x_{iq}'n}}{\binom{\Delta\alpha_{i}n}{\Delta x_{i1}n,\hdots,\Delta x_{iq}n}} 
\sim \frac{\prod_j\big(x_{ij}\big)^{a_{ij}}}{\alpha_i^{2c_i}},
\quad \frac{\frac{\big[\prod_{i\neq j}(\Delta x_{ij}'n)!\big]^{1/2}\prod_{i}(\Delta x_{ii}'n-1)!!}{(\Delta n-2\ell-1)!!}}{\frac{\big[\prod_{i\neq j}(\Delta x_{ij}n)!\big]^{1/2}\prod_{i}(\Delta x_{ii}n-1)!!}{(\Delta n-1)!!}}\sim \frac{1}{\prod_{i,j}(x_{ij})^{a_{ij}/2}}.
\end{equation*}
The asymptotics of the ratio $\E[Y_{\Sc}\mid \mathbf{1}_{\xi,\zeta}=1]/\E[Y_{\Sc}]$ are determined from those $\x^*$ which maximize $\Upsilon_1(\alphab,\x)$. Thus, we obtain
\begin{equation*}
\frac{\E[Y_{\Sc}\mid \mathbf{1}_{\xi,\zeta}=1]}{\E[Y_{\Sc}]}
\sim \frac {\prod_{i,j}\big(x_{ij}^{*}\big)^{a_{ij}/2}}
{\prod_i\alpha_i^{2c_i}}.
\end{equation*}
 For given $\xi$, the number of
possible $\zeta$ in the pairing model is asymptotic to $\prod_i\big[\Delta(\Delta-1)\alpha_in\big]^{c_i}=\big[\Delta(\Delta-1)n\big]^\ell\prod_i\alpha_i^{c_i}$. Since $p_1=(\Delta n-2\ell-1)!!/(\Delta n-1)!!\sim (\Delta n)^{-\ell}$, we have
\begin{equation*}
\begin{aligned}
\frac{\sum_\zeta
p_1\E[Y_{\Sc}\mid \mathbf{1}_{\xi,\zeta}=1]}{\E[Y_{\Sc}]}\sim
\frac
{(\Delta-1)^{\ell}\prod_i\alpha^{c_i}_i\prod_{i,j}\big(x_{ij}^{*}\big)^{a_{ij}/2}}
{\prod_i\alpha_i^{2c_i}}
&=\frac
{(\Delta-1)^{\ell}\prod_{i,j}\big(x_{ij}^{*}\big)^{a_{ij}/2}}
{\prod_i\alpha_i^{\frac{1}{2}\sum_j a_{ij}}}\\
&=(\Delta-1)^{\ell}\prod_{i\leq j}\Big(\frac{x_{ij}^{*}}{\sqrt{\alpha_i\alpha_j}}\Big)^{a_{ij}'},
\end{aligned}
\end{equation*}
where in the last equality we used that $a_{ii}=2a_{ii}'$, $a_{ij}=a'_{ij}$ for $i\neq j$, $a_{ij}'=a_{ji}'$ and $x_{ij}^{*}=x_{ji}^{*}$. Note that the r.h.s. evaluates to 0 whenever there exist $i,j$ such that $B_{ij}=0$ but $a_{ij}\neq 0$, since then we have $x_{ij}^{*}=0$ (cf. \eqref{eq:optimalxs}. This is in complete accordance with the fact that the configuration induced by the partition $\Sc$ has zero weight. Thus, by \eqref{eq:smallratio}, we may write
\begin{equation*}
\frac{\E[Z^{\alphab}_GX_\ell]}{\E[Z^{\alphab}_G]}\sim\frac{(\Delta-1)^{\ell}}{2\ell}\cdot
\sum_{\a'} N_{\a'}
\prod_{i\leq j}\Big(\frac{x_{ij}^{*}}{\sqrt{\alpha_i\alpha_j}}\Big)^{a_{ij}'},
\end{equation*}
where $\a'=\{a_{11}',\hdots,a_{qq}'\}$ and $N_{\a'}$ is the number of
possible $\xi$ with $a_{ij}'$ edges of type $\{i,j\}$. Using \eqref{eq:optimalxs}, we have that $x^{*}_{ij}/\sqrt{\alpha_i\alpha_j}$ is equal to the $(i,j)$-entry of the matrix $\M$. Thus, the sum can be reformulated as the (multiplicative) weight of walks in a weighted multigraph whose (weighted) adjacency matrix is given by $\M$ (for more details on the technique see \cite{Janson}). It thus follows that the sum equals $\mathrm{Tr}(\M^\ell)=1+\sum^{q-1}_{j=1}\mu_j^\ell$. The fact that the $\mu_i$'s are positive follows from the fact that $\B$ is a positive definite matrix, while the fact that the $\mu_i$'s are less than $1/(\Delta-1)$ follows from the results of Section~\ref{sec:connection}.
\end{proof}

\begin{lemma}\label{lem:a3}
Assumption (A3) holds with
\[\exp(\sum_{i\geq 1} \lambda_i\delta_i^2)=\prod^{q-1}_{i=1}\prod^{q-1}_{j=1}\big(1-(\Delta-1)\mu_i\mu_j\big)^{-1/2},\] where the $\mu_i$'s are as in Lemma~\ref{lem:a2}.
\end{lemma}
\begin{proof}[Proof of Lemma~\ref{lem:a3}]
We have
\begin{align*}
\sum_{i\geq 1} \lambda_i\delta_i^2=\sum_{i\geq 1}\frac{(\Delta-1)^{i}}{2i}\sum^{q-1}_{j=1}\sum^{q-1}_{k=1}\mu_j^{i}\mu_k^{i}=-\frac{1}{2}\sum_{j,k\in[q-1]}\ln\big(1-(\Delta-1)\mu_j\mu_k\big),
\end{align*}
where we used that $\sum_{i\geq 1} \frac{x^{i}}{i}=-\ln(1-x)$ for $|x|<1$.
\end{proof}

\begin{lemma}\label{lem:a4}
For a ferromagnetic model, for all Hessian dominant $\alphab$ it holds that
\[\frac{\E_{\G}\big[(Z^{\alphab}_G)^2\big]}{\big(\E_{\G}[Z^{\alphab}_G]\big)^2}\rightarrow\prod^{q-1}_{i=1}\prod^{q-1}_{j=1}\big(1-(\Delta-1)\mu_i\mu_j\big)^{-1/2},\]
where the $\mu_i$'s are as in Lemma~\ref{lem:a2}.
\end{lemma}
\begin{proof}
Let $\x^{*}=\arg\max_{\x}\Upsilon_1(\alphab,\x)$, $(\gammab^{*},\y^{*})=\arg\max_{\gammab,\y}\Upsilon_2(\gammab,\y)$. For $\alphab$ which is Hessian dominant, Theorem~\ref{thm:copo} yields  $\Upsilon_2(\gammab^{*},\y^{*})=2\Upsilon_1(\alphab,\x^{*})$. 

Using methods in \cite[Appendix B]{GSV:colorings}, we show in Section~\ref{sec:linearalgebra} that 
\begin{equation}\label{eq:asymptoticsfirst}
\lim_{n\rightarrow\infty}\frac{(2\pi n)^{(q-1)/2}\E_{\G}[Z^{\alphab}_G]}{e^{n\Upsilon_1(\alphab,\x^{*})}}=\Big(\prod_{i\in [q]}\alpha_i\prod_{i\in [q-1]}\big(1+\mu_i\big)\Big)^{-1/2},
\end{equation}
and
\begin{equation}\label{eq:asymptoticssecond}
\lim_{n\rightarrow\infty}\frac{(2\pi
n)^{q-1}\E_{\G}[(Z^{\alphab}_G)^2]}{e^{n\Upsilon_2(\gammab^{*},\y^{*})}}
=\Big(\prod_{i\in [q]}\alpha_i\prod_{i\in [q-1]}\big(1+\mu_i\big)\Big)^{-1}\prod^{q-1}_{i=1}\prod^{q-1}_{j=1}\big(1-(\Delta-1)\mu_i\mu_j\big)^{-1/2}.
\end{equation} 
The statement of the lemma follows.
\end{proof}

\begin{proof}[Proof of Lemma~\ref{lem:smallgraph}]
Lemmas~\ref{lem:a1}--\ref{lem:a4} verify the assumptions of Theorem~\ref{thm:smallgraphmethod}. The lemma thus follows by applying Theorem~\ref{thm:smallgraphmethod}, for $r(n)=1/n$.
\end{proof}

\subsection{The asymptotics of the moments}\label{sec:linearalgebra}
We follow closely the proof in \cite[Appendix B]{GSV:colorings}, where very similar asymptotics are computed in detail. We first overview the approach in \cite[Appendix B]{GSV:colorings} in our setting.

The asymptotics of $\E_{\G}[Z^{\alphab}_G]$ and $\E_{\G}[Z^{\alphab}_G]^2$ are derived by first rewriting the sums in \eqref{eq:firstmoment} and \eqref{eq:secondmoment} as integrals and approximating the latter with Gaussian integrals. The principle behind the technique is the negative-definiteness of the Hessian at the maximizers of the functions $\Upsilon_1$ and $\Upsilon_2$, which control the exponential order of the terms in the sums \eqref{eq:firstmoment} and \eqref{eq:secondmoment}, respectively. This allows to focus on terms within  $O(1/\sqrt{n})$ distance around the term with the maximum contribution. A thorough exposition of the technical details can be found in \cite[Section 9.4]{JLR}.

Carrying out the above scheme in our setting is impeded by the fact that the sums in \eqref{eq:firstmoment} and \eqref{eq:secondmoment} are over variables which are linearly dependent. We will get rid of this linear dependence in the simplest way: for each of the two sums, we pick a subset $S$ of the variables (with minimum cardinality) such that every variable is a (non-trivial) linear combination of variables in $S$. Variables in $S$ span a full-dimensional space, inducing what we call a ``full-dimensional representation"  of the functions $\Upsilon_1, \Upsilon_2$ when these are viewed as functions  of the variables in $S$.   The inconvenience that this procedure causes is, that in the calculation of the Gaussian integrals, (the determinant of) the Hessian matrices of the full-dimensional representations of $\Upsilon_1, \Upsilon_2$ come into play. 

In \cite[Appendix B.1.1]{GSV:colorings}, the above setting is abstracted as follows: given a linear subspace $\A \z=0$, compute (the determinant of) the Hessian matrix of the full-dimensional representation of a function $\Upsilon(\z)$. It is not hard to see that a full-dimensional representation of $\A \z=0$, assuming that $\A$ has row rank $r$, is obtained by first picking a submatrix $\A_f$ induced by $r$ linearly independent rows of $\A$, and then picking $r$ columns of $\A_f$ to obtain an $r\times r$ invertible submatrix $\A_{fs}$ (the variables corresponding to the columns of $\A_{fs}$ can be written as non-trivial linear combinations of the remaining variables; the latter yield the full-dimensional representation). We denote by $\H^{f}$ the Hessian matrix of the full-dimensional representation of $\Upsilon$ induced by the matrices $\A_f,\A_{fs}$. We further denote by $\H$ the Hessian matrix of $\Upsilon(z)$ (where $\z$ is now assumed to be unconstrained); note that $\H$ is diagonal. The following is proved in \cite{GSV:colorings}.

For a polynomial $p(s)$, $[s^t]p(s)$ denotes the coefficient of $s^t$ in $p(s)$.
\begin{lemma}{\cite[Lemma B.7]{GSV:colorings}}\label{lem:blackbox}
Suppose $\A$ consists of $m$ rows and has rank $r$. Let $\Tb$ be a positive semi-definite diagonal matrix  such that $[\Tb\ \A]$ has full row rank.  If $\H$ is invertible, then
\begin{equation}\label{eq:blackbox}
\Det\big(-\H^f\big)=\frac{L\big(\A_f,\A,\Tb\big)}{\Det\big(\A_{fs}\big)^2}\,\Det(-\H)\, [\epsilon^{m-r}]\,\Det\big(\epsilon\Tb-\A\H^{-1}\A^{\T}\big),
\end{equation}
where $L\big(\A_{f}, \A,\Tb\big)=(-1)^r\left.\Det\big(\A_{f}\A_{f}^\T\big)\middle/[\epsilon^{m-r}]\,\Det\big(\epsilon\Tb-\A\A^{\T}\big)\right.$.
\end{lemma}
\begin{remark}\label{rem:fullrowrank}
When $\A$ has full row rank, i.e., $r=m$, one can take $\Tb$ to be the identity matrix. Then, the r.h.s. in \eqref{eq:blackbox} simplifies into  \[\Det(-\H)\,\Det\big(-\A\H^{-1}\A^{\T}\big)/\Det\big(\A_{fs}\big)^2.\] 
\end{remark}
We will apply Lemma~\ref{lem:blackbox} in Sections~\ref{sec:momentone} and~\ref{sec:momenttwo} to calculate the asymptotics of the moments. To do this, we will need more information on the maximizers of the functions $\Upsilon_1$ and $\Upsilon_2$ for a Hessian dominant phase $\alphab$. In particular, let $\x^*$ be the maximizer of $\Upsilon_1(\alphab,\x)$ (cf. \eqref{eq:optimalxs}) and $(\gammab^*,\y^*)$ be the maximizer of $\Upsilon_2(\gammab,\y)$. Adapting the proof of \cite[Lemma 3.2]{GSV:colorings}, we have that $\gamma^{*}_{ij}=\alpha_i\alpha_j$ for all $i,j\in[q]$, from where it easily follows that $y^{*}_{ikjl}=x^{*}_{ij}x^{*}_{kl}$.

Following \cite[Appendix B.1.2]{GSV:colorings}, we use the following notation in Sections~\ref{sec:momentone} and~\ref{sec:momenttwo}. For a vector $\z\in \mathbb{R}^n$ we denote by $\z^D$ the $n\times n$ diagonal matrix whose $i$-entry on the diagonal equals $z_i$ for $i\in n$. For vectors $\z_j\in \mathbb{R}^{m_j}$, $j=1,\hdots,t$, we denote by $[\z_1,\hdots,\z_{t}]^{\T}$ the $\mathbb{R}^{\sum_j m_j}$ vector which is the concatenation of the vectors $\z_1,\hdots,\z_t$. $\I_{n}$ will denote  the identity matrix with dimensions $n\times n$ and $\mathbf{0}$ will denote the all-zeros matrix whose dimensions will be inferred from context.  For matrices $\M_1$ and $\M_2$, $\M_1\otimes \M_2$ will denote the Kronecker product of $\M_1,\M_2$, while  $\M_1\oplus\M_2$ denotes the direct sum of $\M_1,\M_2$, that is, the block matrix $[\begin{smallmatrix} \M_1& \mathbf{0}\\\mathbf{0}&\M_2 \end{smallmatrix}]$. 

\subsubsection{Proof of \eqref{eq:asymptoticsfirst}}\label{sec:momentone}
The first moment $\E_{\G}[Z^{\alphab}_G]$ is a sum over $\x$ (and $\alphab$ is fixed).  Note that if $B_{ij}=0$, then we may restrict the sum in \eqref{eq:firstmoment} to those $\x$ which satisfy $x_{ij}=0$ without changing the sum's value. Further, since $x_{ij}=x_{ji}$ and $B_{ij}=B_{ji}$, we may write the sum in \eqref{eq:firstmoment} in terms of those $x_{ij}$ with $i\leq j$. Let
\begin{equation}
P_1=\big\{(i,j)\in[q]^2\,\big|\, B_{ij}> 0, i\leq j\big\}.\label{eq:defpone}
\end{equation}
Henceforth $\x$ will denote $\{x_{ij}\}_{(i,j)\in P_1}$. The observations above imply that the sum in \eqref{eq:firstmoment} can be written over the possible values of the vector $\x$.  Note that for ferromagnetic models we have that $(i,i)\in P_1$ for every $i\in [q]$ (cf. the discussion after Theorem~\ref{thm:prediction}).   We are left to account for the linear dependencies induced by the $q$ constraints $a_i=\sum_{j}x_{ij}$. In matrix form, we can write those as
\begin{equation}\label{eq:fullfirst}
\A_{1}\x=\alphab,
\end{equation}
where $\A_{1}$ is a $\{0,1\}$-matrix with dimension $q\times |P_1|$. For a ferromagnetic model, as we shall display shortly, we have that the rank of $\A_1$ is $q$ (this holds more generally for matrices $\B$ which are irreducible and aperiodic,  see for example footnote~\ref{foot:invertible}). To get a full-dimensional representation of the space \eqref{eq:fullfirst}, we will eliminate $q$ variables from the vector $\x$. This corresponds to picking $q$ columns of $\A_1$ which induce a $q\times q$ invertible submatrix of $\A_1$. We will denote this submatrix by $\A_{1,s}$. For ferromagnetic models, we can choose the columns corresponding to the variables $x_{ii}$ for $i\in[q]$, in which case  $\A_{1,s}$ is simply the identity matrix.\footnote{\label{foot:invertible}In general, the $q\times q$ invertible submatrices $\A_{1,s}$ can be characterized as follows. First, view $\A_1$ as the (unsigned) incidence matrix of a graph $H$ with vertex set $[q]$ and edge set $P_1$, where vertex $i$ corresponds to the $i$-th row of $\A_1$ and an edge labelled $(i,j)$ corresponds to the column labelled $(i,j)$ in $\A_1$ (note that $H$ has a self-loop on vertex $i$ iff $B_{ii}>0$). Then $\A_{1,s}$ specifies a subgraph $H'$ of $H$ with exactly $q$ edges. It can be shown that $\A_{1,s}$ is invertible if $H'$ is spanning (i.e, every vertex in $H'$ has non-zero degree) and  all of the connected components of $H'$ are unicyclic and non-bipartite (i.e., every connected component of $H'$ has a unique cycle of odd length, where self-loops count as cycles of length 1).} 

Adapting the proof of \cite[Lemma B.3]{GSV:colorings}, yields the following asymptotics for the first moment $\E_{\G}[Z^{\alphab}_G]$. The details of the proof can be found in Appendix~\ref{app:A}.
\begin{lemma}\label{lem:firasympdet}
For a ferromagnetic model,\footnote{\label{foot:det}We briefly comment on how the choice of the full-dimensional representation (i.e, the choice of $\A_{1,s}$) has been  used in the derivation of \eqref{eq:asymfirst}. Relative to footnote \ref{foot:invertible}, if the invertible submatrix $\A_{1,s}$ corresponds to a subgraph with exactly $c$ components which contain a non-trivial odd cycle (i.e., an odd cycle  of length $\geq 3$), there is a correction factor $2^{-c}$ in the r.h.s. of \eqref{eq:hfone}. The factor comes from $(\mathrm{mod}\, 2)$ constraints imposed by considering the sum of constraints in \eqref{eq:fullfirst} corresponding to vertices in each such unicyclic component (in the derivation below, this factor  cancels with  the factor $|\mathrm{Det}(\A_{1,s})|$ coming from $\mathrm{Det}(-\H_{1,\x}^f)^{-1/2}$; it can be shown that $|\mathrm{Det}(\A_{1,s})|=2^c$). Note that for our choice of $\A_{1,s}$, $c$ equals zero, since the subgraph induced by the columns of $\A_{1,s}$ consists of $q$ components, each of which is a single vertex with a self-loop.} it holds that
\begin{equation}\label{eq:asymfirst}
\lim_{n\rightarrow\infty}\frac{(2\pi n)^{(q-1)/2}\E_{\G}[Z^{\alphab}_G]}{e^{n\Upsilon_1(\alphab,\x^{*})}}=\Big( 2^{q-1}  \prod_{(i,j)\in P_1} x_{ij}^{*}\Big)^{-1/2}\mathrm{Det}(-\H_{1,\x}^f)^{-1/2},
\end{equation} 
where $\H_{1,\x}^f$ is the Hessian of the full-dimensional representation of $g_1(\x)$ evaluated at $\x=\x^{*}$.
\end{lemma}

To expand the determinant in \eqref{eq:asymfirst}, we apply Lemma~\ref{lem:blackbox} and in particular Remark~\ref{rem:fullrowrank}. This yields
\begin{equation}\label{eq:hfone}
\begin{aligned}
\Det\big(-\H^f_{1,\x}\big)&=\frac{1}{\Det(\A_{1,s})^2}\,\Det(-\H_{1,\x})\,\Det\big(-\A_1\H_{1,\x}^{-1}\A^{\T}_1\big)\\
&=\Det(-\H_{1,\x})\,\Det\big(-\A_1\H_{1,\x}^{-1}\A^{\T}_1\big),
\end{aligned}
\end{equation}
where $\H_{1,\x}$ is the $|P_1|\times |P_1|$ diagonal matrix corresponding to the Hessian matrix of $g_1(\x)$ (when $\x$ is unconstrained) and in the second equality in \eqref{eq:hfone} we used that $\Det(\A_{1,s})=1$ (by our choice of $\A_{1,s}$). 

Since $\H_{1,\x}$ is diagonal, we obtain that
\begin{equation*}
\Det(-\H_{1,\X})^{-1}=2^{q}\prod_{(i,j)\in P_1}x_{ij}^{*},
\end{equation*}
so \eqref{eq:asymptoticsfirst} will follow from 
\begin{equation}\label{eq:firstXXX}
\Det\big(-\A_1(\H_{1,\X})^{-1}\A_1^{\T} \big)=2\prod_{i\in [q]}\alpha_i\prod_{i\in [q-1]}\big(1+\mu_i\big).
\end{equation}
To show \eqref{eq:firstXXX}, it can be checked that 
\begin{equation}\label{eq:weightfir}
-\A_1\H_{1,\x}^{-1}\A^{\T}=\alphab^{D}+\S_{\x},
\end{equation}
where $\alphab^{D}$ is the $q\times q$ diagonal matrix whose $i$-th diagonal entry is $\alpha_i$ and $\S_{\x}$ is the $q\times q$ symmetric matrix whose $(i,j)$ entry (when $i\leq j$) is $x_{ij}^{*}$ whenever $(i,j)\in P_1$ and 0 otherwise.

Observe that $(\alphab^{D})^{-1/2}\S_{\x}(\alphab^{D})^{-1/2}=\M$, where $\M$ is the matrix in Lemma~\ref{lem:a2}. From this, we obtain
\begin{equation*}
\Det(-\A_1\H_{1,\x}^{-1}\A^{\T}_1)=\Big(\prod_{i\in [q]}\alpha_i\Big)\, \Det\big(\I_q+\M\big).
\end{equation*}
Since the spectrum of $\M$ is $\{1,\mu_1,\hdots,\mu_q\}$, it follows that the spectrum of the matrix $\I_q+\M$ is $\{2,1+ \mu_1, \hdots,1+ \mu_{q-1}\}$. This yields \eqref{eq:firstXXX}, thus completing the proof of \eqref{eq:asymptoticsfirst}.

\subsubsection{Proof of \eqref{eq:asymptoticssecond}}\label{sec:momenttwo}
For the second moment, $\E_{\G}\big[(Z^{\alphab}_G)^2\big]$ is a sum over $\gammab,\y$ while $\alphab$ is fixed. Analogously to \eqref{eq:defpone}, let
\begin{equation}
P_2=\big\{(i,k,j,l)\in [q]^4\,\big|\, B_{ij}B_{kl}>0, i\leq j, k\leq l\big.\big\}.\label{eq:defptwo}
\end{equation}
Henceforth, $\y$ will denote $\{y_{ikjl}\}_{(i,k,j,l)\in P_2}$. The constraints in \eqref{eq:constraintsecond} can now be written as 
\begin{equation}\label{eq:fullsecond}
\A_2\,\big[\gammab,\, \Y\big]^{\T}=\big[\alphab,\,\alphab,\,\mathbf{0}_{q^2}\big]^{\T}, \mbox{ where } \A_2=\Big[\begin{array}{cc} \A_{2,\gammab}& \mathbf{0}\\ -\I_{q^2} & \A_{2,\Y}\end{array}\Big],
\end{equation}
and $\A_{2,\gammab}, \A_{2,\Y}$ are $\{0,1\}$-matrices with dimensions $2q\times q^2$ and $q^2\times |P_2|$, respectively. It is easy to see that $\A_{2,\Y}$ has full row rank $r_{\y}=q^2$, while $\A_{2,\gammab}$ has rank $r_{\gammab}=2q-1$, so that the rank of $\A_2$ is $r_2=r_{\y}+r_{\gammab}=q^2+2q-1$. Thus, to specify a full-dimensional representation of \eqref{eq:fullsecond}, we need to specify an  $r_2\times r_2$ invertible submatrix $\A_{2,s}$ of $\A_2$. It can be checked that any such submatrix $\A_{2,s}$ of $\A_2$ must have the form
\[\A_{2,s}=\Big[\begin{array}{cc} \A^s_{2,\gammab}& \mathbf{0}\\ -\I_{q^2} & \A^s_{2,\Y}\end{array}\Big],\]
where $\A^s_{2,\gammab},\A^s_{2,\Y}$ are $r_{\gammab}\times r_{\gammab}$ and $r_{\y}\times r_{\y}$  invertible submatrices of $\A^s_{2,\gammab},\A^s_{2,\y}$ respectively. Thus, we only need to specify the matrices $\A^s_{2,\gammab},\A^s_{2,\y}$. We will choose $\A_{2,\gammab}^s$ to be an arbitrary invertible submatrix of $\A_{2,\gammab}$; since $\A_{2,\gammab}$ is totally unimodular (it corresponds to the incidence matrix of the complete bipartite graph with $q$ vertices on each side), we have $\Det(\A^s_{2,\gammab})^2=1$. For ferromagnetic models, we can choose $\A^s_{2,\y}$ to be the identity matrix using the columns corresponding to  variables $y_{ikik}$ with $i,k\in[q]$, so $\Det(\A^s_{2,\y})^2=1$. It follows that $\Det(\A_{2,s})^2=1$. For future use (with the scope of applying Lemma~\ref{lem:blackbox}), let $\A_{2,f}$ be the submatrix of $\A_2$ induced by the rows corresponding to rows of $\A_{2,s}$.

We have the following analogue of Lemma~\ref{lem:firasympdet}.  The proof is given in Appendix~\ref{app:A}.
\begin{lemma}\label{lem:secasympdet}
For a ferromagnetic model, it holds that
\begin{equation}\label{eq:asymsecond}
\lim_{n\rightarrow\infty}
\frac{(2\pi
n)^{q-1}\E_{\G}[(Z^{\alphab}_G)^2]}{e^{n\Upsilon_2(\gammab,\y^{*})}}=
\Big(2^{q^2-1}  \prod_{(i,k,j,l)\in P_2} y_{ikjl}^{*}\Big)^{-1/2} \Delta^{-(q-1)^2/2}\mathrm{Det}(-\H_{2}^f)^{-1/2},
\end{equation}
where $\H_2^f$ is the Hessian of the full-dimensional representation of $\Upsilon_2(\gammab,\y)/\Delta$ evaluated at $(\gammab,\y)=(\gammab^*,\y^*)$.
\end{lemma}
Denote by $\H_2$ the diagonal matrix corresponding to the Hessian matrix of $\Upsilon_2(\gammab,\y)/\Delta$ (when $\gammab,\y$ are unconstrained). Note that we may decompose $\H_2$ as $[\begin{smallmatrix} \H_{2,\gammab}& \mathbf{0}\\ \mathbf{0}& \H_{2,\y}\end{smallmatrix}]$, where $\H_{2,\gammab}$ is the $q^2\times q^2$ diagonal matrix corresponding to the Hessian matrix of $(\Delta-1)f_2(\gammab)/\Delta$ and $\H_{2,\y}$ is the $|P_2|\times |P_2|$ diagonal matrix corresponding to the Hessian matrix of $g_2(\y)$ (see \eqref{eq:limitsecond} for the specification of  the functions $f_2$ and $g_2$).

We next apply Lemma \ref{lem:blackbox} with the matrix (recall that $\alphab^{D}$ is the $q\times q$ diagonal matrix whose $i$-th diagonal entry is $\alpha_i$ and $\mathbf{0}_{q^2}$ is the $q^2\times q^2$ all-zeros matrix):
\begin{equation}\label{eq:matricesT}
\Tb_{2}=\alphab^{D}\oplus \alphab^{D}\oplus \mathbf{0}_{q^2},
\end{equation} 
to obtain the following equality:
\begin{align*}
\Det\big(-\H^f_2\big)&=\frac{L(\A_2,\A_{2,f},\Tb_2)}{\mathrm{Det}(\A_{2,s})^2}\Det(-\H_{2})\, [\epsilon]\,\Det\big(\epsilon\Tb_{2}-\A_2\H_{2}^{-1}\A^{\T}_2\big)\\
&=\frac{1}{2}\Det(-\H_{2})\, [\epsilon]\,\Det\big(\epsilon\Tb_{2}-\A_2\H_{2}^{-1}\A^{\T}_2\big),
\end{align*}
where in the latter equality we used that $\mathrm{Det}(\A_{2,s})^2=1$ (which was proved earlier) and $L(\A_2,\A_{2,f},\Tb_2)=1/2$ (follows by \cite[Proof of Lemma B.8]{GSV:colorings}). We calculate
\begin{align*}
\Det(-\H_{2})^{-1}&=\Det(-\H_{2,\gammab})^{-1}\, \Det(-\H_{2,\y})^{-1}\\
&=\Big[(-1)^{q^2}\Big(\frac{\Delta}{\Delta-1}\Big)^{q^2}\prod_{i,k\in[q]}\gamma_{ik}^{*}\Big]\Big[2^{q^2}\prod_{(i,k,j,l)\in P_2}y_{ikjl}^{*}\Big],
\end{align*}
so that \eqref{eq:asymptoticssecond} will follow from
\begin{multline}\label{eq:detH2}
\frac{[\epsilon]\Det\big(\epsilon\Tb_2-\A_{2} (\H_{2})^{-1}\A_{2}^{\T}\big)}{\displaystyle\prod_{i\in[q]}\alpha_i^2\displaystyle\displaystyle\prod_{i,k\in[q]}\gamma_{ik}^{*}}
\\=(-1)^{q^2}\frac{4\Delta^{2q-1}}{(\Delta-1)^{q^2}}\prod_{i\in[q-1]}\big(1+\mu_i\big)^2\prod_{i,j\in[q-1]}\big(1-(\Delta-1)\mu_i\mu_j\big).
\end{multline}

We first write out the block structure of $\epsilon\Tb_2-\A_{2} (\H_{2})^{-1}\A_{2}^{\T}$. First, we have the following analogue of \eqref{eq:weightfir}:
\begin{gather}
-\A_{2,\Y}\H_{2,\y}^{-1}\A_{2,\Y}^{\T}=\gammab^{D}+ \Sb_{\Y},\quad \A_{2,\gammab}\gammab^D\A_{2,\gammab}^{\T}=\Big[\begin{array}{cc} \alphab^{D}& \Sb_{\gammab} \\ \Sb_{\gammab}& \alphab^{D}\end{array}\Big],\label{eq:almostlaplacians1}
\end{gather}
where $\Sb_{\gammab}$ is the $q\times q$ matrix whose $(i,j)$ entry is $\gamma_{ij}^{*}$ and $\Sb_{\Y}$ is the $q^2\times q^2$ matrix whose $((i,k),(j,l))$ entry is $y_{ikjl}^{*}$. From 
\begin{align*}
&\epsilon \Tb_2-\A_2(\H_2)^{-1}\A_2^{\T}\\
&=\epsilon \Tb_2+\Big[\begin{array}{cc} \A_{2,\gammab}& \mathbf{0}\\ -\I_{q^2} & \A_{2,\Y}\end{array}\Big]\Big[\begin{array}{cc} -\frac{\Delta}{\Delta-1}\gammab^{D}& \mathbf{0}\\ \mathbf{0} & -\H_{2,\y}^{-1}\end{array}\Big]\Big[\begin{array}{cc} \A_{2,\gammab}^{\T}& -\I_{q^2}\\ \mathbf{0} & \A_{2,\Y}^{\T}\end{array}\Big]\\
&=\frac{\Delta}{\Delta-1}\Big[\begin{array}{cc} \epsilon \frac{\Delta-1}{\Delta}(\alphab^D\oplus \alphab^D)-\A_{2,\gammab}\gammab^{D}\A_{2,\gammab}^{\T}& \A_{2,\Y}\gammab^{D}\\ \gammab^{D}\A_{2,\Y}^{\T} & -\gammab^{D}-\frac{\Delta-1}{\Delta}\A_{2,\Y}\H_{2,\y}^{-1}\A_{2,\Y}^{\T}\end{array}\Big]
\end{align*}
we obtain that:
\begin{equation}\label{eq:secondWHmoment}
\frac{[\epsilon]\Det\big(\epsilon\Tb_2-\A_{2} (\H_{2})^{-1}\A_{2}^{\T}\big)}{\displaystyle\prod_{i\in[q]}\alpha^2_i\displaystyle\displaystyle\prod_{i,k\in[q]}\gamma_{ik}^{*}}=[\epsilon]\Det\big(\H_2'\big),
\end{equation}
where
\begin{equation*}
\H_2'=\frac{\Delta}{\Delta-1}\Big[\begin{array}{cc}
\epsilon\frac{\Delta-1}{\Delta}\I_{2q}-\Vb \Vb^{\T} & \Vb \\ \Vb^{\T} & -\frac{\Delta-1}{\Delta}\Wb\end{array}\Big],
\end{equation*}
and the matrices $\Wb,\Vb$ are given by (recall that $\M$ is the matrix in Lemma~\ref{lem:a2} whose eigenvalues are $1,\mu_1,\hdots,\mu_{q-1}$)
\begin{equation*}
\Wb=\frac{1}{\Delta-1}\I_{q^2}-\M\otimes\M,\quad
\Vb=(\alphab^D\oplus \alphab^D)^{-1/2}\A_{2,\gammab}\,(\gammab^D)^{1/2}.
\end{equation*}

In light of~\eqref{eq:secondWHmoment}, it suffices to compute $\Det(\H_2')$. To do this, we proceed by taking the Schur complement of the matrix $\Wb$. It is easy to see that $\Wb$ is invertible, since its spectrum is given by
\begin{equation*}
\begin{gathered}
t- 1,t-\mu_1,t-\mu_1,\hdots,t-\mu_{q-1},t-\mu_{q-1},\\
t- \mu^2_1,t- \mu_1\mu_2,\hdots,t- \mu_1\mu_{q-1},t- \mu_2\mu_1,\hdots,t-\mu_{q-1}^2,
\end{gathered}
\end{equation*}
where $t:=1/(\Delta-1$). We also have
\begin{equation}\label{eq:puttogetherb}
\Det(\Wb)=-\frac{(\Delta-2)}{(\Delta-1)^{q^2}}\prod_{i\in[q-1]}\big(1-(\Delta-1)\mu_i\big)^2\prod_{i,j\in[q-1]}\big(1-(\Delta-1)\mu_i\mu_j\big).
\end{equation}

Considering the Schur complement of the matrix $\Wb$, we obtain
\begin{equation}\label{eq:simpleMformone}
\begin{gathered}
\Det(\H_2')=(-1)^{q^2}\Big(\frac{\Delta}{\Delta-1}\Big)^{2q}\Det\big(\Wb\big)\,\Det\Big(\epsilon\frac{\Delta-1}{\Delta}\I_{2q}+\Zb\Big),\\ \text{ with } \Zb=\frac{\Delta}{\Delta-1}\Vb\Wb^{-1} \Vb^{\T}-\Vb\Vb^{\T}.
\end{gathered}
\end{equation}

To compute $\Det\big(\epsilon\frac{\Delta-1}{\Delta}\I_{2q}+\Zb\big)$, we need to obtain a simpler form for $\Zb$. The following lemma, which is proved at the end of this section, will allow for such a simplication. 
\begin{lemma}\label{lem:inversesubs}
It holds that $\Vb \Wb=\displaystyle \Big(\frac{1}{\Delta-1}\I_{2q}-\M'\Big)\Vb$, where $\M':=\M\oplus \M$.
\end{lemma}
It is standard to express the eigenvalues of $\frac{1}{\Delta-1}\I_{2q}-\M'$ in terms of the eigenvalues of $\M$ and hence obtain that the former matrix is invertible (since the eigenvalues  of $\M$ other than 1 are less in absolute value than $1/(\Delta-1)$). Thus,  Lemma~\ref{lem:inversesubs} gives
\begin{align}
\Zb&=\Big[-\I_{2q}+\frac{\Delta}{\Delta-1}\Big(\frac{1}{\Delta-1}\I_{2q}-\M'\Big)^{-1}\Big]\Vb\Vb^{\T}\notag\\
&=\big(\I_{2q}+\M'\big)\big(\frac{1}{\Delta-1}\I_{2q}-\M'\big)^{-1}\Vb\Vb^{\T}.\label{eq:simpleMformtwo}
\end{align}
By \eqref{eq:simpleMformone}, $\Zb$ is trivially symmetric. Using~\eqref{eq:simpleMformtwo}, we obtain the eigenvalues of $\Zb$.
\begin{lemma}\label{lem:Meigenfin}
The spectrum of $\Zb$ is given by
\[0,2f(1),f(\mu_1),f(\mu_1),f(\mu_2),f(\mu_2),\hdots,f(\mu_{q-1}),f(\mu_{q-1}),\]
where $f(x)=(1+x)(\frac{1}{\Delta-1}-x)^{-1}$.
\end{lemma}
\begin{proof}[Proof of Lemma~\ref{lem:Meigenfin}] Let $\u_1=[\sqrt{\alphab},\sqrt{\alphab}]^{\T}$, $\u_2=[\sqrt{\alphab},-\sqrt{\alphab}]^{\T}$. Note that $\u_1,\u_2$ are linearly independent eigenvectors of $\M'$ corresponding to the eigenvalue 1. 

Using \eqref{eq:almostlaplacians1}, we have that $\Vb\Vb^{\T}=\big[\begin{smallmatrix} \I_q& \Sb_{\gammab}' \\ \Sb_{\gammab}'& \I_q\end{smallmatrix}\big]$,
where $\Sb_{\gammab}'$ is the $q\times q$ matrix whose $(i,j)$ entry is $\sqrt{\alpha_i \alpha_j}$. It follows that $\u_1$ and $\u_2$ are eigenvectors of $\Vb\Vb^{\T}$ with eigenvalues $2$ and $0$, respectively, and hence $\u_1$ and $\u_2$ are eigenvectors of $\Zb$ with eigenvalues $2f(1)$ and $0$, respectively.

Let $\u$ be an eigenvector of $\M'$ corresponding to an eigenvalue $\mu\neq 1$.  Note that $\u$ is perpendicular to both $\u_1$ and $\u_2$. It follows   that $\Vb\Vb^{\T}\u=\u$, so that $\Zb\u=f(\mu)\u$. Thus, $\u$ is also an eigenvector of $\Zb$ with eigenvalue $f(\mu)$.
\end{proof}

To simplify the expressions,  set $r=(\Delta-1)/\Delta$. The matrix $\epsilon r\I_{2q}$ shifts the eigenvalues of $\Zb$ by $\epsilon r$. Thus, Lemma~\ref{lem:Meigenfin} yields
\begin{equation*}
\Det\big(\epsilon r\I_{2q}+\Zb\big)=\epsilon r\big(\epsilon r+2f(1)\big)\prod_{i\in[q-1]}\big(\epsilon r+f(\mu_i)\big)^2.
\end{equation*}
We have $f(1),f(\mu_i)\neq 0$ for every $i\in[q-1]$, so that
\begin{multline}\label{eq:puttogethera}
[\epsilon]\Det\big(\epsilon r\I_{2q}+\Zb\big)=2rf(1)\prod_{i\in[q-1]}\big(f(\mu_i)\big)^2\\=-\frac{4(\Delta-1)^{2q}}{\Delta(\Delta-2)} \prod_{i\in[q-1]}\left(\frac{1+\mu_i}{1-(\Delta-1)\mu_i}\right)^2.
\end{multline}
Plugging \eqref{eq:puttogetherb}  and \eqref{eq:puttogethera} in \eqref{eq:simpleMformone}, we obtain
\begin{equation*}
[\epsilon]\Det(\H_2')=(-1)^{q^2}\frac{4\Delta^{2q-1}}{(\Delta-1)^{q^2}}\prod_{i,j\in[q-1]}\big(1-(\Delta-1)\mu_i\mu_j\big)\prod_{i\in[q-1]}\big(1+\mu_i\big)^2.
\end{equation*}
Using this and \eqref{eq:secondWHmoment}, we obtain \eqref{eq:detH2} as wanted. We conclude by giving the deferred proof of Lemma~\ref{lem:inversesubs}.
\begin{proof}[Proof of Lemma~\ref{lem:inversesubs}]
For notational convenience, let $\Nb:=\M\otimes \M$. The lemma clearly reduces to proving $\Vb\,\Nb=\M'\,\Vb$. Let $\Db:=\Vb\,\Nb$, $\Eb:=\M'\,\Vb$. 

The matrices $\Db,\Eb$ clearly have the same dimensions, since $\Vb$ has dimensions $2q\times q^2$, $\Nb$ has dimensions $q^2\times q^2$ and $\M'$ has dimensions $2q\times 2q$. It remains to check that the entries of $\Db,\Eb$ are equal. First, we give explicit expressions for the entries of  $\Vb,\Nb$. We have
\begin{equation*}
V_{t,(i,k)}=\begin{cases}
\sqrt{\frac{\gamma_{ik}^{*}}{\alpha_i}}\mathbf{1}\{i=t\}, & 1\leq t \leq q\\
\sqrt{\frac{\gamma_{ik}^{*}}{\alpha_{k}}}\mathbf{1}\{k=t-q\}, & q+1\leq t \leq 2q
\end{cases},  \quad N_{(i,k),(j,l)}=\frac{x_{ij}^{*}x_{kl}^{*}}{\sqrt{\gamma_{ik}^{*}\gamma_{jl}^{*}}}.
\end{equation*}
We next consider the $\big(i,(j,l)\big)$ entries of the matrices $\Db,\Eb$. Assume first that $i\leq q$. We have
\begin{align*}
D_{i,(j,l)}&=\sum_{i',k}V_{i,(i',k)}N_{(i',k),(j,l)}=\sum_{i',k}\sqrt{\frac{\gamma_{i'k}^{*}}{\alpha_{i'}}}\mathbf{1}\{i'=i\}\frac{x_{i'j}^{*}x_{kl}^{*}}{\sqrt{\gamma_{i'k}^{*}\gamma_{jl}^{*}}}\\
&=\frac{x_{ij}^{*}}{\sqrt{\alpha_i}\sqrt{\gamma_{jl}^{*}}}\sum_{k}x_{kl}^{*}= \frac{\alpha_lx_{ij}^{*}}{\sqrt{\alpha_i}\sqrt{\gamma_{jl}^{*}}},\\
E_{i,(j,l)}&=\sum_{j'}M'_{i,j'}V_{j',(j,l)}=\sum_{j'}M_{i,j'}\sqrt{\frac{\gamma_{jl}^{*}}{\alpha_j}}\mathbf{1}\{j=j'\}=M_{i,j}\sqrt{\frac{\gamma_{jl}^{*}}{\alpha_j}}
\\&=\frac{x_{ij}^{*}\sqrt{\gamma_{jl}^{*}}}{\alpha_j\sqrt{\alpha_i}}=\frac{\alpha_lx_{ij}^{*}}{\sqrt{\alpha_i}\sqrt{\gamma_{jl}^{*}}}.
\end{align*}
An analogous calculation for $q<i\leq 2q$ yields that $D_{i,(j,l)}=E_{i,(j,l)}$ for every $i,j,l$. 
\end{proof}

\subsection{Bethe Prediction for General Models on Random Regular Graphs}
\label{sec:bethe-uniqueness}

In this section, we show how to extend Theorem~\ref{thm:prediction} for general models on random regular graphs as discussed in Section~\ref{sec:general-results}. A more general result has been derived in \cite[Theorem 1.16]{DMS} for sequences of graphs converging locally to (random) trees, under the assumption of uniqueness of the Gibbs measure on the underlying tree. For the special case of random $\Delta$-regular graphs, we show how to extend Theorem~\ref{thm:prediction} when there is a unique {\em semi-translation invariant}  Gibbs measure. Our proof has a different perspective and yields a slightly simpler condition for random $\Delta$-regular graphs.

Semi-translation invariant  Gibbs measures on $\TreeD$ are Gibbs measures that are invariant under any parity-preserving
automorphisms of $\TreeD$ (c.f., \cite{BW}). They can be specified by a pair of probability vectors $(\alphab,\betab)$
for the even and odd, respectively, vertices. Note that if there is a unique semi-translation invariant
measure, then this measure is also translation invariant. Hence, it corresponds to a fixpoint of the tree recursions \eqref{kkrtko}.


\begin{theorem}\label{thm:predictionb}
Let $\B$ be a regular matrix which specifies a model such that for all $\Delta$-regular graphs $Z_G>0$.
If there is a unique semi-translation invariant Gibbs measure on $\TreeD$ and the corresponding fixpoint is Jacobian attractive, then:
\begin{equation*}
\lim_{n\rightarrow\infty}\frac{1}{n}\E_{\Gc}[\log Z_G]=\lim_{n\rightarrow\infty}\frac{1}{n}\log\E_{\Gc}[ Z_G],
\end{equation*}
where $\Gc$ is the uniform distribution on $\Delta$-regular graphs with $n$ vertices.
\end{theorem}

The first assumption in the theorem is mainly to avoid pathological cases where $\log Z_G\equiv-\infty$
in which case the quantities are not well-defined. It is satisfied by many classes of models, e.g., permissive models (\cite{DMS}) such as the hard-core and antiferromagnetic Potts model, or even non-permissive such as $q$-colorings when $q\geq \Delta+1$.

The proof of  Theorem~\ref{thm:predictionb} is analogous to that of Theorem~\ref{thm:prediction}, once we establish the analog of Theorem~\ref{thm:copo} for general models.  As we illustrated in Remark~\ref{rem:anti-fails}, this is hopeless to achieve in general and we must thus use the uniqueness assumption that Theorem~\ref{thm:prediction} requires.  Note that if there is a unique semi-translation invariant
measure (which is the assumption in Theorem \ref{thm:prediction}) then this measure is also translation invariant.

\begin{proof}[Proof of Theorem~\ref{thm:predictionb}]
Let $\alphab^*$ be a dominant phase. By semi-translational uniqueness we have that $\alphab^*$ is unique. We next describe how to obtain the analog of Theorem~\ref{thm:copo} under the assumptions of Theorem~\ref{thm:predictionb}. Let $p=\Delta/(\Delta-1)$. We show that whenever there is  a unique  semi-translation Gibbs measure on $\Tree_\Delta$, it holds that $\exp(2\Psi_1(\alphab^*)/\Delta)= \|\B\|_{p\ra\Delta}$.

From \eqref{zzzok} and \eqref{pwwww1}, we obtain:
\begin{equation}\label{eq:mmmmtr}
\exp(2\Psi_1(\alphab^*)/\Delta)=\max_{\alphab}\exp(2\Psi_1(\alphab)/\Delta)= \max_{\bR}\frac{ \bR^{\T} \B \bR}{\|\bR\|_{p}^2}\leq \max_{\bR,\Cb}\frac{ \bR^{\T}\B \Cb}{\|\bR\|_{p}\|\Cb\|_{p}}.
\end{equation}
Note that the last inequality is trivial; we just enlarged the maximization region we consider. It is proved in \cite{GSV:colorings} that the maximum of the r.h.s. is achieved at a semi-translation invariant fixpoint. If there is a unique  semi-translation invariant Gibbs measure on $\Tree_\Delta$, this must be translation invariant and hence the maximum in the r.h.s. of \eqref{eq:mmmmtr} must occur at $\bR=\Cb$. We thus obtain  that \eqref{eq:mmmmtr} is satisfied at equality. The r.h.s. in \eqref{eq:mmmmtr} is equal to $\|\B\|_{p\ra\Delta}$ (\cite[Section 3.1]{GSV:colorings}), proving the desired claim.

By the same token, one has  the bound
\begin{equation*}
\exp(2\Psi_2(\alphab^*)/\Delta)=\max_{\alphab}\exp(2\Psi_2(\alphab)/\Delta)\leq \|\B\otimes \B\|_{p\ra\Delta},
\end{equation*}
and since $\|\B\otimes \B\|_{p\ra\Delta}=\|\B\|_{p\ra\Delta}^2$ (\cite[Proposition 10.3]{MR0493490}), we obtain that  $\Psi_2(\alphab^*)=2\Psi_1(\alphab^*)$, as wanted.

Since the dominant phase $\alphab^{*}$ corresponds to a Jacobian attractive fixpoint (by assumption), it is also Hessian dominant (see Remark~\ref{rema33}). With minor modifications (see footnotes~\ref{foot:invertible} and~\ref{foot:det}), the results of Section~\ref{sec:small-graph} can  be adapted to obtain a lower bound on $Z_G^{\alphab}$ as in  Lemma~\ref{lem:smallgraph}. Thus the proof of Theorem~\ref{thm:prediction} in Section~\ref{sec:small-graph-overview} extends to the present setting as well.
\end{proof}

\bibliographystyle{siamplain}

\appendix
\section{ Non-Reconstruction for the Ordered Phases on the Tree}\label{app:nonreconstruction}
In this appendix, we give in detail the proof of the doubly exponential upper bound in \eqref{eq:MST}. This appendix is organized as follows. In Appendix~\ref{app:broadcasting}, we review broadcasting processes on trees, the non-reconstruction property and a concentration result from \cite{SlyZhang}. In Appendix~\ref{app:fixpoints}, we review relevant connections between broadcasting processes and Gibbs measures defined by fixpoints of the tree recursions \eqref{kkrtko}, which will allow us to apply the result of \cite{SlyZhang}. Finally, in Appendix~\ref{app:finalproof}, we apply these results to the ferromagnetic Potts model and obtain the bound in \eqref{eq:MST}.

Let us fix some notation that will be used throughout this section. We will denote by $T=(V,E)$  the infinite $(\Delta-1)$-ary tree. The root of $T$ will be denoted by $\rho$. Also, for an integer $\ell\geq 0$, $T_\ell$ will denote the subtree of $T$ consisting of the first $\ell$ levels of $T$ and $W_\ell$ will denote the set of the leaves of $T_\ell$. Further, for a configuration $\sigma:V\rightarrow[q]$, we denote by $\sigma_{W_\ell}$ the restriction of $\sigma$ on $W_\ell$.

\subsection{Non-Reconstruction in Broadcasting Processes on Trees}\label{app:broadcasting}
Let $q\geq 2$ be an integer and $\M=(M_{ij})_{i,j\in[q]}$ be a $q\times q$ stochastic matrix (i.e., the entries are non-negative and the entries in each row have sum equal to 1). We will further assume that $\M$ is irreducible and aperiodic, so that there exists a unique $q$-dimensional probability vector $\pib^*=(\pi^*_i)_{i\in[q]}$ so that $\pib^* \M=\pib^*$. Note that the entries of $\pib^*$ are all positive.  We will refer to $\pib^*$ as the stationary distribution of $\M$. We will also assume that $\M$ is reversible with respect to $\pib^*$, i.e., $\pi^*_i M_{ij}=\pi^*_j M_{ji}$ for all $i,j\in[q]$ (every such matrix is similar to a symmetric matrix and thus has real eigenvalues).

Let $\pib=(\pi_i)_{i\in[q]}$ be a $q$-dimensional probability vector with positive entries (note that it may hold  that $\pib\neq \pib^*$). The broadcasting process $\M$ on the tree $T=(V,E)$ with root $\rho$ is a probability distribution $\nu$ on the set of assignments $\sigma:V\rightarrow [q]$ such that 
\begin{equation}\label{eq:broadcastingmeasure}
\nu(\sigma)=\pi_{\sigma(\rho)} \prod_{(u,v)\in E}M_{\sigma(u),\sigma(v)}.
\end{equation}
To generate $\sigma$ with distribution $\nu$, first pick randomly the spin of the root from the distribution $\pib$ and then broadcast the spin down the tree, where each edge of the tree acts as a noisy channel. In particular, for an edge $(u,v)$ of the tree where $u$ is the parent of $v$, conditioned on the spin $\sigma(u)$, the spin $\sigma(v)$ is picked randomly from the distribution $(M_{\sigma(u),1},\hdots,M_{\sigma(u),q})$.

We next define the non-reconstruction property, which roughly captures whether, as we go deeper into the tree, the information about the spin of the root vanishes.   (For distributions $\mu_1,\mu_2$ defined on the same space $\Omega$, we denote by $d_{TV}(\mu_1,\mu_2)$ the total variation distance between $\mu_1,\mu_2$.)
\begin{definition}[Non-Reconstruction]
A broadcasting process $\M$ has the non-reconstruction property on the tree $T$ if
\begin{equation}\label{eq:forget}
\lim_{\ell\rightarrow\infty}\max_{i,j\in[q]}d_{TV}\big(\nu(\sigma_{W_\ell}=\cdot\mid \sigma_\rho=i),\nu(\sigma_{W_\ell}=\cdot\mid \sigma_\rho=j)\big)=0.
\end{equation}
\end{definition}

Non-reconstruction is often closely connected to the second largest eigenvalue of $\M$. We will use the following concentration result of \cite{SlyZhang}, which can be interpreted as quantifying the rate of convergence to 0, when the second largest eigenvalue of $\M$ is small with respect to the branching factor of the tree. 

\begin{theorem}[{\cite[Theorem 2.3]{SlyZhang}}, see also \cite{MSTb, MST}]\label{thm:importtool}
Consider a broadcasting process $\M$ on the infinite $(\Delta-1)$-ary tree with no hard constraints (i.e., all entries of $\M$ are positive), whose spin at the root is chosen according to some distribution $\pib$ with positive entries. Let $\lambda$ be the second largest eigenvalue of $\M$ in absolute value. Then, if $\M$ has non-reconstruction and $(\Delta-1)\lambda^2< 1$, there exist constants  $C>0$ and $\ell_0\geq 1$ such that the following holds.

Let $\mathcal{B}_\ell:=\Big\{\eta: W_\ell\rightarrow [q]\,\Big|\, \big\|\nu(\sigma_\rho=\cdot\mid \sigma_{W_\ell}=\eta)-\pib\big\|_{\infty}\geq \exp(-C\ell)\Big\}$. Then, for all $\ell\geq \ell_0$,
\[\nu(\sigma_{W_\ell}\in \mathcal{B}_\ell)\leq\exp(-\exp(C\ell)).\]
\end{theorem}
We remark here that the restriction in Theorem~\ref{thm:importtool} that $\M$ has no hard constraints is not needed and, in fact, in \cite{SlyZhang}, the analogous statement is proved for general models $\M$  whose state space satisfies a general connectivity condition. Since we will only apply the result of \cite{SlyZhang} to the ferromagnetic Potts model (which has no hard constraints), such connectivity issues are not present in our setting and thus out of our scope. In particular, in the language/notation of \cite{SlyZhang}, all colors $c,c'$ will be trivially compatible in our setting and thus the measure $\mu^c(\cdot)$ in \cite[Theorem 2.3]{SlyZhang}, which conditions the broadcasting process in the space of configurations where the ``parent of the root" has color $c$, is identical to the unconditioned broadcasting process (denoted by $\nu$ in our setting).  Further, \cite[Theorem 2.3]{SlyZhang} is stated for the case where $\pib=\pib^*$, i.e., when the distribution of the spin of the root $\rho$ is chosen according to the stationary distribution of $\M$. We next display how to derive from this the slightly more general version stated in Theorem~\ref{thm:importtool}.

In particular, suppose that Theorem~\ref{thm:importtool} is true for some distribution $\pib$. Our goal is to show that it also holds for some other distribution $\pib'$ (we assume that both $\pib$ and $\pib'$ have positive entries) for some constants $C',\ell_0'>0$. We will denote by $\nu$ the broadcasting process when the initial distribution is $\pib$ and by $\nu'$ when the initial distribution is $\pib'$. We will also use $\mathcal{B}_\ell$ and $\mathcal{B}_\ell'$ to denote the set of ``bad" configurations on $W_\ell$ for the two processes $\nu,\nu'$, respectively (see Theorem~\ref{thm:importtool}).  Let $\eta\notin\mathcal{B}_\ell$ and, for $i\in[q]$, set $z_i(\eta):=\nu(\sigma_\rho=i\mid \sigma_{W_\ell}=\eta)/\nu(\sigma_\rho=i)$. Denote also by $z_i'(\eta)$ the respective quantity for the measure $\nu'$. Since both processes have the same broadcasting matrix, observe that for any colors $i,j\in [q]$ it holds that
\[\frac{z_i(\eta)}{z_j(\eta)}=\frac{\nu(\sigma_{W_\ell}=\eta\mid \sigma_\rho=i)}{\nu(\sigma_{W_\ell}=\eta\mid \sigma_\rho=j)}=\frac{\nu'(\sigma_{W_\ell}=\eta\mid \sigma_\rho=i)}{\nu'(\sigma_{W_\ell}=\eta\mid \sigma_\rho=j)}=\frac{z_i'(\eta)}{z_j'(\eta)}.\]
Since $\eta\notin \mathcal{B}_\ell$, the ratio  $z_i(\eta)/z_j(\eta)$ is bounded by $1\pm O(\exp(-C\ell))$ and thus the same is true for $z_i'(\eta)/z_j'(\eta)$. This gives that $\eta\notin \mathcal{B}_\ell'$ (for any constant $0<C'<C$ and sufficiently large $\ell'_0$), i.e., $\mathcal{B}_\ell'\subseteq \mathcal{B}_\ell$. To obtain that $\nu'(\sigma_{W_\ell}\in\mathcal{B}_\ell')\leq \exp(-\exp(C'\ell))$ for all sufficiently large $\ell$, observe that for any $\eta:W_\ell\rightarrow [q]$ and $i\in[q]$ it holds that $\nu(\sigma_{W_\ell}=\eta\mid \sigma_\rho=i)=\nu'(\sigma_{W_\ell}=\eta\mid \sigma_\rho=i)$, so that
\[\frac{\nu'(\sigma_{W_\ell}=\eta)}{\nu(\sigma_{W_\ell}=\eta)}=\frac{\sum_{i\in[q]}\pi_i'\,\nu'(\sigma_{W_\ell}=\eta\mid \sigma_\rho=i)}{\sum_{i\in[q]}\pi_i\nu(\sigma_{W_\ell}=\eta\mid \sigma_\rho=i)}\leq \max_{i\in[q]}\frac{\pi_i'}{\pi_i}.\]
Thus the desired bound on $\nu'(\sigma_{W_\ell}\in\mathcal{B}_\ell')$ follows from the bound on $\nu(\sigma_{W_\ell}\in\mathcal{B}_\ell)$.

\subsection{Broadcasting Processes and fixpoints of the tree recursions}\label{app:fixpoints}
In light of Theorem~\ref{thm:importtool}, our strategy for proving the bound in \eqref{eq:MST} will be to show that the measure $\nu^i$ (corresponding to the $i$-th ordered phase in the Potts model) corresponds to a broadcasting process on the $(\Delta-1)$-ary tree (and then simply verify the assumptions of the theorem). The purpose of this section is to make this correspondence explicit. In fact, we will workout the relevant connections for general spin models.

Let $\B$ be the interaction matrix of a $q$-spin system. As in Section~\ref{sec:bits}, we assume that $\B$ is symmetric, irreducible and aperiodic. For an integer $\Delta\geq 3$, recall that a fixpoint of the tree recursions is a vector $\Rb=(R_1,\hdots,R_q)$ with positive entries such that
\begin{equation*}\tag{\ref{kkrtko}}
R_i\propto \bigg(\sum_{j} B_{ij} R_j\bigg)^{\Delta-1} \mbox{ for all } i\in[q].
\end{equation*}
For the purpose of this section, we assume that the normalization in \eqref{kkrtko} is such that $\sum_{i}R_i=1$, i.e., $\Rb$ is a $q$-dimensional probability vector. 

We next define the broadcasting process corresponding to the fixpoint $\Rb$ by first specifying an appropriate broadcasting matrix. In particular, let $\M$ be the $q\times q$ matrix whose $(i,j)$-entry is given by
\begin{equation}\label{eq:entriesofM}
M_{ij}=\frac{B_{ij}R_j}{\sum_{j'}B_{ij'}R_{j'}} \mbox{ for } i,j\in[q].
\end{equation}
We remark here that the normalization of the $R_i$'s in \eqref{kkrtko} is not important for defining the matrix $\M$ (the entries remain unchanged if we scale the $R_i$'s); we normalize $\Rb$ to be a probability vector so that we can use it  as the initial distribution $\pib$ of the spin of the root in the broadcasting process. In particular, in the notation of Appendix~\ref{app:broadcasting}, we will set $\pib=\Rb$. This completes the specification of the broadcasting process (cf. \eqref{eq:broadcastingmeasure}). Note that $\M$ is stochastic, irreducible and aperiodic. Further, its stationary distribution $\pib^*$ is given by the probability vector whose entries satisfy  $\pi_i^*\propto R_i\sum_{j}B_{ij}R_j$ for all $i\in[q]$. Finally, we have that $\M$ is reversible with respect to $\pib^*$.

In the rest of this section, we state several results that eventually will allow us to apply Theorem~\ref{thm:importtool}. First, we connect the spectral properties of $\M$ with the attractiveness of the fixpoint $\Rb$ of the tree recursions (see Section~\ref{sec:conn2tree} for the relevant definitions).
\begin{lemma}\label{lem:broadcasteigen}
Let $\Rb$ be a Jacobian attractive fixpoint of the tree recursions and let $\M$ be the broadcasting matrix corresponding to $\Rb$. Let $\lambda$ be the second largest eigenvalue of $\M$ in absolute value. Then $(\Delta-1)\lambda<1$.
\end{lemma}
\begin{proof}
Recall from Section~\ref{sec:connection} (see also the beginning of Section~\ref{sec:Potts}) that $\Rb$ is a Jacobian attractive fixpoint of the tree recursions if every eigenvalue $x\neq 1$ of the matrix
\begin{equation*}
\widetilde{\M}=\left\{\frac{B_{ij}R_iR_j}{\sqrt{\alpha_i\alpha_j}}\right\}_{i,j=1}^q \mbox{ with } \alpha_i=R_i\sum_{j}B_{ij}R_j \mbox{ for $i\in[q]$}
\end{equation*}
satisfies $(\Delta-1)|x|<1$. The result will thus follow by showing that the eigenvalues of $\M$ are identical to those of $\widetilde{\M}$. 

We will show that $\M$ and $\widetilde{\M}$ are similar matrices, thus showing the result. Let $\A$ be the diagonal matrix whose $i$-th diagonal entry is given by $\sqrt{\alpha_i}$. Note that $\A$ is invertible (since the $R_i$'s are positive). By a direct calculation, it also holds that ${\A} \M \A^{-1}=\widetilde{\M}$, thus proving that $\M$ and $\widetilde{\M}$ are similar. This concludes the proof.
\end{proof}

We now focus on connecting the Gibbs distribution of the spin model with interaction matrix $\B$ and the broadcasting process $\M$. As before, let $T$  be the infinite $(\Delta-1)$-ary tree with root $\rho$ and denote by $T_\ell$ the subtree of $T$ consisting of the first $\ell$ levels of $T$ and by $W_\ell$ the set of leaves of $T_{\ell}$. We will denote by $\mu_\ell$ the Gibbs distribution on $T_\ell$ corresponding to the spin system with interaction matrix $\B$. We will use $\sigma$ to denote configurations on $T_\ell$ and by $\sigma_{W_\ell}$ the restriction of $\sigma$ to the leaves $W_\ell$. 

To connect $\mu_\ell$ to the broadcasting process $\M$ on $T$, we will need just a few more definitions. Let $Q_{W_\ell}(\cdot)$ be the following product distribution on configurations on the leaves $W_\ell$. For a configuration $\eta:W_{\ell}\rightarrow[q]$, 
\begin{equation}
Q_{W_\ell}(\eta):=\prod_{i\in [q]}(R_i)^{|\eta^{-1}(i)\cap W_{\ell}|}.
\end{equation}
Finally, consider the following distribution $\widehat{\nu}_\ell$, which is also defined on configurations on the leaves $W_\ell$, given by
\begin{equation}\label{eq:conditionedmeasures}
\widehat{\nu}_{\ell}(\eta)\propto \mu_\ell(\sigma_{W_\ell}=\eta)\, Q_{W_\ell}(\eta)\mbox{ for all }\eta:W_{\ell}\rightarrow[q].
\end{equation}
It is instructive at this point to spell out the interplay of these definitions with the bound in \eqref{eq:MST}. Namely, the product distribution $Q_{W_\ell}(\eta)$ is the generalization of the product distribution $Q^i_W(\eta)$ (defined just after \eqref{def:proddistribution}) and $\widehat{\nu}_{\ell}(\cdot)$ is the generalization of the distribution $\nu^i(\cdot)$ (defined in \eqref{eq:nuidef}). 

We are now ready to state the desired connection.
\begin{lemma}\label{lem:Markovconstruction}
Let $\B,\Rb,\M$ be as above. Let $\nu$ denote the broadcasting measure $\M$ on $T$ (defined in \eqref{eq:broadcastingmeasure}) and, for integer $\ell\geq 0$, let $\mu_\ell$ be the Gibbs distribution on $T_\ell$ corresponding to the spin system with interaction matrix $\B$, and $\widehat{\nu}_\ell(\cdot)$ be the distribution in \eqref{eq:conditionedmeasures} corresponding to the fixpoint $\Rb$ of the tree recursions. 

Then, for all $\ell\geq 0$, for all $\eta:W_\ell\rightarrow [q]$ and $i\in[q]$, it holds that 
\[\widehat{\nu}_\ell(\eta)=\nu(\sigma_{W_\ell}=\eta) \mbox{ and }  \mu_\ell(\sigma_\rho=i\mid \sigma_{W_\ell}=\eta)=\nu(\sigma_\rho=i\mid \sigma_{W_\ell}=\eta).\]
\end{lemma}
\begin{proof}
Let $d:=\Delta-1$. The proof is by induction on $\ell$. For $\ell=0$, the lemma is trivial. Let us assume that the lemma holds for $\ell$, we will prove it for $\ell+1$. For a vertex $v\in W_{\ell}$ denote by $v_1,\hdots,v_d$ the children of $v$ in $T$ (note that $v_1,\hdots,v_d\in W_{\ell+1}$). 

We prove first that $\widehat{\nu}_{\ell+1}(\eta)=\nu(\sigma_{W_{\ell+1}}=\eta)$ for all $\eta:W_{\ell+1}\rightarrow [q]$. By definition of the broadcasting process, conditioned on the configuration $\tau$ on $W_\ell$, the spins of the vertices in $W_{\ell+1}$ are independent. We thus have that, for all $\eta:W_{\ell+1}\rightarrow [q]$,
\[\nu(\sigma_{W_{\ell+1}}=\eta)=\sum_{\tau:W_{\ell}\rightarrow[q]}\nu(\sigma_{W_\ell}=\tau)\prod_{v\in W_{\ell}}\prod^{d}_{j=1}M_{\tau_v,\eta_{v_j}}.\]
Using the induction hypothesis we have that $\nu(\sigma_{W_\ell}=\tau)=\widehat{\nu}_\ell(\tau)\propto \mu_\ell(\sigma_{W_\ell}=\tau)Q_\ell(\tau)$ for all $\tau: W_{\ell}\rightarrow[q]$ and, substituting the value of the product measure $Q_{W_\ell}(\tau)$, we obtain
\begin{equation}\label{eq:qry3}
\nu(\sigma_{W_{\ell+1}}=\eta)\propto \sum_{\tau:W_{\ell}\rightarrow[q]}\mu_\ell(\sigma_{W_\ell}=\tau)\prod_{v\in W_{\ell}}R_{\tau_v}\prod^{d}_{j=1}M_{\tau_v,\eta_{v_j}}\mbox{ for all $\eta:W_{\ell+1}\rightarrow [q]$.}
\end{equation}
We also have that $\mu_{\ell+1}(\sigma_{W_{\ell+1}}=\eta)\propto \sum_{\tau:W_{\ell}\rightarrow[q]}\mu_\ell(\sigma_{W_\ell}=\tau)\prod_{v\in W_{\ell}}\prod^{d}_{j=1}B_{\tau_v,\eta_{v_j}}$ for all $\eta:W_{\ell+1}\rightarrow [q]$, so substituting the value of the product measure $Q_{W_{\ell+1}}(\eta)$ we obtain
\begin{equation}\label{eq:qry4}
\widehat{\nu}_{\ell+1}(\eta)\propto \sum_{\tau:W_{\ell}\rightarrow[q]}\mu_\ell(\sigma_{W_\ell}=\tau)\prod_{v\in W_{\ell}}\prod^{d}_{j=1}R_{\eta_{v_j}} B_{\tau_v,\eta_{v_j}}\mbox{ for all $\eta:W_{\ell+1}\rightarrow [q]$.}
\end{equation}
To complete the induction step, it thus remains to show that the r.h.s. in \eqref{eq:qry3} and \eqref{eq:qry4} are proportional by a factor that does not depend on $\eta$. This will follow from 
\begin{equation}\label{eq:proppropprop}
R_{i}M_{i,j_1}\cdots M_{i,j_d}\propto (R_{j_1} B_{i,j_1})\cdots (R_{j_d}B_{i,j_d}) \mbox{ for all } i,j_1,\hdots,j_d\in [q].
\end{equation}
(Set $i=\tau_v$, $j_1=\eta_{v_1},\hdots,j_d=\eta_{v_d}$ and multiply over $v\in W_{\ell}$.) To see \eqref{eq:proppropprop}, note that by \eqref{eq:entriesofM}, we have that, for every $i,j\in [q]$, 
\begin{equation*}\tag{\ref{eq:entriesofM}}
M_{ij}=\frac{B_{ij}R_j}{\sum_{j'}B_{ij'}R_{j'}},
\end{equation*}
so, to prove \eqref{eq:proppropprop}, it suffices to show that the quantities $\frac{R_i}{\big(\sum_{j'\in[q]}B_{ij'}R_{j'}\big)^{d}}$ do not depend on $i\in[q]$. This is a consequence of the fact that $\Rb$ is a fixpoint of the tree recursions, i.e., $R_1,\hdots, R_q$ satisfy \eqref{kkrtko}. This completes the induction step for the first equality in the lemma.

We next show the induction step for the second equality in the lemma. Fix $\eta:W_{\ell+1}\rightarrow [q]$. Denote by $\rho_1,\hdots,\rho_d$ the children of the root $\rho$. For $k\in[d]$, denote by $W_{\ell+1,k}$ the set of vertices in $W_{\ell+1}$ which are in the subtree of $T$ rooted at $\rho_k$ and by $\eta_{k}$ the restriction of $\eta$ on $W_{\ell+1,k}$. By the induction hypothesis, we have, for every $k\in[d]$ and $i\in[q]$,  
\[X_{k}(i):=\mu_{\ell}(\sigma_{\rho_k}=i\mid \sigma_{W_{\ell+1,k}}=\eta_k)=\nu(\sigma_{\rho}=i\mid \sigma_{W_{\ell}}=\eta_k).\]
Note that $\nu(\sigma_{W_{\ell+1,k}}=\eta_k\mid \sigma_{\rho_k}=i)=\nu(\sigma_{W_{\ell}}=\eta_k\mid \sigma_{\rho}=i)$ from where we obtain that 
\[\frac{\nu(\sigma_{\rho_k}=i\mid \sigma_{W_{\ell+1,k}}=\eta_k)}{\nu(\sigma_{\rho_k}=i)}\propto\frac{\nu(\sigma_{\rho}=i\mid \sigma_{W_{\ell}}=\eta_k)}{\nu(\sigma_{\rho}=i)}=\frac{X_k(i)}{R_i} \mbox{ for all } i\in[q].\]
(Note that the normalizing factor depends on $\eta_k$.)

Using that $T_{\ell+1}$ is a tree, we then calculate that 
\begin{equation}\label{eq:qry1}
\mu_{\ell+1}(\sigma_\rho=i\mid \sigma_{W_{\ell+1}}=\eta)\propto\prod_{k\in [d]}\bigg(\sum_{j\in[q]} B_{ij}X_{k}(j)\bigg) \mbox{ for } i\in[q],
\end{equation}
and (see \cite[Lemma 3.1]{SlyZhang} for a thorough derivation)
\begin{equation}\label{eq:qry2}
\nu(\sigma_\rho=i\mid \sigma_{W_{\ell+1}}=\eta)\propto R_i\prod_{k\in [d]}\bigg(\sum_{j\in[q]} M_{ij}\frac{X_{k}(j)}{R_j}\bigg) \mbox{ for } i\in[q].
\end{equation}
By a completely analogous argument to the one we used for \eqref{eq:qry3} and \eqref{eq:qry4} (i.e., using \eqref{eq:entriesofM} and the fact that  $\Rb$ is a fixpoint of the tree recursions \eqref{kkrtko}),  we obtain that the r.h.s. in \eqref{eq:qry1} and \eqref{eq:qry2} are proportional by a factor that does not depend on $i$, thus completing the induction step for the second equality in the lemma.

This concludes the proof of Lemma~\ref{lem:Markovconstruction}.
\end{proof}

\subsection{Application to the Ferromagnetic Potts model -- Proof of \eqref{eq:MST}}\label{app:finalproof}
We are now able to apply the results of Appendices~\ref{app:broadcasting} and~\ref{app:fixpoints} to the ferromagnetic Potts model and prove the bound \eqref{eq:MST} for the ordered phases on the tree.

Recall that the interaction matrix $\B$ of the $q$-state ferromagnetic Potts model has diagonal entries equal to $B>1$ and off-diagonal entries equal to 1. An ordered phase corresponds to a fixpoint $\Rb=(R_1,\hdots,R_q)$ of the tree recursions. Thus, the $R_i$'s satisfy
\begin{equation}\label{eq:potts4rfv}
R_i\propto \bigg(B_{ii}R_i+\sum_{j\neq i} R_j\bigg)^{\Delta-1} \mbox{ for all } i\in[q].
\end{equation}
Recall that there are $q$ ordered phases which are symmetric, each corresponding to a color $i\in[q]$. W.l.o.g., we will focus on the ordered phase corresponding to the color $i=1$. As we showed in Section~\ref{sec:Potts}, the solution of \eqref{eq:potts4rfv} corresponding to the ordered phase $i=1$ is given by the vector $\Rb$ which satisfies \eqref{eq:potts4rfv}, $R_1>R_2=\hdots=R_q$ and  $R_1/R_q$ is maximum (see Remark~\ref{rem:majorityordered}). Such a solution exists in the non-uniqueness region, i.e., when $B>\Bu$.  The broadcasting matrix $\M$ corresponding to the ordered phase $i=1$ is given by \eqref{eq:entriesofM}:
\begin{equation}\label{eq:numericspotts}
\begin{aligned}
M_{11}&=\frac{B R_{1}}{BR_1+(q-1)R_q},&& M_{1j}=\frac{R_q}{BR_1+(q-1)R_q}\mbox{ for $j\neq 1$}\\
M_{ii}&=\frac{B R_{q}}{R_1+(q-2+B)R_q} \mbox{ for $i\neq 1$},&& M_{ij}=\frac{ R_{q}}{R_1+(q-2+B)R_q} \mbox{ for $i\neq 1,\,i\neq j$}.
\end{aligned}
\end{equation}
We need the following lemma, which can be inferred from \cite[Proof of Theorem 1.4]{MST}. For completeness, we give  the proof. 
\begin{lemma}\label{lem:pottsnon345}
Let $\Delta\geq 3$ be an integer and $B>\Bu$. Then, the broadcasting process $\M$ defined by \eqref{eq:numericspotts} is non-reconstructible on the $(\Delta-1)$-ary tree.
\end{lemma}
\begin{proof}
Let $i,j\in[q]$ be two arbitrary colors with $i\neq j$ and consider two copies $X,Y$ of the broadcasting process on the $(\Delta-1)$-ary tree where the spins of the root $\rho$ are conditioned to be $i$ and $j$ respectively. To show that the total variation distance between the distributions $\nu(\sigma_{W_{\ell}}=\cdot\mid \sigma_\rho=i)$ and $\nu(\sigma_{W_{\ell}}=\cdot\mid \sigma_\rho=j)$ goes to 0 as $\ell\rightarrow \infty$, it suffices to couple $X,Y$ so that the expected number of disagreements, i.e., vertices in $W_\ell$ whose spins are different, goes to 0 as $\ell\rightarrow \infty$. In turn, it suffices to couple one step of the broadcasting process so that the expected number of disagreements is bounded by some constant $\kappa<1/(\Delta-1)$, since this yields  that the expected number of disagreements at level $\ell$ decays exponentially with $\ell$, at least as fast as $((\Delta-1)\kappa)^\ell$. 

In particular, let $(u,v)$ be an arbitrary edge in the tree, with $u$ being the parent of $v$. By the Coupling Lemma, conditioned on the spin of $u$ in $X$ and $Y$, we can couple the spins of $v$ in $X$ and $Y$ so that the probability that they are different is bounded by $\kappa$, where  
\[\kappa:=\max_{i,j\in[q]}d_{TV}\big(\nu(\sigma_{v}=\cdot\mid \sigma_u=i),\nu(\sigma_{v}=\cdot\mid \sigma_u=j)\big)=\max_{i,j}\frac{1}{2}\sum_{k\in [q]}|M_{ik}-M_{jk}|.\]
In the following, we justify that $\kappa<1/(\Delta-1)$. We will see that, in the case of the ferromagnetic Potts model, $\kappa$ is related to the  eigenvalues of the Jacobian matrix of the tree recursions (evaluated at the fixpoint), which we have already studied in Section~\ref{sec:Potts}. First, we find a simpler expression for $\kappa$. Consider  colors $i,j\neq 1$. Then
\begin{equation*}
\frac{1}{2}\sum_{k\in [q]}|M_{ik}-M_{jk}|=M_{ii}-M_{ij}=\frac{(B-1) R_{q}}{R_1+(q-2+B)R_q}=\lambda_1,
\end{equation*}
where $\lambda_1$ is as in \eqref{eq:lambda1lambda2}. Consider now the case that $i=1$ and $j\neq 1$. Using that $B>1$ and $R_1>R_q$, we have $M_{11}>M_{j1}$ and $M_{1k}<M_{jk}$ for $k\neq 1$. It follows that 
\begin{equation*}
\frac{1}{2}\sum_{k\in [q]}|M_{1k}-M_{jk}|=M_{11}-M_{j1}=\frac{B R_{1}}{BR_1+(q-1)R_q}-\frac{R_q}{R_1+(q-2+B)R_q}=\lambda_2,
\end{equation*}
where $\lambda_2$ is as in \eqref{eq:lambda1lambda2}. In the proof of Lemma~\ref{lem:exactone}, we showed that $\lambda_1,\lambda_2<1/(\Delta-1)$ for all $B>\Bu$, which shows that $\kappa<1/(\Delta-1)$, thus completing the proof of the lemma.
\end{proof}

We conclude this appendix by giving the proof of \eqref{eq:MST}.
\begin{proof}[Proof of \eqref{eq:MST}] Recall that we only need to consider the case where $J$ consists of a single $(\Delta-1)$-ary tree of height $\ell=\left\lfloor \psi \log_{\Delta-1} n\right\rfloor$. Using the second equality in Lemma~\ref{lem:Markovconstruction}, we have that the set of ``bad" configurations $\mathcal{B}_i$ defined in \eqref{eq:realbi} is a subset of the set $\mathcal{B}_\ell$ defined in Theorem~\ref{thm:importtool} (for all $0<\theta<\frac{C \psi}{3\ln (\Delta-1)}$, it holds that $n^{-3\theta}>\exp(-C\ell)$). Further, using Lemmas~\ref{lem:broadcasteigen} and~\ref{lem:pottsnon345}, the assumptions of Theorem~\ref{thm:importtool} are all satisfied for the broadcasting process $\M$ defined by \eqref{eq:numericspotts}. As we observed just after \eqref{eq:conditionedmeasures}, $\widehat{\nu}_{\ell}(\cdot)$ is identical to the distribution $\nu^i(\cdot)$ (defined in \eqref{eq:nuidef}). Thus, the conclusion of Theorem~\ref{thm:importtool} and the first equality in Lemma~\ref{lem:Markovconstruction} yield
\begin{equation*}\tag{\ref{eq:MST}}
\nu^i(\sigma_W\in\mathcal{B}_i)\leq \exp(-\exp(C\ell)),
\end{equation*}
as wanted.
\end{proof}

\section{Moment Asymptotics --- Proof of Lemmas~\ref{lem:firasympdet} and~\ref{lem:secasympdet}}\label{app:A}
\label{sec:momentasymptotics}
In this appendix, we give the proofs of Lemmas~\ref{lem:firasympdet} and~\ref{lem:secasympdet} which express the asymptotics of the moments in terms of certain determinants. The proofs of these lemmas are similar and closely follow \cite[Lemma B.3]{GSV:colorings}.

\begin{proof}[Proof of Lemma~\ref{lem:firasympdet}]
We have already seen that the ferromagnetism of the model implies a full-dimensional representation of $\x$ which consists of the variables $x_{ij}$ with $(i,j)\in P_1^{*}$ where $P_1^*=P_1\backslash \big\{(i,i)\mid i\in[q]\big\}$. For $(i,j)\notin P_1^{*}$, we will still use $x_{ij}$ as a shorthand for the appropriate linear combination of the variables inside the full-dimensional representation.

Recall that
\begin{multline}\tag{\ref{eq:firstmoment}}
\E_{\G}[Z^{\alphab}_G]=\binom{n}{\alpha_{1}n,\hdots,\alpha_{q}n}\sum_{\x}\left\{\prod_{i}\binom{\Delta\alpha_{i}n}{\Delta x_{i1}n,\hdots,\Delta x_{iq}n}\right.\\ \left.\times\frac{\big[\prod_{i\neq j}(\Delta x_{ij}n)!\big]^{1/2}\prod_{i}(\Delta x_{ii}n-1)!!}{(\Delta n-1)!!}\prod_{i,j}B^{\Delta x_{ij}n/2}_{ij}\right\},
\end{multline}
Note, the range of $\x$ in the summation is over those vectors $\x$ such that, for each $i\in [q]$, $\Delta x_{ii}n=2e_{ii}n$ is an even integer, which yields the constraint that $\sum_{j\neq i}\Delta x_{ij}n \equiv \Delta \alpha_i n (\mathrm{mod}\, 2)$. Since $\Delta n$ is even and $\sum_i\alpha_i=1$, only $q-1$ of these constraints are linearly independent $(\mathrm{mod}\, 2)$.

Since $\alphab$ is a dominant phase, we have $\alpha_{i}>0$ for all $i\in[q]$ (from the last part of Lemma~\ref{new:zako1}). Also, for the maximizer $\x^*$ of $g_1(\x)$, it holds that $x^*_{ij}>0$ for all $(i,j)\in P_1$ (see \eqref{eq:optimalxs}). Pick $\delta$ sufficiently small such that:
\begin{equation*}
\norm{\X-\X^*}_2\leq \delta\text{ implies } x_{ij}>0\text{ for all }(i,j)\in P_1.
\end{equation*}
Since $g_1(\x)$ has the unique global maximum $\X^*$, standard compactness arguments imply that there exists $\epsilon(\delta)>0$ such that $\norm{\X-\X^*}_2\nobreak\geq \delta$ implies $g(\X^*)-g(\X)\geq\epsilon$. It follows that the contribution of terms with $\norm{\x-\X^*}_2\geq \delta$ to $\E_{\G}[Z^{\alphab}_G]$ is exponentially small and may be ignored. Hence we may restrict our attention to $\X$ satisfying $\norm{\X-\X^*}_2< \delta$. 

We now approximate the terms in \eqref{eq:firstmoment} with $\norm{\X-\X^*}_2< \delta$ using Stirling's approximation. The only difference with the asymptotics in Section~\ref{sec:defer-prelim} is that now the relative error of the approximation will be asymptotically $O\big(n^{-1}\big)$.  In particular, we will use the following asymptotics for factorials (which are a refinement of \eqref{eq:simpleasymp}). For any constant $c>0$, it holds that
\begin{equation}\label{eq:simpleasymp22}
\begin{aligned}
(cn)!&=\Big(1+O\big(n^{-1}\big)\Big) \sqrt{2\pi c n}\exp(cn\ln n+cn\ln c-cn),\\
(cn-1)!!&=\Big(1+O\big(n^{-1}\big)\Big) \sqrt{2}\exp\big(\frac{cn}{2}\ln n+\frac{cn}{2}\ln c-\frac{cn}{2}\big).
\end{aligned}
\end{equation}
Using these asymptotics to expand the terms in \eqref{eq:firstmoment} (together with $\sum_{j}x_{ij}=\alpha_i$ and $\sum_{i}\alpha_i=1$), we  obtain

\begin{multline}\label{eq:mgtyuib123}
\frac{(2\pi n)^{(q-1)/2}\E_{\G}[Z^{\alphab}_G]}{e^{n\Upsilon_1(\alphab,\X^{*})}}=
\Big(1+O\big(n^{-1}\big)\Big)2^{(q-1)/2}\\\times\sum_\X\big(\sqrt{2\pi\Delta n}\big)^{q}
\Big(\prod_{(i,j) \in P_1}  \frac{1}{\sqrt{2\pi\Delta n x_{ij}}}     \Big)
e^{n   \Delta \big(    g_1(\X)-g_1(\X^*)     \big)}.
\end{multline}
In the l.h.s., the factor $(2\pi n)^{(q-1)/2}$ comes from the expansion of $\binom{n}{\alpha_{1}n,\hdots,\alpha_{q}n}$; in the r.h.s., the factor $\big(\sqrt{2\pi\Delta n}\big)^{q}$ comes from the expansion of $(\Delta \alpha_i n)!$ for $i\in[q]$, the factor $1/\sqrt{2\pi\Delta n x_{ij}}$ comes from the expansion of $(\Delta x_{ij} n)!$ for $(i,j)\in P_1$ and the factor $2^{(q-1)/2}$ comes from the expansion of $(\Delta n-1)!!$ and $(\Delta x_{ii}n-1)!!$ for $i\in [q]$. 

We are now ready to compute
\[L:=\lim_{n\rightarrow\infty}
\frac{(2\pi n)^{(q-1)/2}\E_{\G}[Z^{\alphab}_G]}{e^{n\Upsilon_1(\alphab,\X^{*})}}.
\]
Since $\x^*$ is a critical point of $g_1(\x)$, for all sufficiently small $\delta>0$, we have the expansion
\[g_1(\X)-g_1(\X^*)=\frac{1}{2}(\x-\x^*)^{\T}\mathbf{H}(\x-\x^*)+O(\delta^{3}),\]
where $\mathbf{H}=\mathbf{H}^{f}_{1,\x}$ is the Hessian of the full-dimensional representation of $g_1(\x)$ evaluated at $\x=\x^{*}$ (the matrix $\mathbf{H}$ has dimension $(|P_1|-q)\times(|P_1|-q)$). Note that $g_1$ is a strictly concave function, so $\mathbf{H}$ is negative definite. Using standard techniques of rewriting sums as integrals and the dominated convergence theorem (see \cite[Section 9.4]{JLR}), we obtain
\begin{multline}\label{eq:valueofintegral}
\lim_{n\rightarrow \infty}\sum_\X\big(\sqrt{2\pi\Delta n}\big)^{q}
\Big(\prod_{(i,j) \in P_1}  \frac{1}{\sqrt{2\pi\Delta n x_{ij}}}     \Big)
e^{n   \Delta \big(    g_1(\X)-g_1(\X^*)     \big)}\\=  \frac{\big( \sqrt{\Delta}     \big)^{|P_1|-q}}{2^{q-1}\big(\prod_{(i,j)\in P_1}x^*_{ij}\big)^{1/2} }
\bigg(\frac{1}{(\sqrt{2\pi})^{|P_1|-q}}\int^{\infty}_{-\infty}   \cdots \int^{\infty}_{-\infty}
\exp\Big(\frac{\Delta}{2}   \X^{\T}\,   \mathbf{H} \X \Big)
d\X\bigg).
\end{multline}
Note, in the r.h.s. of \eqref{eq:valueofintegral}, the factor $1/2^{q-1}$ comes from the $q-1$ constraints $(\mathrm{mod}\, 2)$  restricting the range of $\x$ (discussed just after the expression \eqref{eq:firstmoment} for $\E_{\G}[Z^{\alphab}_G]$ in the beginning of the proof).

Plugging in \eqref{eq:valueofintegral} the value of the Gaussian integral and substituting back in \eqref{eq:mgtyuib123} yields that
\begin{equation*}
\begin{aligned}
L&=\Big(   2^{q-1}  \prod_{(i,j)\in P_1}x^*_{ij}       \Big)^{-1/2}
\big(   \sqrt{\Delta}     \big)^{|P_1|-q}
\mathrm{Det}(-\Delta \H)^{-1/2}\\
&=\Big(   2^{q-1}  \prod_{(i,j)\in P_1}x^*_{ij}       \Big)^{-1/2}\mathrm{Det}(- \H)^{-1/2},
\end{aligned}
\end{equation*}
as wanted. This concludes the proof of Lemma~\ref{lem:firasympdet}.
\end{proof}

\begin{proof}[Proof of Lemma~\ref{lem:secasympdet}]
We proceed analogously to the proof of Lemma~\ref{lem:firasympdet}. In particular, we will assume a full-dimensional representation of $(\gammab,\Y)$ (see the beginning of Section~\ref{sec:momenttwo} for details on the choice of the representation). Note that $\gammab,\Y$ have $(q-1)^2$ and $|P_2|-q^2$ variables, respectively.

Recall that 
\begin{multline}\tag{\ref{eq:secondmoment}}
\E_{\G}[(Z^{\alphab}_G)^2] = \sum_{\gammab}\binom{n}{\gamma_{11}n,\hdots,\gamma_{qq}n}\sum_{\y}\left\{\prod_{i,k}\binom{\Delta\gamma_{ik}n}{\Delta y_{ik11}n,\hdots,\Delta y_{ikqq}n}\right.\\
\left.\times\frac{\big[\prod_{(i, k)\neq (j,l)}(\Delta y_{ikjl}n)!\big]^{1/2}\prod_{i,k}(\Delta y_{ikik}n-1)!!}{(\Delta n-1)!!}\prod_{i,j,k,l}\big(B_{ij}B_{kl}\big)^{\Delta y_{ikjl}n/2}\right\},
\end{multline}
 We have that $\gamma^*_{ik}>0$ for all $i,k\in [q]$ and $y^*_{ikjl}>0$ for $(i,k,j,l)\in P_2$ (see the paragraph following Remark~\ref{rem:fullrowrank}). Pick $\delta$ sufficiently small such that:
\begin{equation*}
\norm{(\gammab,\Y)-(\gammab^*,\Y^*)}_2\leq \delta\text{ implies }\gamma_{ik}>0\text{ for }i,k\in[q] \text{ \& }y_{ikjl}>0\text{ for }(i,k,j,l)\in P_2.
\end{equation*}
Using that $\Upsilon_2(\gammab,\Y)$ is maximized (uniquely) at $(\gammab,\Y)=(\gammab^*,\Y^*)$, the same line of arguments as in the proof of Lemma~\ref{lem:firasympdet} yield
\begin{multline*}
\frac{\E_{\G}[(Z^{\alphab}_G)^2]}{e^{n\Upsilon_2(\gammab^*,\Y^{*})}}
=\Big(1+O\big(n^{-1}\big)\Big)
\sum_{\gammab,\Y}\left\{
(\sqrt{2\pi n})2^{(q^2-1)/2} \Delta^{q^2/2}\vphantom{.\Big(   \prod_{(i,k,j,l) \in P_2}  \frac{1}{\sqrt{2\pi \Delta y_{ikjl}n}}     \Big)}\right.\\\left.\Big(   \prod_{(i,k,j,l) \in P_2}  \frac{1}{\sqrt{2\pi \Delta y_{ikjl}n}}     \Big)
e^{n    \big(    \Upsilon_2(\gammab,\Y)-\Upsilon_2(\gammab^*,\Y^*)     \big)}\right\}.
\end{multline*}
In the r.h.s., the factor $\sqrt{2\pi n}$ comes from the expansion of $n!$, the factor $2^{(q^2-1)/2}$ comes from the expansion of $(\Delta n-1)!!$ and $(\Delta y_{ikik}n-1)!!$ for $i,k\in [q]$, the factor $\Delta^{q^2/2}$ from the expansion of $(\Delta \gamma_{ik}n)!$ for $i,k\in [q]$ and the factor $1/\sqrt{2\pi \Delta y_{ikjl}n}$ from the expansion of $(\Delta y_{ikjl}n)!$ for $(i,k,j,l)\in P_2$.  

We now compute
\[L=\lim\limits_{n\rightarrow\infty}
\frac{(2\pi n)^{(q-1)}\E_{\G}[(Z^{\alphab}_G)^2]}{e^{n\Upsilon_2(\gammab^*,\Y^{*})}}.
\]
Since $\gammab^*,\Y^*$ is a critical point of $\Upsilon_2(\gammab,\Y)$, for all sufficiently small $\delta>0$, we have the expansion
\[\Upsilon_2(\gammab,\Y)-\Upsilon_2(\gammab^*,\Y^*)=\frac{\Delta}{2}\big([\gammab,\Y]-[\gammab^*,\Y^*]\big)^{\T}\mathbf{H}\big([\gammab,\Y]-[\gammab^*,\Y^*]\big)+O(\delta^{3}),\]
where $\H=\H_2^{f}$ is the Hessian matrix of $\Upsilon_2$ evaluated at $(\gammab^*,\Y^*)$ scaled by 1/$\Delta$. Now, we may proceed analogously to the proof of Lemma~\ref{lem:firasympdet} and obtain
\begin{align*}
L&=\Big(    2^{q^2-1} \prod_{(i,k,j,l)\in P_2}y^*_{ikjl}       \Big)^{-1/2}
(\sqrt{\Delta})^{|P_2|-q^2}\\
&\qquad \qquad \qquad \quad
\Big(   \frac{1}{\sqrt{2\pi}}     \Big)^{|P_2|-2q+1}\int^{\infty}_{-\infty}   \cdots \int^{\infty}_{-\infty}
\exp\Big(   \frac{\Delta}{2}   [\gammab,\Y]^{\T}  \mathbf{H}[\gammab,\Y] \Big)
d\Y
d\gammab,\\
&=\Big(    2^{q^2-1} \prod_{(i,k,j,l)\in P_2}y^*_{ikjl}       \Big)^{-1/2}\Delta^{-(q-1)^{2}/2}\,\mathrm{Det}(-\H)^{-1/2},
\end{align*}
where in the last equality we substituted the value of the Gaussian integral. 

This concludes the proof of Lemma~\ref{lem:secasympdet}.
\end{proof}

\end{document}